%% file: arxiv.tex
\documentclass[preprint,12pt]{imsart}

\RequirePackage{amsthm,amsmath,amsfonts,amssymb}
\RequirePackage[authoryear]{natbib}
\RequirePackage[colorlinks,citecolor=blue,urlcolor=blue]{hyperref}
\RequirePackage{graphicx}
\RequirePackage{fullpage}
\RequirePackage{stmaryrd}

\RequirePackage{todonotes}

\RequirePackage{grffile}
\RequirePackage{booktabs,rotating}
\RequirePackage{bm}

\startlocaldefs
\numberwithin{equation}{section}
\theoremstyle{plain}
\newtheorem{prop}{Proposition}[section]

\newtheorem{thm}[prop]{Theorem}
\newtheorem{cor}[prop]{Corollary}
\newtheorem{lem}[prop]{Lemma}
\newtheorem{cond}[prop]{Condition}

\theoremstyle{remark}
\newtheorem{remark}[prop]{Remark}

\newcommand{\eps}{\varepsilon}
\newcommand{\N}{\mathbb{N}}
\newcommand{\R}{\mathbb{R}}

\newcommand{\Fc}{\mathcal{F}}
\newcommand{\Mc}{\mathcal{M}}
\newcommand{\Lc}{\mathcal{L}}
\newcommand{\Xc}{\mathcal{X}}
\newcommand{\Pc}{\mathcal{P}}
\newcommand{\Sc}{\mathcal{S}}
\newcommand{\sPc}{{\scriptscriptstyle{\mathcal{P}}}}

\renewcommand{\P}{\mathbb{P}}
\def\disp{\displaystyle}
\newcommand{\Ex}{\mathbb{E}}

\newcommand{\Cov}{\mathrm{Cov}}

\newcommand{\1}{\mathbf{1}}
\newcommand{\ip}[1]{\lfloor #1 \rfloor}
\newcommand{\p}{\overset{\sss\P}{\to}}
\newcommand{\as}{\overset{\sss a.s.}{\to}}
\renewcommand{\d}{\overset{d}{=}}
\newcommand{\pobs}[1]{\hat{\bm #1}}
\renewcommand{\d}{\overset{d}{=}}
\newcommand{\floor}[1]{\lfloor #1 \rfloor}
\newcommand{\sss}[1]{\scriptscriptstyle{#1}}
\newcommand{\vep}{\varepsilon}
\newcommand{\llb}{\llbracket}
\newcommand{\rrb}{\rrbracket}

\endlocaldefs

\begin{document}

\begin{frontmatter}
\title{Multi-purpose open-end monitoring procedures for multivariate observations based on the empirical distribution function}
\runtitle{}

\begin{aug}
  \author[M]{\fnms{Mark} \snm{Holmes}\ead[label=eM]{holmes.m@unimelb.edu.au}},
  \author[I]{\fnms{Ivan} \snm{Kojadinovic}\ead[label=eI]{ivan.kojadinovic@univ-pau.fr}}
  \and
  \author[M,I]{\fnms{Alex} \snm{Verhoijsen}\ead[label=eA]{alexverhoijsen@gmail.com}}
  \address[M]{School of Mathematics \& Statistics \\ The University of Melbourne \\ Parkville, VIC 3010, Australia}
  \address[I]{CNRS / Universit\'e de Pau et des Pays de l'Adour / E2S UPPA \\ Laboratoire de math\'ematiques et applications IPRA, UMR 5142 \\ B.P. 1155, 64013 Pau Cedex, France}
\end{aug}

\begin{abstract}
We propose nonparametric open-end sequential testing procedures that can detect all types of changes in the contemporary distribution function of possibly multivariate observations. Their asymptotic properties are theoretically investigated under stationarity and under alternatives to stationarity. Monte Carlo experiments reveal their good finite-sample behavior in the case of continuous univariate, bivariate and trivariate observations. A short data example concludes the work.  
\end{abstract}



\end{frontmatter}


\section{Introduction}

From an historical perspective, monitoring is often associated with \emph{control charts} also known as \emph{Shewart charts}. Such graphical tools central to \emph{statistical process control} \citep[see, e.g.,][for an overview]{Lai01,Mon07} are usually calibrated in terms of the so-called \emph{average run length} (ARL) controlling how many monitoring steps are necessary on average before the data generating process is declared \emph{out of control}. The fact that this conclusion (that is, that the probabilistic properties of the monitored observations have changed) is reached with probability one could be regarded as a drawback of this type of procedure, in particular if \emph{false alarms} are very costly. To remedy this situation, \cite{ChuStiWhi96} have proposed to treat the issue of monitoring from the point of view of statistical testing. The main advantage is that, when observations arise from a stationary time series, monitoring procedures \emph{à la} \cite{ChuStiWhi96} will lead to the conclusion that a change has occurred in the data generating process only with a small probability $\alpha$ controlled by the user. For a recent nicely written literature review comparing monitoring as carried out in statistical process control to approaches based on statistical tests \emph{à la} \cite{ChuStiWhi96}, we refer the reader to the introduction of \cite{GosStoHeiDet22}.

In addition to being statistical tests, the monitoring procedures investigated in this work are nonparametric and deal with $d$-dimensional observations, $d \geq 1$. In that respect, as we continue, we use the superscript $\strut^{\sss[\ell]}$ to denote the $\ell$th coordinate of a vector (for instance, $\bm x = (x^{\sss[1]},\dots, x^{\sss[d]}) \in \R^d)$. As is customary in the sequential testing literature, we assume that we have at hand $m \geq 1$ observations from the initial data generating process. Monitoring starts immediately thereafter. To be more precise, we assume that we have at our disposal a stretch $\bm X_i = (X_i^{\sss[1]},\dots, X_i^{\sss[d]})$, $i \in \{1,\dots, m\}$, from a $d$-dimensional stationary time series with unknown contemporary distribution function (d.f.) $F$ given by $F(\bm x) = \P(X_1^{\sss[1]} \leq x^{\sss[1]}, \dots, X_1^{\sss[d]} \leq x^{\sss[d]}) = \P(\bm X_1 \leq \bm x)$, $\bm x \in \R^d$. These available observations will be referred to as the \emph{learning sample} as we continue. Once the monitoring starts, new observations $\bm X_{m+1}, \bm X_{m+2},\dots$ arrive sequentially and the aim is to issue an alarm as soon as possible if there is evidence that the contemporary distribution of the most recent observations is no longer equal to $F$.

Most approaches in the literature are of a \emph{closed-end} nature: the monitoring eventually stops if stationarity is not rejected after the arrival of a final observation $\bm X_n$, $n > m$. Our focus in this work is on the more difficult scenario in which the monitoring can in principle continue indefinitely: this is called the \emph{open-end} setting. From a practical perspective, it can be argued that the fact that the \emph{monitoring horizon} $n$ does not need to be specified is a great advantage of open-end procedures. The price to pay for open-endness is however a significantly more complicated theoretical setting. Indeed, as discussed for instance in Remark 2.2 of \cite{GosKleDet21}, while the asymptotics of closed-end procedures can usually be derived using functional central limit theorems, such results are insufficient in the open-end case and need to be either combined with H\'ay\'ek-R\'eyni type inequalities \citep[see, e.g.,][and the references therein]{KirWeb18} or replaced by approximations of the form of forthcoming Condition~\ref{cond:H0} \citep[see][]{AueHor04,AueHorHusKok06}.

The null hypothesis of the sequential testing procedures studied in this work is
\begin{equation}
  \label{eq:H0}
  H_0: \,  \bm X_1,\dots, \bm X_m, \bm X_{m+1}, \bm X_{m+2},\dots, \text{ is a stretch from a stationary time series}.
\end{equation}

When $d = 1$, starting from the work of \cite{GosKleDet21}, \cite{HolKoj21} have recently introduced a detection procedure that is particularly sensitive to changes in the mean. Because it uses the \emph{retrospective cumulative sum (CUSUM) statistic} as \emph{detector}, it turns out to be more powerful then existing procedures as long as changes do not occur at the very beginning of the monitoring.  As noted in Section 6 of the latter reference, this approach can be adapted to obtain alternative procedures that are particularly sensitive for instance to changes in the variance or some other moments. The goal of this work is to generalize the method of \cite{HolKoj21} in order to obtain open-end monitoring procedures that can be sensitive simultaneously to all types of changes in the d.f.\ $F$. Although such procedures already exist in a closed-end setting \citep[see, e.g.,][]{KojVer21}, to the best of our knowledge, they are unavailable in the open-end setting.

This paper is organized as follows. In the second section, starting from the work of \cite{GosKleDet21} and \cite{HolKoj21}, we propose a detector that can be sensitive to all types of changes in the contemporary d.f.\ of multivariate observations. We additionally introduce a suitable threshold function and study the asymptotics of the resulting monitoring procedure under $H_0$ in~\eqref{eq:H0} and under sequences of alternatives to $H_0$. In the third section, we focus on the case of continuous observations and provide additional asymptotic results under the null. In the fourth section,  using \emph{asymptotic regression models}, we address the estimation of high quantiles of the distributions appearing in the asymptotic results under $H_0$ which are necessary in practice to carry out the sequential tests. The fifth section summarizes the results of numerous Monte Carlo experiments for $d \in \{1,2,3\}$ whose aim is to study the finite-sample behavior of the monitoring procedures under $H_0$ and under alternatives to $H_0$. A data example and concluding remarks are gathered in the last section.

Unless mentioned otherwise, all convergences are as $m \to \infty$. Also, as there are many discrete intervals appearing in this work, we will conveniently use the notation $\llb j,k \rrb$, $\llb j,k \llb$, $\rrb j,k \rrb$, and $\rrb j,k \llb$ for the sets of integers $\{j,\dots,k\}$, $\{j,\dots,k-1\}$, $\{j-1,\dots,k\}$, and $\{j-1,\dots,k-1\}$, respectively. Note that all mathematical proofs are gathered in a series of appendices and a non-optimized implementation of the monitoring procedures studied in this work is available in the package \texttt{npcp} \citep{npcp} for the \textsf{R} statistical environment \citep{Rsystem}.

\section{A first detector, the threshold function and related asymptotics}
\label{sec:det:thres:asym}

\subsection{A first detector and the threshold function}

One of the two main ingredients of a monitoring procedure \emph{à la} \cite{ChuStiWhi96}  is a statistic, also called a \emph{detector}, that is potentially computed after the arrival of every new observation $\bm X_k$, $k > m$. This statistic is typically positive and quantifies some type of departure from stationarity. Once computed, it is compared to a positive threshold, possibly also depending on $k$. If greater, evidence against $H_0$ in~\eqref{eq:H0} is deemed significant and the monitoring stops. Otherwise, a new observation is collected.

For sensitivity to changes in the mean of univariate ($d=1$) observations, \cite{HolKoj21} used among others the detector
\begin{equation}
  \label{eq:det:mean}
  R_m(k) =  \max_{j \in \llb m,k \llb} \frac{j (k-j)}{m^{\frac32}}  | \bar X_{1:j}^{\sss[1]} - \bar X_{j+1:k}^{\sss[1]} |, \qquad k \geq m+1,
\end{equation}
where $\bar X_{j:k}^{\sss[1]} =\frac{1}{k-j+1}\sum_{i=j}^{k} X_j^{\sss[1]}$, $1 \leq j \leq k$. Some thought reveals that, for any fixed $k \geq m+1$, $R_m(k)$ is akin to the so-called \emph{retrospective CUSUM statistic} frequently used in offline change-point detection tests \citep[see, e.g.,][]{CsoHor97, AueHor13}.

When $d \geq 1$, to be sensitive to changes in the d.f.\ at a fixed point $\bm x =(x^{\sss[1]},\dots, x^{\sss[d]})\in \R^d$, a straightforward adaptation of the previous approach would be to compute~\eqref{eq:det:mean} from the stretch of univariate observations $\1(\bm X_1 \leq \bm x), \dots, \1(\bm X_m \leq \bm x), \1(\bm X_{m+1} \leq \bm x), \dots, \1(\bm X_k \leq \bm x)$, where inequalities between vectors are to be understood componentwise. The detector $R_m$ can then be equivalently expressed as
\begin{equation}
  \label{eq:Em}
  E_m^{\bm x}(k) =  \max_{j \in \llb m,k \llb} \frac{j (k-j)}{m^{\frac32}}  | F_{1:j}(\bm x) - F_{j+1:k}(\bm x) |, \qquad k \geq m+1,
\end{equation}
where, for any integers $j,k \geq 1$,
\begin{equation}
  \label{eq:Fjk}
  F_{j:k}(\bm x) =
  \left\{
    \begin{array}{ll}
      \disp \frac{1}{k-j+1}\sum_{i=j}^k \1(\bm X_i \leq \bm x), \qquad & \text{if } j \leq k, \\
      0, \qquad &\text{otherwise,}
    \end{array}
  \right.
\end{equation}
is the empirical d.f.\ of $\bm X_j, \dots, \bm X_k$ evaluated at $\bm x$. Our aim is to extend the previous approach using $p \geq 1$ points $\bm x_1,\dots, \bm x_p$ in $\R^d$, where the integer $p$ and the points $\bm x_1,\dots, \bm x_p$ are chosen by the user. Let $\Pc = (\bm x_1,\dots, \bm x_p)$ and, for any $i \in \N$, let $\bm Y_i^\sPc=\big( \1(\bm X_i \leq \bm x_1), \dots, \1(\bm X_i \leq \bm x_p) \big)$, which is a $p$-dimensional random vector.  Combining the approach of \cite{GosKleDet21} with the one of \cite{HolKoj21}, the first detector considered in this work is defined by
\begin{equation}
  \label{eq:Dm}
  D_m^\sPc(k) = \max_{j \in \llb m,k \llb} \frac{j (k-j)}{m^{\frac32}}  \| \bm F_{1:j}^\sPc - \bm F_{j+1:k}^\sPc \|_{(\Sigma_m^\sPc)^{-1}}, \qquad k \geq m+1,
\end{equation}
where, for any integers $j,k \geq 1$, $\bm F_{j:k}^\sPc = \big( F_{j:k}(\bm x_1),\dots,F_{j:k}(\bm x_p) \big) \in \R^p$, $\Sigma_m^\sPc$ is an estimator (based on $\bm Y_1^\sPc,\dots,\bm Y_m^\sPc$) of the long-run $p \times p$ covariance matrix
\begin{equation}
  \label{eq:sigma}
\Sigma^\sPc = \Cov(\bm Y_1^\sPc, \bm Y_1^\sPc) + \sum_{i = 2}^\infty \{ \Cov(\bm Y_1^\sPc, \bm Y_i^\sPc) + \Cov(\bm Y_i^\sPc, \bm Y_1^\sPc) \}
\end{equation}
of the $p$-dimensional time series $\big( \bm Y_i^\sPc \big)_{i \in \N}$ and, for any $\bm y \in \R^p$, $\| \bm y \|_M = \sqrt{(\bm y^\top M \bm y)/p}$ denotes a weighted norm of $\bm y$ induced by a $p \times p$ positive-definite matrix $M$ and the integer~$p$.

Note that~\eqref{eq:Dm} can be equivalently rewritten as
\begin{equation}
  \label{eq:Dm:Yi}
  D_m^\sPc(k) = \max_{j \in \llb m,k \llb} \frac{j (k-j)}{m^{\frac32}}  \| \bm {\bar Y}^\sPc_{1:j} - \bm {\bar Y}^\sPc_{j+1:k} \|_{(\Sigma_m^\sPc)^{-1}}, \qquad k \geq m+1,
\end{equation}
where, for any integers $1 \leq j \leq k$, $ \bm {\bar Y}^\sPc_{j:k} = \frac{1}{k-j+1}\sum_{i=j}^k \bm Y_i^\sPc$.

\begin{remark}
  \label{rem:norm}    
  The role of the matrix $(\Sigma_m^\sPc)^{-1}$ when using the Mahalanobis-like norm $\| \cdot \|_{(\Sigma_m^\sPc)^{-1}}$ in~\eqref{eq:Dm} is, roughly speaking, to standardize and decorrelate under the null vectors of the form $\bm F_{1:j}^\sPc - \bm F_{j+1:k}^\sPc$ before computing their $L_2$ norm (scaled by $1/\sqrt{p}$). As shall become clearer from Theorem~\ref{thm:H0} below, a consequence of that step is that a key limiting null distribution playing a central role in the testing procedure will not depend on the characteristics of the underlying time series $( \bm X_i )_{i \in \N}$ but only on the number of points $p$ chosen by the user. The latter desirable property from a practical perspective is also the reason why we did not consider, instead of $D_m^\sPc$ in~\eqref{eq:Dm}, alternative detectors that evaluate differences of empirical d.f.s at all the points in $\R^d$. One natural such alternative detector is
  \begin{equation}
  \label{eq:Dm:sup}
  D_m^{\sup}(k) = \max_{j \in \llb m,k \llb} \frac{j (k-j)}{m^{\frac32}}  \sup_{\bm x \in \R^d} | F_{1:j}(\bm x) - F_{j+1:k}(\bm x) |, \qquad k \geq m+1.
\end{equation}
Such a detector was actually considered in a closed-end setting by \cite{KojVer21} and required in practice the use of bootstrapping when monitoring serially dependent observations. The major practical obstacle related to its use in an open-end setting will be discussed in Remark~\ref{rem:Dm:sup:asym}.
\end{remark}

The second key ingredient of a monitoring procedure is a \emph{threshold function}. In the considered open-end setting, given a significance level $\alpha \in (0,\frac12)$, the aim is to define a deterministic function $w:[1,\infty) \to (0,\infty)$ such that (ideally) under $H_0$ in~\eqref{eq:H0},
\begin{equation}
  \label{eq:typeIerror}
  \P\Big(D_m^\sPc(k) \leq w(k/m) \text{ for all } k>m \Big) =  \P \Bigg( \sup_{k > m} \frac{D_m^\sPc(k)}{w(k/m)} \leq 1 \Bigg) = 1 - \alpha,
\end{equation}
where the supremum is over integers $k>m$. Because the detector $D_m^\sPc$ in~\eqref{eq:Dm} or in~\eqref{eq:Dm:Yi} can be regarded as a multivariate generalization of the detector $R_m$ in~\eqref{eq:det:mean}, one can use the same reasoning as 
in Section~2 of \cite{HolKoj21} to suggest that a meaningful threshold function in the considered case is
\begin{equation}
  \label{eq:w}
  w(t) = q_{p,\eta}^{\sss{(1 - \alpha)}} t^{\frac32+\eta}, \qquad t \in [1,\infty),
\end{equation}
where $\eta$ is a positive real parameter and $q_{p,\eta}^{\sss{(1 - \alpha)}} > 0$ is the $(1-\alpha)$-quantile of $\Lc_{p,\eta}$, the weak limit of $\sup_{k > m} (m/k)^{\frac32+\eta} D_m^\sPc(k)$ under $H_0$, assuming that $\Lc_{p,\eta}$ is continuous. In that case, under $H_0$, by the Portmanteau theorem,
\begin{align}
  \nonumber
  \lim_{m\to\infty} \P \Bigg( \sup_{k>m} \frac{D_m^\sPc(k)}{w(k/m)} \leq 1 \Bigg) 
  &=   \lim_{m\to\infty}\P \Bigg( \sup_{k>m} (m/k)^{\frac32+\eta} D_m^\sPc(k) \leq q_{p,\eta}^{\sss{(1 - \alpha)}} \Bigg) \\
  &= \P(\Lc_{p,\eta} \leq q_{p,\eta}^{\sss{(1 - \alpha)}}) = 1 - \alpha,   \label{eq:typeIerror:asym}
\end{align}
which can be regarded as an ``asymptotic version'' of~\eqref{eq:typeIerror}. As far as $\eta$ is concerned, as a consequence of Proposition~\ref{prop:eta:0} below, it needs to be chosen strictly positive so that $\Lc_{p,\eta}$ is almost surely finite. In practice, we follow the recommendation made in \cite{HolKoj21} and set $\eta$ to 0.001.

\begin{remark}
Proceeding for instance along the lines of \cite{HorHusKokSte04}, \cite{Fre15}, \cite{KirWeb18}, \cite{GosKleDet21} or \cite{HolKoj21}, instead of $w$ in~\eqref{eq:w}, one could alternatively consider as a threshold function $\tilde w$ defined by $\tilde w(t) = \bar w_\gamma(t) w(t)$, $t \in [1,\infty)$, where 
\begin{equation*}
  \bar w_\gamma(t) = \max \left\{ \left(\frac{t-1}{t}\right)^\gamma, \epsilon \right\},  \qquad t \in [1,\infty),
\end{equation*}
with $\gamma \geq 0$ a real parameter and $\epsilon > 0$ a technical constant that can be taken very small in practice. The multiplication of a candidate threshold function by $\bar w_\gamma$ was initially considered in \cite{HorHusKokSte04} and \cite{AueHor04} for the so-called \emph{ordinary CUSUM} detector in order to study, under suitable alternatives, the limiting distribution of the detection delay (the time delay after which the detector exceeds the threshold function). From a practical perspective, as discussed in \cite{HolKoj21}, an appropriate choice of $\gamma \geq 0$ may improve the finite-sample performance of the sequential test at the beginning of the monitoring. The multiplication of a candidate threshold function by $\bar w_\gamma$ does not however affect the asymptotics of the underlying monitoring procedure.
\end{remark}

\subsection{Asymptotics under the null}
\label{sec:asym:H0}

One of the first assumptions required to be able to study the asymptotics (as $m \to \infty$, of the monitoring procedure based on $D_m^\sPc$ in~\eqref{eq:Dm} and $w$ in~\eqref{eq:w})  concerns the long-run covariance matrix $\Sigma^\sPc$ in~\eqref{eq:sigma} of the $p$-dimensional time series $\big( \bm Y_i^\sPc \big)_{i \in \N}$, under the null. As we shall see later in this section, it will be necessary to consider both its inverse $(\Sigma^\sPc)^{-1}$ and its square root $(\Sigma^\sPc)^{\frac12}$. The following assumption guarantees that these two matrices exist (and are unique).

\begin{cond}[On the long-run covariance matrix $\Sigma^\sPc$]
  \label{cond:sigma}
Under $H_0$ in~\eqref{eq:H0}, the long-run covariance matrix $\Sigma^\sPc$ in~\eqref{eq:sigma} of the $p$-dimensional (stationary) time series $\big( \bm Y_i^\sPc \big)_{i \in \N}$ exists and is positive-definite.
\end{cond}

\begin{remark}
  \label{remark:points}
  Under $H_0$ and when the time series $( \bm X_i )_{i \in \N}$ consists of independent observations, the $p \times p$ elements of $\Sigma^\sPc$ in~\eqref{eq:sigma} are simply
$$
\Cov\{\1(\bm X_1 \leq \bm x_i), \1(\bm X_1 \leq \bm x_j) \} = F \big( \min(\bm x_i, \bm x_j) \big) -  F(\bm x_i) F(\bm x_j), \qquad i,j \in \llb 1, p \rrb,
$$
where $\min$ denotes the element-wise minimum operator. Hence, in the case of serially independent observations, by definition of positive-definiteness, Condition~\ref{cond:sigma} will hold if the points $\bm x_1,\dots,\bm x_p$ appearing in $\Pc$ are chosen such that any linear combination of the $\1(\bm X_1 \leq \bm x_i)$, $i \in \llb 1, p \rrb$, has a strictly positive variance. A necessary condition for this is that $\bm x_1,\dots,\bm x_p$ all belong to the support of $\bm X_1$ and are all distinct. Since the law of $\bm X_1$ is unknown, the user could in practice rely on the learning sample $\bm X_1,\dots,\bm X_m$ to choose $\bm x_1,\dots,\bm x_p$. If the learning sample seems to be a stretch from a discrete time series, a natural possibility consists of choosing $\bm x_1,\dots,\bm x_p$ from a subset of frequently occurring observations. The choice of $\Pc$ when the observations in the learning sample seem to arise from a continuous time series will be discussed in Section~\ref{sec:cont}.
\end{remark}

As shall become clearer in the forthcoming paragraphs, studying the asymptotics under the null (of the monitoring procedure as $m \to \infty$) actually amounts to establishing the weak limit of $\sup_{k>m} (m/k)^{\frac32+\eta} D_m^\sPc(k)$ under $H_0$ in~\eqref{eq:H0}. The following assumption on the time series $\big( \bm Y_i^\sPc \big)_{i \in \N}$  (a type of ``strong approximation'' condition) is typical of the kinds of assumptions in the sequential change-point literature; see, e.g., Assumption 2.3 in \cite{GosKleDet21}, Condition 3.1 in \cite{HolKoj21} and the corresponding discussions in these references. Let $\|\cdot\|_2$ denote the Euclidean norm.

\begin{cond}[Approximation]
  \label{cond:H0}
  There exists a probability space $(\Omega, \Fc,\P)$ on which:
  \begin{itemize}
  \item $(\bm Y_i^\sPc)_{i \in \N}$ is a $p$-dimensional stationary time series satisfying Condition~\ref{cond:sigma}, \\
  \item for each $m \in \N$,  $\bm W_{1,m}$ and $\bm W_{2,m}$ are two independent $p$-dimensional standard Brownian motions,
  \end{itemize}
  such that, for some $0 < \xi < \frac12$,
\begin{equation}
  \label{eq:sup:cond}
\sup_{k>m} \frac{1}{(k-m)^\xi} \left\| \sum_{i=m+1}^{k} \{\bm Y_i^\sPc - \Ex(\bm  Y_1^\sPc) \}  - (\Sigma^\sPc)^{\frac12} \bm W_{1,m}(k-m) \right\|_2 = O_\P(1)
\end{equation}
and
\begin{equation}
  \label{eq:cond}
\frac{1}{m^\xi} \left\| \sum_{i=1}^{m} \{\bm Y_i^\sPc - \Ex( \bm Y_1^\sPc) \}  - (\Sigma^\sPc)^{\frac12} \bm W_{2,m}(m) \right\|_2 = O_\P(1).
\end{equation}
\end{cond}

We use the notation `$\leadsto$' to  denote convergence in distribution (weak convergence) and $I_p$ to denote the $p \times p$ identity matrix. The following result, proven in Appendix~\ref{proof:thm:H0}, can be regarded as a multivariate extension of Theorem~3.3 of \cite{HolKoj21}.

\begin{thm}
\label{thm:H0}
Fix $\eta>0$.  Under Condition~\ref{cond:H0}, if $\Sigma_m^\sPc \p \Sigma^\sPc$ then
\begin{align*}
\sup_{k>m} (m/k)^{\frac32+\eta}  D_m^\sPc(k) \leadsto \Lc_{p,\eta} = \sup_{1 \leq s \leq t < \infty} t^{-\frac32-\eta} \| t \bm W(s) - s \bm W(t) \|_{I_p},
\end{align*}
where $D_m^\sPc$ is defined in~\eqref{eq:Dm} and $\bm W$ is a $p$-dimensional standard Brownian motion.  In addition, the limiting random variable $\Lc_{p,\eta}$ is almost surely finite.
\end{thm}

Note that in Theorem \ref{thm:H0} the supremum on the left is over integers $k$ while the supremum on the right is over real numbers $s,t$.

\begin{remark}
  \label{rem:L}
  It is important to note that the limiting random variable $\Lc_{p,\eta}$ depends neither on the characteristics of the underlying time series $( \bm X_i )_{i \in \N}$ (such as its dimension $d$, its serial dependence properties or the unknown d.f.\ $F$), nor on the user-chosen points $\Pc = (\bm x_1,\dots, \bm x_p)$ involved in the definition of $D_m^\sPc$. It only depends on the integer $p$ and on the real $\eta$. The latter is due to the use of the Mahalanobis-like norm $\| \cdot \|_{(\Sigma_m^\sPc)^{-1}}$ in~\eqref{eq:Dm} as hinted at in Remark~\ref{rem:norm}. As shall become clearer below, an important practical consequence of this is that the monitoring procedure can be used as soon as it is possible to compute or estimate quantiles of $\Lc_{p,\eta}$ for the chosen parameters $p$ and $\eta$. This important aspect will be investigated in Section~\ref{sec:quant} in more detail. 
\end{remark}

For a given serial dependence scenario under $H_0$ in~\eqref{eq:H0}, it is hoped that Condition~\ref{cond:H0} will hold for many different vectors of points $\Pc = (\bm x_1,\dots,\bm x_p )$. The following proposition shows that this is for instance the case when the time series $( \bm X_i )_{i \in \N}$ is \emph{strongly mixing} under~$H_0$. Given a time series $(\bm Z_i)_{i \in \N}$ and for any $j, k \in \N \cup \{+\infty \}$, denote by $\Mc_j^k$ the $\sigma$-field generated by $(\bm Z_i)_{j \leq i \leq k}$ and recall that the strong mixing coefficients corresponding to $(\bm Z_i)_{i \in \N}$ are defined by
\begin{equation*}
\alpha_r^{\bm Z} = \sup_{k \in \N} \sup_{A \in \Mc_1^k,B\in \Mc_{k+r}^{+\infty}} \big| \P(A \cap B) - \P(A) \P(B) \big|, \qquad r \in \N.
\end{equation*}
The sequence $(\bm Z_i)_{i \in \N}$ is then said to be \emph{strongly mixing} if $\alpha_r^{\bm Z} \to 0$ as $r \to \infty$.

The following result, proven in Appendix~\ref{proof:prop:H0}, is a consequence of Theorem~4 of \cite{KuePhi80}.

\begin{prop}
  \label{prop:cond:H0} Assume that the time series $( \bm X_i )_{i \in \N}$ is stationary and strongly mixing, and that its strong mixing coefficients satisfy $\alpha_r^{\bm X} = O(r^{-a})$ as $r \to \infty$ with $a > 3$. Then, Condition~\ref{cond:H0} holds for all vectors of points $\Pc$ such that Condition~\ref{cond:sigma} holds.
\end{prop}

The previous proposition leads to the following immediate corollary of Theorem~\ref{thm:H0}.

\begin{cor}
\label{cor:H0:mixing}
Assume that the time series $( \bm X_i )_{i \in \N}$ is stationary and strongly mixing, and that its strong mixing coefficients satisfy $\alpha_r^{\bm X} = O(r^{-a})$ as $r \to \infty$ with $a > 3$. Then, for any fixed $\eta > 0$ and any vector of points $\Pc$ such that Condition~\ref{cond:sigma} holds,
\begin{align*}
\sup_{k>m} (m/k)^{\frac32+\eta}  D_m^\sPc(k) \leadsto \Lc_{p,\eta} = \sup_{1 \leq s \leq t < \infty} t^{-\frac32-\eta} \| t \bm W(s) - s \bm W(t) \|_{I_p}.
\end{align*}
\end{cor}

The strong mixing conditions in the previous corollary are for instance satisfied (with much to spare) when $( \bm X_i )_{i \in \N}$ is a stationary vector ARMA process with absolutely continuous innovations \citep[see][]{Mok88}.

The following result, proven in Appendix~\ref{proof:prop:H0}, can be regarded as a multivariate extension of Proposition 3.4 of \cite{HolKoj21}. It shows that imposing that $\eta$ is strictly positive in Theorem~\ref{thm:H0} and Corollary~\ref{cor:H0:mixing} is necessary and sufficient for ensuring that the limiting random variable $\Lc_{p,\eta}$ is almost surely finite.

\begin{prop}
  \label{prop:eta:0}
For any fixed $M > 0$,
\begin{equation*}
  \P \Big(
  \sup_{1 \leq s \leq t < \infty} t^{-\frac32} \| t \bm W(s) - s \bm W(t) \|_{I_p} \geq M \Big) = 1.
\end{equation*}
\end{prop}

\begin{remark}
In relation to the previous result, note that $t^\eta$ in~\eqref{eq:w} could actually be replaced by $h(t)$, where $h(t) = \sqrt{\log \log t}$ when $t>e^e$ and $h(t)=1$ when $t\le e^e$. Indeed, as explained in Remark~3.5 of \cite{HolKoj21}, by the law of the iterated logarithm for Brownian motion, all the results stated before Proposition~\ref{prop:eta:0} should continue to hold with such a modification which could be considered optimal in the sense that, as $t\to \infty$, $h$ diverges slower to infinity than $t \mapsto t^\eta$ for any $\eta > 0$. We did not however consider such a change as it is unwieldy from a practical perspective as shall become clearer from Section~\ref{sec:quant}.  
\end{remark}

The next proposition, also proven in Appendix~\ref{proof:prop:H0}, shows that the weak limit appearing in Theorem~\ref{thm:H0} and Corollary~\ref{cor:H0:mixing} is absolutely continuous. The proof is an application of Theorem 7.1 of \cite{DavLif84} together with an argument allowing us to reduce the problem to compact sets.

\begin{prop}
  \label{prop:L:continuous}
For any $\eta > 0$ and $p \in \N $, $\Lc_{p,\eta}$ is an absolutely continuous random variable.
\end{prop}

Let us finally explain how Theorem~\ref{thm:H0} can be used to carry out the monitoring in practice for a chosen vector of points $\Pc$ for which Condition~\ref{cond:sigma} is assumed to hold. Given a significance level $\alpha \in (0,\frac12)$, suppose that we are able to compute $q_{p,\eta}^{\sss{(1 - \alpha)}}$, the $(1-\alpha)$-quantile of $\Lc_{p,\eta}$. Then, under $H_0$ in~\eqref{eq:H0} and Condition~\ref{cond:H0}, from the Portmanteau theorem, \eqref{eq:typeIerror:asym} holds. Hence, for large $m$, we can expect that, under $H_0$ and Condition~\ref{cond:H0},
$$
\P \Big(  D_m^\sPc(k) > q_{p,\eta}^{\sss{(1 - \alpha)}} (k/m)^{\frac32+\eta} \text{ for some } k \geq m+1 \Big)  \simeq \alpha.
$$
In practice, after the arrival of observation $\bm X_k$, $k > m$, $D_m^\sPc(k)$ is computed from $\bm X_1,\dots, \bm X_k$ and compared to the threshold $q_{p,\eta}^{\sss{(1 - \alpha)}} (k/m)^{\frac32+\eta}$ (or, equivalently, $(m/k)^{\frac32+\eta}D_m^\sPc(k)$ is computed and compared to $q_{p,\eta}^{\sss{(1 - \alpha)}}$). If greater, the null hypothesis is rejected and the monitoring stops. Otherwise, $\bm X_{k+1}$ is collected and the previous iteration is repeated using the $k+1$ available observations.

\begin{remark}
  \label{rem:Dm:sup:asym}
  Under a suitable transformation of Condition~\ref{cond:H0}, we suspect that it is possible to obtain an analogue of Theorem~\ref{thm:H0} for the detector $D_m^{\sup}$ in~\eqref{eq:Dm:sup}. From the closed-end results obtained in Proposition 2.5 of \cite{KojVer21}, we can actually guess the form of the corresponding weak limit. This leads us to believe that, under $H_0$ in~\eqref{eq:H0} and a suitable version of Condition~\ref{cond:H0},
\begin{equation}
  \label{eq:conj}
  \sup_{k>m} (m/k)^{\frac32+\eta}  D_m^{\sup}(k) \leadsto \sup_{1 \leq s \leq t < \infty} t^{-\frac32-\eta} \sup_{\bm x \in \R^d} | t K(s, \bm x) - s K(t, \bm x) |,
\end{equation}
where the limit is almost surely finite and $K$ is a \emph{Kiefer process}, that is, a two-parameter centered Gaussian process whose covariance function is given, for any $s,t \in [0,\infty)$ and $\bm x,\bm y \in \R^d$, by 
\begin{multline}
  \label{eq:Gamma}
  \Gamma(s,t,\bm x,\bm y) = \min(s, t)  \Big( \Cov\{1(\bm X_1 \leq \bm x), 1(\bm X_1 \leq \bm y)\} \\ + \sum_{i=2}^\infty \big[ \Cov\{1(\bm X_1 \leq \bm x), 1(\bm X_i \leq \bm y) \} + \Cov\{1(\bm X_i \leq \bm x), 1(\bm X_1 \leq \bm y) \} \big] \Big).
\end{multline}
Thus, it appears that in general, the weak limit in~\eqref{eq:conj} depends on the characteristics of the underlying time series $( \bm X_i )_{i \in \N}$.  This implies that in general, to carry out monitoring based on the detector $D_m^{\sup}$, one would need to be able to estimate high quantiles of the weak limit in~\eqref{eq:conj} prior to every execution of the procedure. Given the unwieldy form of the weak limit, this seems to be a major obstacle to the use of the detector $D_m^{\sup}$. One exception may be when the monitored observations are univariate ($d=1$), continuous and serially independent. In that case, using a change of variable $F(x) \mapsto u$ (as is classically done for instance when dealing with a Brownian bridge), it can be verified that the weak limit in~\eqref{eq:conj} no longer depends on the characteristics of the underlying time series $( X_i^{\sss[1]} )_{i \in \N}$. To estimate high quantiles of the resulting weak limit, one could then proceed as in forthcoming Section~\ref{sec:quant} where the estimation of high quantiles of the weak limit $\Lc_{p,\eta}$ appearing in Theorem~\ref{thm:H0} is addressed. Due to the apparently rather limited scope of application of an open-end monitoring procedure based on $D_m^{\sup}$, we do not pursue the investigation of such a sequential test in this work and leave this for future research.
\end{remark}

\subsection{Asymptotics under alternatives}
\label{sec:asym:H1}

To complement the previously stated asymptotic results, it is necessary to study the asymptotics of the monitoring procedure under sequences of alternatives to $H_0$ in~\eqref{eq:H0}. Because it is based on the detector $D_m^\sPc$ in~\eqref{eq:Dm}, the studied monitoring procedure is expected to be particularly sensitive to alternative hypotheses of the form 
\begin{multline*}
  H_1: \, \exists \, k^\star \geq m \text{ and } \ell \in \llb 1, p \rrb \text{ such that }  \P(\bm X_1 \leq \bm x_\ell) = \dots = \P(\bm X_{k^\star} \leq \bm x_\ell)  \\ \neq \P(\bm X_{k^\star+1} \leq \bm x_\ell) = \P(\bm X_{k^\star+2} \leq \bm x_\ell) = \dots
\end{multline*}
corresponding to a change in the d.f.\ at one or more of the chosen evaluation points. Note that this can be interpreted as a change in  mean since by rewriting in terms of the univariate time series $\big( Y_i^{\sPc,\sss[\ell]} \big)_{i \in \N} = \big( \1(\bm X_i \leq \bm x_\ell) \big)_{i \in \N}$, we get the following equivalent statement
\begin{multline}
  \label{eq:H1:mean}
  H_1: \, \exists \, k^\star \geq m \text{ and } \ell \in \llb 1, p \rrb \text{ such that }  \Ex(Y_1^{\sPc,\sss[\ell]}) = \dots = \Ex(Y_{k^\star}^{\sPc,\sss[\ell]})  \\ \neq \Ex(Y_{k^\star+1}^{\sPc,\sss[\ell]}) = \Ex(Y_{k^\star+2}^{\sPc,\sss[\ell]}) = \dots.
\end{multline}
As already mentioned at the beginning of Section~\ref{sec:det:thres:asym}, monitoring procedures designed to be particularly sensitive to changes in the mean were studied in \cite{HolKoj21}. Theorem~3.7 and Condition~3.7 in the latter reference specifically provide conditions under which, for a sequence of alternatives to $H_0$ related to $H_1$ in~\eqref{eq:H1:mean}, $\sup_{k>m} (m/k)^{\frac32+\eta}  E_m^{\bm x_\ell}(k) \p \infty$, where $E_m^{\bm x_\ell}$ is defined as in~\eqref{eq:Em} with $\bm x = \bm x_\ell$. For the sake of brevity, we do not restate these conditions with the notation used in this work as they are lengthy to write. Very roughly speaking, they imply that for ``early'' or ``late'' changes in the d.f.\ of the observations at $\bm x_\ell$, the scaled detector $k \mapsto (m/k)^{\frac32+\eta}  E_m^{\bm x_\ell}(k)$ will end up exceeding any fixed threshold provided $m$ is sufficiently large. The following result, proven in Appendix~\ref{proof:prop:H1}, shows that, as $\bm x_\ell \in \Pc$, the same will hold for the scaled detector $k \mapsto (m/k)^{\frac32+\eta}   D_m^\sPc(k)$, where $D_m^\sPc$ is defined in~\eqref{eq:Dm}.

\begin{prop}
  \label{prop:H1}
  Let $\eta > 0$ and assume that for some $\bm x_\ell \in \Pc$, $\sup_{k>m} (m/k)^{\frac32+\eta} E_m^{\bm x_\ell}(k) \p \infty$. Then, if $\Sigma_m^\sPc \p \Sigma^\sPc$, where $\Sigma^\sPc$ is positive-definite and $\Sigma_m^\sPc$ is positive-definite almost surely for all $m \in \N$,
$$
\sup_{k>m} (m/k)^{\frac32+\eta} D_m^\sPc(k) \p \infty.
$$
\end{prop}
\section{The case of continuous observations: practical implementation and additional asymptotic results under the null}
\label{sec:cont}

Prior to using the monitoring procedure based on the detector $D_m^\sPc$ in~\eqref{eq:Dm}, the user needs to choose the points $\Pc = (\bm x_1, \dots, \bm x_p)$. Following the discussion initiated in Remark~\ref{remark:points}, using the learning sample $\bm X_1,\dots,\bm X_m$ to do so seems meaningful. As mentioned in the latter remark, when the observations are discrete, a natural possibility consists of choosing $\bm x_1,\dots,\bm x_p$ from a subset of frequently occurring observations. We focus in this section on the more complicated situation when the learning sample seems to be a stretch from a continuous time series. 

Fix $\eta > 0$ and assume that $m$ is large. Having Theorem~\ref{thm:H0} as well as Remark~\ref{rem:L} and Corollary~\ref{cor:H0:mixing} in mind, one can hope that, under $H_0$ in~\eqref{eq:H0}, $\sup_{k>m} (m/k)^{\frac32+\eta} D_m^\sPc(k)$ has roughly the same distribution as the random variable $\Lc_{p,\eta}$ for all vectors of points $\Pc$ such that Condition~\ref{cond:sigma} holds. A user who is interested in very specific changes in the d.f.\ may choose $\Pc$ accordingly.  Otherwise, one natural possibility is to select the vector of points $\Pc$ such that the coordinates of each of the $p$ points are empirical quantiles computed from the coordinate samples of the learning sample $\bm X_1,\dots,\bm X_m$. As we continue, for any $1 \leq j \leq k$, let  $F_{j:k}^{\sss{[1]}},\dots,F_{j:k}^{\sss[d]}$ be the $d$ univariate margins of $F_{j:k}$ defined in~\eqref{eq:Fjk}. Also, for any univariate d.f.\ $H$, let $H^{-1}$ denote its associated quantile function (generalized inverse) defined by $H^{-1}(y) = \inf\{x \in \R : H(x) \geq y \}$, $y \in [0,1]$, with the convention that $\inf \emptyset = \infty$. Finally, let $\bm \Xc_{1,m}, \dots, \bm \Xc_{p,m}$ denote the points $\bm x_1, \dots, \bm x_p$ when chosen automatically from the learning sample and let $\Pc_m = (\bm \Xc_{1,m}, \dots, \bm \Xc_{p,m} )$.

\subsection{The univariate case}
\label{sec:univ:points}
 
When $d = 1$, a natural instantiation of the previous generic strategy for choosing $\Pc_m$ consists of setting $\Xc_{j,m}^{\sss [1]} = F_{1:m}^{\sss[1], -1} \big( j/(p+1) \big)$, $j \in \llb 1, p \rrb$, that is, the $\Xc_{j,m}^{\sss[1]}$'s are merely taken as the $j/(p+1)$-empirical quantiles of the learning sample $X_1^{\sss[1]}, \dots, X_m^{\sss[1]}$. As we will see in Section~\ref{sec:MC}, this strategy seems to lead to powerful multi-purpose open-end monitoring procedures in the case of univariate observations.  

\subsection{The multivariate case}
\label{sec:mult:points}

A natural first idea when $d > 1$ is simply to apply the univariate strategy above to each component sample of the learning sample $\bm X_1, \dots, \bm X_m$, yielding sets $\Pc_m^{\sss[i]}=\{\Xc_{j,m}^{\sss [i]}: j\in \llb 1, p \rrb\}$ of size $p$ for all $i \in \llb 1, d \rrb$.  For each dimension $i \in \llb 1, d \rrb$, each selected coordinate $\Xc_{j,m}^{\sss [i]}$, $j\in \llb 1, p \rrb$, typically corresponds to a unique $d$-dimensional vector of the learning sample. One could then define $\Pc_m$ to be the union of the corresponding $d$ sets of $d$-dimensional points, which implies that $p \leq |\Pc_m| \leq dp$. A preliminary implementation of this strategy showed however that (among other things) this approach can sometimes lead to the selection of points in the learning sample that are too close to the ``border'' of the point cloud $\bm X_1, \dots, \bm X_m$, resulting in numerical difficulties when computing the inverse or the square root of $\Sigma_m^{\sPc_m}$ (see also Remark~\ref{remark:points}). 

Another natural adaption of the strategy considered in the univariate case would be to choose an integer $r \geq 1$, consider the uniformly-spaced grid (containing $r^d$ points)
\begin{equation}
  \label{eq:unif:grid}
  \Pi = \{(j_1/(r+1), \dots,j_d/(r+1)) : j_1,\dots,j_d \in \llb 1, r \rrb   \} \subseteq (0,1)^d
\end{equation}
and define $\Pc_m$ as consisting of the $r^d$ points
\begin{equation}
  \label{eq:candidate:points}
\left\{ \left( F_{1:m}^{\sss[1], -1} (\pi^{\sss[1]}), \dots, F_{1:m}^{\sss[d], -1} (\pi^{\sss[d]}) \right) : \bm \pi \in \Pi \right\}.   
\end{equation}
This strategy needs however to be refined because some of the above points might not belong to the support of $F$ which, as hinted at in Remark~\ref{remark:points}, is a necessary condition for Condition~\ref{cond:sigma} to hold. Let $\bm U_1, \dots, \bm U_m$ be the unobservable sample obtained from the learning sample $\bm X_1, \dots, \bm X_m$ by probability integral transformations, that is, let
\begin{equation}
  \label{eq:prob:int:trans}
\bm U_i = (U_i^{\sss[1]},\dots,U_i^{\sss[d]}) = \big( F^{\sss[1]}(X_i^{\sss[1]}),\dots,F^{\sss[d]}(X_i^{\sss[d]}) \big), 
\end{equation}
where $F^{\sss[1]},\dots,F^{\sss[d]}$ are the $d$ unknown univariate margins of $F$. Note in passing that $\bm U_1, \dots, \bm U_m$ can be regarded as a stretch from a $d$-dimensional time series of continuous random vectors with contemporary d.f.\ $C$, where $C$ is the (unique) copula of $F$  \citep[see, e.g.,][]{Skl59} satisfying
$$
C(\bm u) = F \big( F^{\sss[1],-1} (u^{\sss[1]}), \dots, F^{\sss[d],-1} (u^{\sss[d]}) \big), \qquad \bm u \in [0,1]^d,
$$
and
$$
F(\bm x) = C \big( F^{\sss[1]}(x^{\sss[1]}),\dots,F^{\sss[d]}(x^{\sss[d]}) \big), \qquad \bm x \in \R^d.
$$
Adapting the approach briefly described in Section 4.2 of \cite{LiGen13}, we propose to keep in $\Pc_m$ only those points in~\eqref{eq:candidate:points} constructed from grid points in~\eqref{eq:unif:grid} whose ``neighborhood'' contains a sufficiently large proportion of the $\bm U_i$'s. As $F^{\sss[1]},\dots,F^{\sss[d]}$ are unknown, we follow one of the classical approaches used in the copula literature \citep[see, e.g.,][and the references therein]{HofKojMaeYan18} and use $\pobs{U}_1, \dots, \pobs{U}_m$ as a proxy for $\bm U_1, \dots, \bm U_m$, where
\begin{equation}
  \label{eq:pseudo:obs}
\pobs{U}_i = \frac{m}{m+1} \big( F_{1:m}^{\sss[1]}(X_i^{\sss[1]}),\dots,F_{1:m}^{\sss[d]}(X_i^{\sss[d]}) \big).
\end{equation}
For any $\bm a, \bm b \in [0,1]^d$ such that $\bm a < \bm b$, let $(\bm a, \bm b] = \{\bm u \in [0,1]^d : \bm a < \bm u \leq \bm b \}$. Furthermore, let $\nu_m$ be the empirical measure of $\pobs{U}_1, \dots, \pobs{U}_m$ and let $\bm s = \big( 1/(r+1), \dots, 1/(r+1) \big) \in \R^d$. Given $\bm \pi \in \Pi$ and if the $d$ components of $\bm U_1$ are independent, it is expected that the proportion of $\bm U_i$'s in $(\bm \pi - \bm s, \bm \pi ]$ be approximately equal to $1/(r+1)^d$. This motivates the following strategy: we choose to retain in $\Pc_m$ only the points in~\eqref{eq:candidate:points} constructed from grid points $\bm \pi \in \Pi_m$, where 
\begin{equation}
  \label{eq:Gc:m}
\Pi_m = \left\{ \bm \pi \in \Pi : \nu_m \big( (\bm \pi - \bm s, \bm \pi ] \big) > \frac{1}{\kappa (r+1)^d} \right\},
\end{equation}
$\Pi$ is defined in~\eqref{eq:unif:grid}  and $\kappa > 1$ is a user-chosen parameter. The number of automatically chosen points $p = |\Pi_m|$ depends on $m$. Figure~\ref{fig:point:sel} illustrates the automatic choice of $\Pc_m$ in the bivariate case for $r = 4$ and $\kappa = 1.5$. Values for $\kappa$ and $r$ appearing to lead to powerful multi-purpose open-end monitoring procedures will be recommended in the case $d \in \{2,3\}$ in Section~\ref{sec:MC}.

\begin{figure}[t!]
\begin{center}
  \includegraphics*[width=0.75\linewidth]{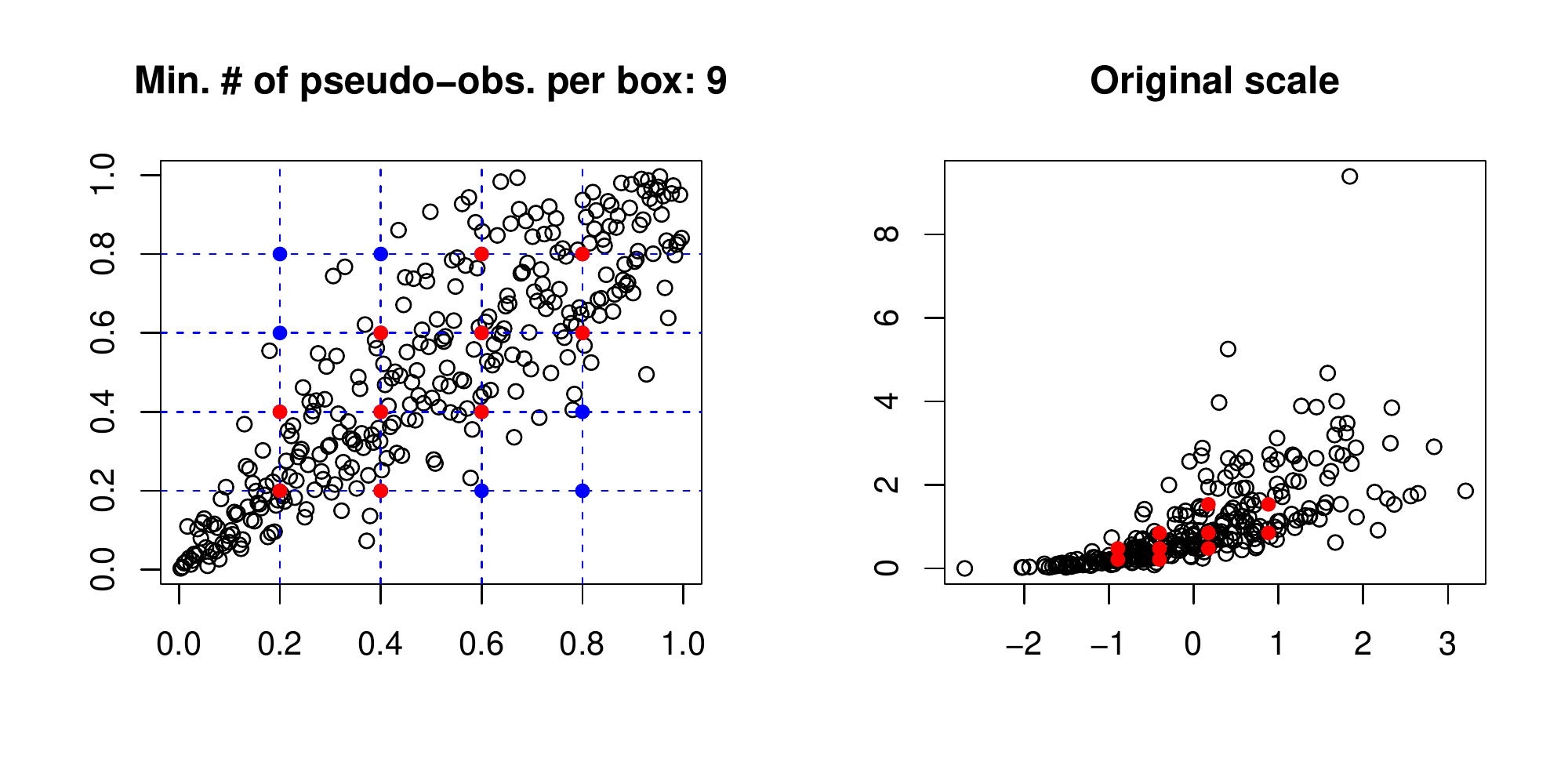}
  \caption{\label{fig:point:sel} Automatic choice of $\Pc_m$ in the bivariate case for $m=300$, $r = 4$ and $\kappa = 1.5$. Left: scatterplot of the ``pseudo-observations'' $\pobs{U}_1, \dots, \pobs{U}_m$ obtained from the learning sample, initial uniformly-spaced grid $\Pi$ in~\eqref{eq:unif:grid} in blue and selected points $\Pi_m \subseteq \Pi$ in red. Right: scatterplot of the learning sample in black and points in $\Pc_m$ in red. }
\end{center}
\end{figure}

We end this section by stating an asymptotic property of the proposed selection procedure. Let $\nu_C$ be the measure on $[0,1]^d$ associated with the copula $C$ of $F$ and let 
\begin{equation}
  \label{eq:Gc:C}
\Pi_C = \left\{\bm \pi \in \Pi : \nu_C \big( (\bm \pi - \bm s, \bm \pi ] \big) > \frac{1}{\kappa (r+1)^d} \right\}.
\end{equation}
Also, recall the definition of $\Pi_m$ in~\eqref{eq:Gc:m} and, as classically done in the literature \citep[see, e.g.,][Chapter~4 and the references therein]{HofKojMaeYan18}, let the empirical copula $C_m$ of $\bm X_1,\dots,\bm X_m$ be defined as the empirical d.f.\ of the ``pseudo-observations'' $\pobs{U}_1, \dots, \pobs{U}_m$ defined in~\eqref{eq:pseudo:obs}. Of course, one hopes that $\Pi_m$ is close to the deterministic (but unknown) set $\Pi_C$ when $m$ is large. Proposition~\ref{prop:mult:point} below (proven in Appendix~\ref{proofs:prop:points}) makes this statement rigorous  under the following mild condition.

\begin{cond} $\strut$
  \label{cond:mult:point}
  \begin{itemize}
  \item[\textup{(i)}]  For each $\bm \pi \in \Pi$, $\nu_C \big( (\bm \pi - \bm s, \bm \pi ] \big) \neq 1/ (\kappa (r+1)^d)$, and 
  \item[\textup{(ii)}] $\sup_{\bm u \in [0,1]^d} | C_m(\bm u) -  C(\bm u) | \as 0$.
  \end{itemize}
\end{cond}

\begin{prop}
  \label{prop:mult:point}
  Assume that Condition~\ref{cond:mult:point} holds. Then, almost surely, for all $m$ sufficiently large, $\Pi_m=\Pi_C$.
\end{prop}

In other words, under Condition~\ref{cond:mult:point} and provided $m$ is sufficiently large, we can regard $\Pc_m = (\bm \Xc_{1,m}, \dots, \bm \Xc_{p,m} )$ as being formed of the $p = |\Pi_C|$ points
\begin{equation*}
\big\{\left( F_{1:m}^{\sss[1], -1} (\pi^{\sss[1]}), \dots, F_{1:m}^{\sss[d], -d} (\pi^{\sss[d]}) \right): \bm \pi \in \Pi_C\big\},  
\end{equation*}
where $\Pi_C$ is defined in~\eqref{eq:Gc:C}.

\subsection{Additional asymptotic results under the null}

The asymptotic results stated in Sections~\ref{sec:asym:H0} and~\ref{sec:asym:H1} concern the monitoring procedure based on the detector $D_m^\sPc$ in~\eqref{eq:Dm} where the chosen evaluation points $\Pc = (\bm x_1, \dots, \bm x_p)$ are fixed (they are not allowed to change in the asymptotics with $m$). To fully asymptotically justify the use of the detector $D_m^{\sPc_m}$ resulting from the automatic choices of points $\Pc_m = (\bm \Xc_{1,m}, \dots, \bm \Xc_{p,m} )$ considered in Sections~\ref{sec:univ:points} or~\ref{sec:mult:points}, one needs an analogue of Theorem~\ref{thm:H0} in which the points at which the empirical d.f.s are evaluated are allowed to change with $m$. 

In the rest of this section, we assume that the automatically chosen points in $\Pc_m$ are of the form
\begin{equation}
  \label{eq:Xcm}
\bm \Xc_{i,m} = \left( F_{1:m}^{\sss[1], -1} (\pi_i^{\sss[1]}), \dots, F_{1:m}^{\sss[d], -1} (\pi_i^{\sss[d]}) \right), 
\end{equation}
for some $p$ vectors of probabilities $\bm \pi_1,\dots,\bm \pi_p \in (0,1)^d$ not depending on $m$. This is clearly the case for the univariate selection strategy proposed in Section~\ref{sec:univ:points}. From Proposition~\ref{prop:mult:point}, it is also the case for the multivariate strategy proposed in Section~\ref{sec:mult:points} upon additionally assuming that Condition~\ref{cond:mult:point} holds and that $m$ is sufficiently large.

Next, set $\bm x_i = \big( F^{\sss[1],-1}(\pi_{i}^{\sss[1]}), \dots, F^{\sss[d],-1}(\pi_{i}^{\sss[d]}) \big)$, write $\Pc = (\bm x_1,\dots, \bm x_p)$ and notice that the points in $\Pc$ are unobservable since $F^{\sss[1]},\dots,F^{\sss[d]}$ are unknown. The $p$-dimensional random vectors $\bm Y_i^\sPc = \big( \1(\bm X_i \leq \bm x_1), \dots, \1(\bm X_i \leq \bm x_p) \big)$ are also unobservable. Since $\Pc_m = (\bm \Xc_{1,m}, \dots, \bm \Xc_{p,m} )$ with $\bm \Xc_{i,m}$ given by~\eqref{eq:Xcm} is an estimator of $\Pc$, the long-run covariance matrix $\Sigma^\sPc$ in~\eqref{eq:sigma} of the unobservable $p$-dimensional time series $\big( \bm Y_i^\sPc \big)_{i \in \N}$ may still be estimated from the sample $\bm Y_1^{\sPc_m},\dots,\bm Y_m^{\sPc_m}$ which is a proxy for the sample $\bm Y_1^\sPc,\dots,\bm Y_m^\sPc$ . 

We first state a condition under which the monitoring procedure based on $D_m^{\sPc_m}$ and the unobservable monitoring procedure based on $D_m^\sPc$ in~\eqref{eq:Dm} are asymptotically equivalent.

\begin{cond}[For the asymptotic equivalence of $D_m^{\sPc_m}$ and $D_m^\sPc$]
  \label{cond:points}
  For any $i \in \llb 1, p \rrb$,
  \begin{equation}
    \label{eq:cond:points}
    \sup_{k>m} k^{-\frac12} \max_{j \in \llb 1, k \rrb} j | F_{1:j}(\bm \Xc_{i,m}) - F(\bm \Xc_{i,m}) -  F_{1:j}(\bm x_i) +  F(\bm x_i) |  = o_\P(1).
  \end{equation}
\end{cond}

The following result, proven in Appendix~\ref{proofs:prop:points}, shows that the previous condition can be satisfied under the null and \emph{absolute regularity}. Given a time series $(\bm Z_i)_{i \in \N}$, recall that, for $j, k \in \N \cup \{+\infty \}$, $\Mc_j^k$ denotes the $\sigma$-field generated by $(\bm Z_i)_{j \leq i \leq k}$, that the absolute regularity coefficients corresponding to $(\bm Z_i)_{i \in \N}$ are defined by
\begin{equation}
\label{eq:beta}
\beta_r^{\bm Z} = \Ex \left\{ \sup_{k \in \N} \sup_{B\in \Mc_{k+r}^{+\infty}} \big| \P(B \mid \Mc_1^k) - \P(B) \big| \right\}, \qquad r \in \N,
\end{equation}
and that the sequence $(\bm Z_i)_{i \in \N}$ is said to be \emph{absolutely regular} if $\beta_r^{\bm Z} \to 0$ as $r \to \infty$. Also, note that absolute regularity is known to imply strong mixing \citep[see, e.g.,][]{DehPhi02} and that an independent sequence is clearly absolutely regular since in this case for every $k$ and $B$ as in~\eqref{eq:beta}, $\P(B|\Mc_1^k)=\P(B)$ almost surely. 

\begin{prop}[Condition~\ref{cond:points} can hold under the null]
  \label{prop:cond:points}
Assume that the underlying time series $(\bm X_i)_{i \in \N}$ is stationary and absolutely regular, and that its absolute regularity coefficients satisfy $\beta^{\bm X}_r = O(r^{-a})$ as $r \to \infty$ with $a>1$. Then, if the $d$ univariate margins $F^{\sss[1]},\dots,F^{\sss[d]}$ of $F$ are continuous and if, for each $i \in \llb 1, p \rrb$, $\bm \Xc_{i,m} \p \bm x_i$, Condition~\ref{cond:points} holds. 
\end{prop}

\begin{remark}
  Assumptions related to Condition~\ref{cond:points} appear in \cite{DetGos20} in the context of the study of the asymptotics of closed-end sequential tests designed to be sensitive to changes in the mean, the variance or certain quantiles. In an open-end setting, related conditions are stated in Assumption~2.5 of \cite{GosKleDet21} and in Condition~6.1 of \cite{HolKoj21} under an ``almost sure'' form. As an inspection of the proof of Proposition~\ref{prop:cond:points} reveals, Condition~\ref{cond:points} is essentially a consequence of the continuity of the margins of $F$ and Theorem 3.1 of \cite{DedMerRio14} which provides an adequate strong approximation result for the empirical process under absolute regularity (but not under strong mixing).
\end{remark}

\begin{remark}
In the statement of Proposition~\ref{prop:cond:points}, it is assumed that, for each $i \in \llb 1, p \rrb$, $\bm \Xc_{i,m} \p \bm x_i$. This condition can actually be dispensed with provided additional conditions of the true unobservable quantile functions $F^{\sss[1],-1}, \dots, F^{\sss[d],-1}$ are assumed instead.  Indeed, from \cite{Rio98}, we know that the condition on the absolute regularity coefficients in Proposition~\ref{prop:cond:points} implies that, for any $\ell \in \llb 1, d \rrb$, $F_{1:m}^{\sss[\ell]} \p F^{\sss[\ell]}$ in $\ell^\infty(\R)$, where $\ell^\infty(\R)$ denotes the space of bounded functions on $\R$ equipped with the uniform metric. From Lemma 21.2 in \cite{van98}, this is then equivalent to the fact that $F_{1:m}^{\sss[\ell],-1}(\pi) \p F^{\sss[\ell],-1}(\pi)$ at every $\pi \in (0,1)$ at which $F^{\sss[\ell],-1}$ is continuous. Consequently, the condition that, for any $i \in \llb 1, p \rrb$, $\bm \Xc_{i,m} \p \bm x_i$ could be replaced by the condition that, for any $\ell \in \llb 1, d \rrb$, $F^{\sss[\ell],-1}$ is continuous at $\pi_i^{\sss[\ell]}$, for all $i \in \llb 1, p \rrb$.  
\end{remark}

The next proposition, also proven in Appendix~\ref{proofs:prop:points}, states that, under Conditions~\ref{cond:H0} and~\ref{cond:points}, the monitoring procedures based on $D_m^{\sPc_m}$ and $D_m^\sPc$ are asymptotically equivalent.

\begin{prop}
  \label{prop:equiv:points}
  Under Conditions~\ref{cond:H0} and~\ref{cond:points}, and if $\Sigma_m^\sPc \p \Sigma^\sPc$ and $\Sigma_m^{\sPc_m} \p \Sigma^\sPc$, for any $\eta > 0$, 
$$
\sup_{k>m}  (m/k)^{\frac32+\eta} | D_m^{\sPc_m}(k) - D_m^\sPc(k) | = o_\P(1),
$$
and, consequently,
\begin{align*}
\sup_{k>m} (m/k)^{\frac32+\eta}  D_m^{\sPc_m}(k) \leadsto \Lc_{p,\eta} = \sup_{1 \leq s \leq t < \infty} t^{-\frac32-\eta} \| t \bm W(s) - s \bm W(t) \|_{I_p}.
\end{align*}
\end{prop}

The last claim of the previous proposition suggests to carry out the monitoring procedure based on the detector $D_m^{\sPc_m}$ exactly as the procedure based on the detector $D_m^\sPc$ when the points $\Pc$ are hand-picked by the user (see the last paragraph of Section~\ref{sec:asym:H0}).

\subsection{The monitoring procedure based on $D_m^{\sPc_m}$ is margin-free under the null}
\label{sec:margin:free}
  
We end this section by verifying that, under the considered assumption that the true unknown marginal d.f.s $F^{\sss[1]},\dots,F^{\sss[d]}$ are continuous, the procedure based on the detector $D_m^{\sPc_m}$ is margin-free under $H_0$ in~\eqref{eq:H0}, that is, it does not depend on $F^{\sss[1]},\dots,F^{\sss[d]}$ under the null. To see this, let $\bm U_1, \dots, \bm U_k$ be the unobservable sample obtained from the available observations $\bm X_1, \dots, \bm X_k$ using \eqref{eq:prob:int:trans} and let $G_{1:m}$ be the empirical d.f.\ of $\bm U_1,\dots,\bm U_m$. Notice that we can recover the $\bm X_i$ from the $\bm U_i$ by marginal quantile transformations, that is, $\bm X_i = (F_1^{\sss[1],-1}(U_i^{\sss[1]}),\dots,F_d^{\sss[d],-1}(U_i^{\sss[d]}))$. Furthermore, for any $j\in \llb 1,d \rrb$, by (right) continuity of $F^{\sss[j]}$, we have that $\1 \{ F^{\sss[j],-1}(u) \leq x \} = \1 \{ u \leq F^{\sss[j]}(x) \}$ for all $u \in [0,1]$ and $x \in \R$; see, e.g., Proposition~1~(5) in \cite{EmbHof13}. Then, it can be verified that, for every $\bm \pi \in (0,1)^d$, $i\in\llb 1,k\rrb$ and $j \in \llb 1,d\rrb$,
$$
1\{X_i^{[j]} \leq  F_{1:m}^{\sss[j],-1}(\pi^{\sss[j]})\} = 1 \{ U_i^{[j]} \leq  F^{\sss[j]} ( F_{1:m}^{\sss[j],-1}(\pi^{\sss[j]}) ) \} =  1\{U_i^{[j]} \leq  G_{1:m}^{\sss[j],-1}(\pi^{\sss[j]})\},
$$
which implies that, under the null and the current setting, the detector at $k$ can be rewritten to depend only on $\bm U_1,\dots,\bm U_k$.

\section{Estimation of high quantiles of the limiting distribution}
\label{sec:quant}

From the two previous sections, we know that, to carry out the studied monitoring procedures, it is necessary to be able to accurately estimate high quantiles of the random variable $\Lc_{p,\eta}$, $p \geq 1$, $\eta > 0$, appearing first in the statement of Theorem~\ref{thm:H0}. The underlying estimation problem was empirically solved in \cite{HolKoj21} for $p = 1$ using \emph{asymptotic regression modeling}. We choose to use the same approach when $p > 1$ and refer the reader to Section~4 of the aforementioned reference where the motivation for this way of proceeding is explained in detail.

Fix $\eta > 0$, $p \geq 1$, $\alpha \in (0,\frac12)$ and let $q_{p,\eta}^{\sss{(1 - \alpha)}}$ be the $(1-\alpha)$-quantile of $\Lc_{p,\eta}$. Furthermore, let $d = 1$, let $\big( X_i^{\sss[1]} \big)_{i \in \N}$ be an infinite sequence of independent standard normals and let $\Pc = (x_1^{\sss[1]},\dots, x_p^{\sss[1]})$ where  $x_i^{\sss[1]} = \Phi^{-1}\big(i/(p+1)\big)$, $i \in \llb 1, p \rrb$ and $\Phi$ is the d.f.\ of the standard normal. According to Theorem~\ref{thm:H0}, for large $m$ the distribution of $\sup_{k > m} (m/k)^{\frac32+\eta} D_m^\sPc (k)$ should be close to that of $\Lc_{p,\eta}$. If for a given realization of $\big( X_i^{\sss[1]} \big)_{i \in \N}$ we could compute the corresponding realization of $\sup_{k > m} (m/k)^{\frac32+\eta} D_m^\sPc (k)$, then  $q_{p,\eta}^{\sss{(1 - \alpha)}}$ could be estimated by $\hat q_{p,\eta}^{\sss{(1 - \alpha)}}$, the $(1-\alpha)$-empirical quantile of a large sample of realizations of $\sup_{k > m} (m/k)^{\frac32+\eta} D_m^\sPc (k)$. As this is not possible because of the supremum over $k > m$, the idea taken from \cite{HolKoj21} is to model the relationship between $r \in \llb 9, 16 \rrb$ and $\hat q_{p,\eta,r}^{\sss{(1 - \alpha)}}$, the $(1-\alpha)$-empirical quantile of $\sup_{k \in \rrb m, m+2^r \rrb } (m/k)^{\frac32+\eta} D_m^\sPc (k)$, using an asymptotic regression model.

To begin with, using a computer grid, we computed the empirical quantiles $\hat q_{p,\eta,r}^{\sss{(1 - \alpha)}}$, $r \in \llb 9, 16 \rrb$, from 10,000 simulated trajectories of the scaled detector $k \mapsto (m/k)^{\frac32+\eta} D_m^\sPc (k)$ for $k \in \rrb m, m + 2^{16} \rrb$ and $m=500$. In a next step, an asymptotic regression model was fitted to the points $(r, \hat q_{p,\eta,r}^{\sss{(1 - \alpha)}})$, $r \in \llb 9, 16 \rrb$. The considered model is a three-parameter model with mean function 
\begin{align*}
f(x) = \beta_1 + (\beta_2-\beta_1)\{ 1-\exp(-x/\beta_3)\},
\end{align*}
where $y = \beta_2$ is the equation of the upper horizontal asymptote of $f$. Its fitting was carried out using the \textsf{R} package \texttt{drc} \citep{drc}. A candidate estimate $\hat q_{p,\eta}^{\sss{(1 - \alpha)}}$ of $q_{p,\eta}^{\sss{(1 - \alpha)}}$, the $(1-\alpha)$-quantile of $\Lc_{p,\eta}$, is then the resulting estimate of the parameter $\beta_2$. The previous steps were carried out for $p \in \{2,5,10,20\}$, $\alpha \in \{0.01, 0.05, 0.1\}$ and $\eta = 0.001$ (following the practical recommendation made in \cite{HolKoj21}), and can be visualized in the first four panels of Figure~\ref{fig:quantiles} for $\alpha = 0.05$. The corresponding estimates of the quantiles $q_{p,\eta}^{\sss{(1 - \alpha)}}$ are given in columns two to five of Table~\ref{tab:quantiles}.

\begin{figure}[t!]
\begin{center}
  \includegraphics*[width=1\linewidth]{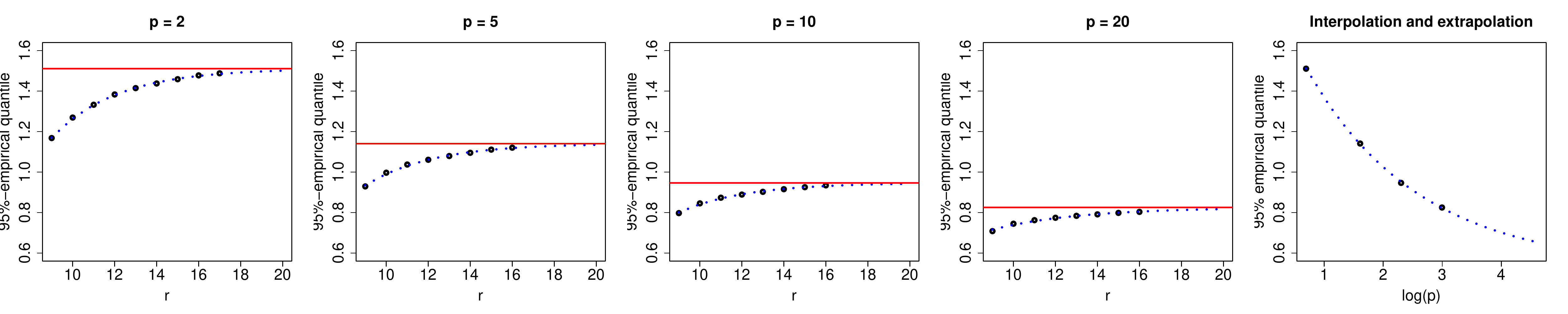}
  \caption{\label{fig:quantiles} First four panels: for $\alpha = 0.05$, $\eta = 0.001$ and $p \in \{2,5,10,20\}$, scatter plots of $\{(r,\hat q_{p,\eta,r}^{\sss{(1 - \alpha)}})\}_{r \in \llb 9, 16 \rrb}$, corresponding fitted asymptotic regression models (dotted blue curves) and estimates of the upper horizontal asymptotes (solid red lines) which are candidate estimates of $q_{p,\eta}^{\sss{(1 - \alpha)}}$, the $(1-\alpha)$-quantile of $\Lc_{p,\eta}$. Fifth panel:  scatter plot of $\{(\log(p),\hat q_{p,\eta}^{\sss{(1 - \alpha)}})\}_{p  \in \{2,5,10,20\}}$ and corresponding transformed fitted asymptotic regression model (dotted blue curve) that could be used to interpolate (resp.\ extrapolate) the value of $\hat q_{p,\eta}^{\sss{(1 - \alpha)}}$ for $p \in [2,20]$ (resp.\ for $p$ slightly larger than 20).}
\end{center}
\end{figure}

\input{quantiles.tex}

In a last step, to be able to carry out the monitoring procedures for some values of $p$ different than those in $\{2,5,10,20\}$, we fitted asymptotic regression models to the points $(\log(p),2-\hat q_{p,\eta}^{\sss{(1 - \alpha)}})$, $p  \in \{2,5,10,20\}$, for $\alpha \in \{0.01, 0.05, 0.1\}$ and $\eta = 0.001$.  The estimates of the parameters $\beta_1$, $\beta_2$ and $\beta_3$ are reported in the last three columns of Table~\ref{tab:quantiles}. We make no claim regarding the theoretical adequacy of this type of model. The aim is only to be able to interpolate (resp.\ extrapolate) the value of $\hat q_{p,\eta}^{\sss{(1 - \alpha)}}$ for $p \in [2,20]$ (resp.\ for $p$ slightly larger than 20). Note that, since $\Lc_{p,\eta} \as 0$ as $p \to \infty$ (as a consequence of the definition of the norm $\|\cdot\|_{I_p}$), a fitted two-parameter submodel with $\beta_2$ fixed to 2 is expected to behave better for large $p$. We have nonetheless decided to keep the fitted three parameter model because its accuracy for values of $p$ slightly larger than 20 was found to be better in our Monte Carlo experiments summarized in the forthcoming section.

\section{Monte Carlo experiments}
\label{sec:MC}

We carried out rather extensive numerical experiments in the case of low-dimensional ($d \in \{1,2,3\}$) continuous observations to investigate the finite-sample behavior of the monitoring procedure based on the detector $D_m^{\sPc_m}$ introduced in Section~\ref{sec:cont}. Specifically, for $d=1$, recall from Section~\ref{sec:univ:points} that the $p$ real points at which the (univariate) empirical d.f.s are evaluated are chosen as empirical quantiles of order $i/(p+1)$, $i \in \llb 1, p \rrb$, computed from the learning sample $X_1^{\sss[1]},\dots,X_m^{\sss[1]}$. For $d \in \{2,3\}$, the approach is slightly more involved: as explained in Section~\ref{sec:mult:points}, the evaluation points are chosen using a point selection procedure which relies on two parameters: an integer $r \geq 1$ specifying the maximal value $r^d$ of $p$ and the constant $\kappa$ controlling how many of the initial $r^d$ grid points will actually be retained. 

From the definition of the detector $D_m^{\sPc_m}$, we see that an underlying unknown long-run covariance matrix needs to be estimated from the sample $\bm Y_1^{\sPc_m},\dots,\bm Y_m^{\sPc_m}$. In practice, for $\Sigma_m^{\sPc_m}$, we used the estimator of \cite{And91} based on the quadratic spectral kernel with automatic bandwidth selection as implemented in the function \texttt{lrvar()} of the \textsf{R} package \texttt{sandwich} \citep{Zei04,ZeiKolGra20}. Note that we did not however use prewhitening as suggested in \cite{AndMon92}. The fact that the monitoring procedure studied in this work is available in the \textsf{R} package \texttt{npcp} \citep{npcp} not only makes all our experiments fully reproducible but also allows a user to change the long run covariance estimator \citep[for instance to that of][]{NewWes87} by passing parameters to the main function which will be passed to the function \texttt{lrvar()}.

Before we present our empirical findings,  it is important to keep in mind that, in the case of open-end approaches, numerical experiments only provide a biased view of their behavior as finite computing resources impose that monitoring has to be stopped eventually. From the point of view of statistical testing,  the main consequence of this is that all rejection percentages are underestimated.

The full details of our simulations are available in Appendix~\ref{sec:MC:details}. We provide hereafter a summary of our findings:
\begin{itemize}

\item For monitoring univariate data, taking $p \in \{5,\dots,10\}$ evaluation points and choosing them as suggested in Section~\ref{sec:univ:points} seems a good choice in general. Unless serial dependence is very strong, a learning sample of size $m \geq 800$ seems to lead to a procedure that holds its level well and displays good power against several alternatives involving a change in the mean or the variance, or such that the d.f.\ changes with the mean and variance remaining constant. In the case of very strong serial dependence (such as for an AR(1) model with autoregressive parameter 0.7), a larger learning sample (for instance $m=1600$) seems necessary for the monitoring procedure to hold its level reasonably well.

\item For monitoring low-dimensional data ($d \in \{2,3\}$), using $r=4$ if $d=2$, $r=3$ if $d=3$ and $\kappa = 1.5$ in the point selection procedure of Section~\ref{sec:mult:points} seems to be a reasonable choice in general.  As for $d=1$, in the case of mild serial dependence, taking $m=800$ seems sufficient to obtain a sequential test that holds its level reasonably well. With such settings, the procedure seems powerful against various alternatives involving changes in one margin or in the copula of the contemporary d.f. 
  
\end{itemize}

\section{Data example and concluding remarks}

Let us briefly illustrate how the procedure based on the detector $D_m^{\sPc_m}$ considered in Section~\ref{sec:cont} could be used for monitoring changes in the contemporary distribution of daily log-returns of a financial index. In our fictitious example, the learning sample (whose stationarity would need to be tested) corresponds to 1257 trading days of the NASDAQ for the period January 3rd 2012 -- December 30th 2016. Monitoring starts on the first trading day of 2017. The corresponding daily log-returns are represented in Figure~\ref{fig:appl-returns}, where the dashed vertical line represents the start of the monitoring. 

\begin{figure}[t!]
\begin{center}
  \includegraphics*[width=1\linewidth]{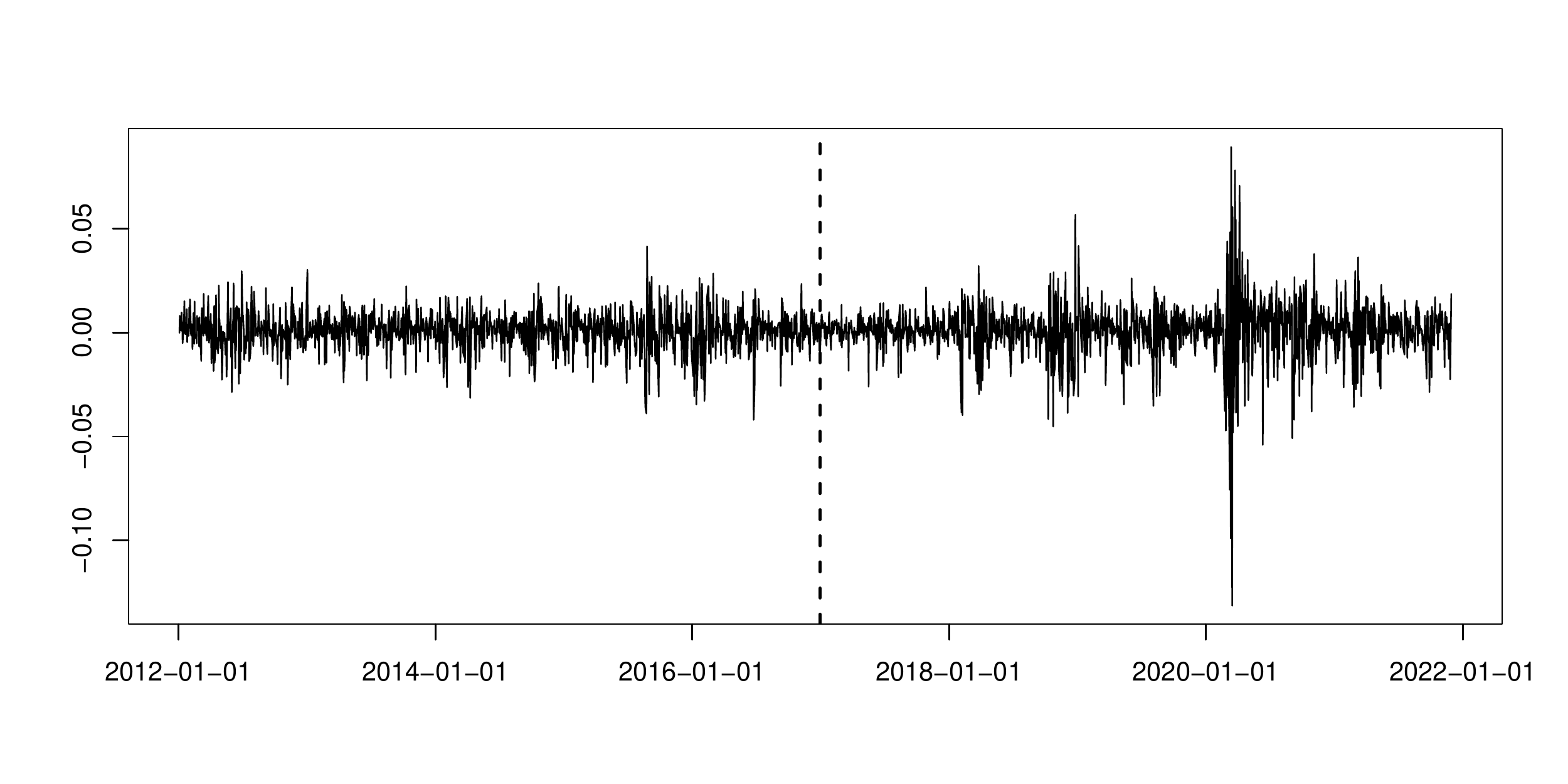}
  \caption{\label{fig:appl-returns} Daily log-returns computed from closing quotes of the NASDAQ composite index from 2012 to 2021. The learning sample in our fictitious data example corresponds to the period 2012 -- 2016. Monitoring starts on the first trading day of 2017, which is represented by a dashed vertical line.}
\end{center}
\end{figure}

The sample paths of $k \mapsto (m/k)^{\frac32+\eta} D_m^{\sPc_m}(k)$ are represented in Figure~\ref{fig:appl-detector} for $p \in \{5, 10\}$. The horizontal dashed line in each panel represent one of the 95\%-quantiles given in the second row of Table~\ref{tab:quantiles}. The dashed vertical lines mark the first time that the threshold is exceeded (corresponding to May 2020). A change in the data generating process slightly prior to this date seems  likely as an inspection of Figure~\ref{fig:appl-returns} reveals: the first months of 2020 are indeed characterized by a period of very high volatility.

\begin{figure}[t!]
\begin{center}
  \includegraphics*[width=1\linewidth]{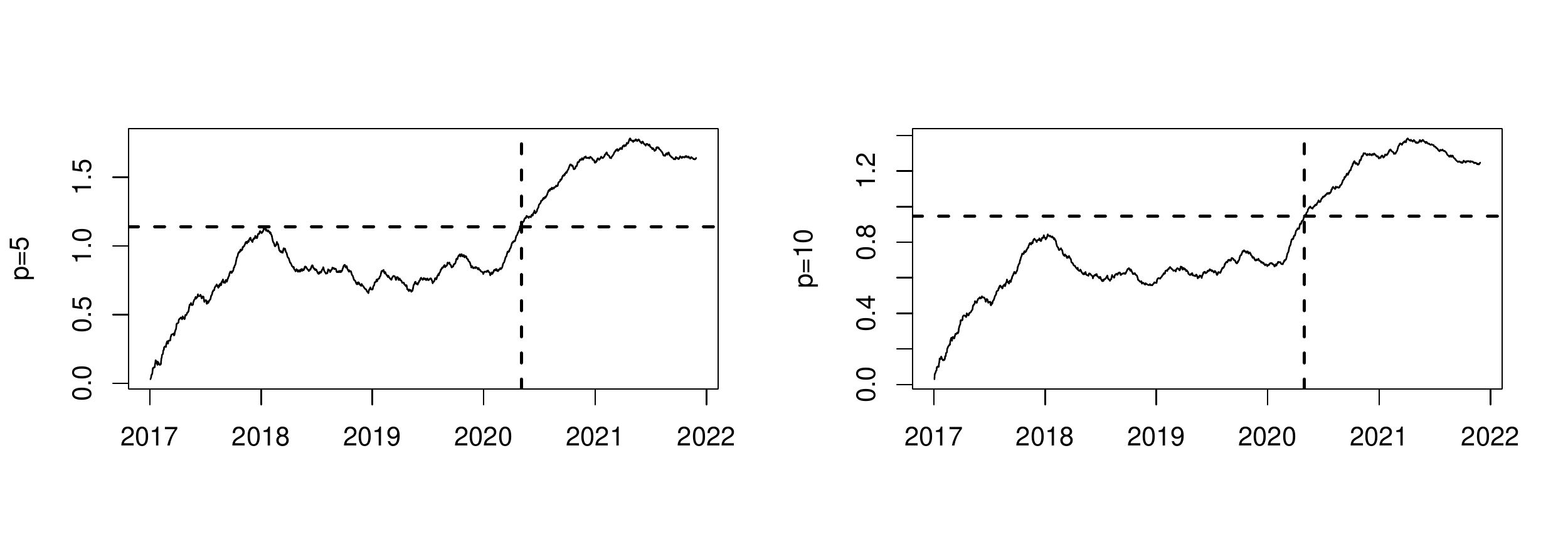}
  \caption{\label{fig:appl-detector} Sample paths of $k \mapsto (m/k)^{\frac32+\eta} D_m^{\sPc_m}(k)$ for $p=5$ and $10$. The horizontal dashed lines represent the 95\%-quantiles given in the second row of Table~\ref{tab:quantiles}. The vertical dashed lines mark the first time the scaled versions of the detectors exceed the corresponding quantiles.}
\end{center}
\end{figure}

We end this section by stating a few remarks:
\begin{itemize}

\item One important practical advantage of the monitoring procedure proposed in this work comes from its open-end nature and is that the monitoring horizon does not need to be specified. As a consequence, monitoring could theoretically run forever. This is however not possible from a practical perspective because of the form of the detector in~\eqref{eq:Dm}: it is indeed clear that the cost of computing the detector at time $k$ increases with $k$. The latter implies that its computation will become impossible for $k$ large enough. Some other detectors, such as the ordinary CUSUM, will not be affected by such a practical issue. Yet, as empirically observed in \cite{GosKleDet21} and in \cite{HolKoj21}, the ordinary CUSUM is substantially less powerful than more computationally costly detectors similar to the ones considered in this work. In a related way, let us mention that the implementation of the studied monitoring procedure available in the \textsf{R} package \texttt{npcp} is merely a proof of concept and is not optimized for very long-term monitoring.

\item The price to pay for open-end monitoring is that the detection power decreases as time elapses. For the studied class of procedures, this is due to the parameter $\eta$ as discussed in Section~4 of \cite{HolKoj21}. On one hand, the smaller the value of $\eta$, the weaker the power decrease. On the other hand, the smaller $\eta$, the more conceptually difficult and computationally costly it is to estimate high quantiles of the distribution $\Lc_{p,\eta}$ appearing in Theorem~\ref{thm:H0}. The latter suggests to devote more research to the estimation of high quantiles of $\Lc_{p,\eta}$.
  
\item The type of monitoring procedure used in this work could also be used to detect changes in the serial dependence. For instance, to detect such changes ``at lag 1'' from an initial sequence of univariate observations $Z_1,\dots,Z_m,Z_{m+1},\dots$, the $\bm X_i$ could be formed as $\bm X_i = (Z_i,Z_{i+1})$.

\end{itemize}

\section*{Acknowledgments}

The authors would like to thank a co-editor and two anonymous referees for their very constructive comments on an earlier version of this manuscript. MH was supported in part by Future Fellowship FT160100166 from the Australian Research Council. AV is supported by a University of Melbourne Research Scholarship.


\begin{appendix}

\section{Proof of Theorem \ref{thm:H0}}
\label{proof:thm:H0}

The proof of Theorem~\ref{thm:H0} is based on four lemmas which we prove first. Throughout the remainder of the proof, we let
\begin{align}
  \label{eq:tilde:Dm}
\tilde D_m^\sPc(k)
&=
\max_{j \in \llb m,k \llb} \frac{j (k-j)}{m^{\frac32}}  \| \bm {\bar Y}_{1:j}^\sPc - \bm {\bar Y}_{j+1:k}^\sPc \|_{(\Sigma^\sPc)^{-1}}, \qquad k \geq m+1, 
\end{align}
which is the version of the detector in~\eqref{eq:Dm:Yi}, in which the estimated long-run variance $\Sigma_m^\sPc$ is replaced by the true long-run variance $\Sigma^\sPc$.

On the probability space of Condition~\ref{cond:H0} (assuming that this condition holds), we may define 
\begin{align}
	\label{eq:bar:Dm}
\bar{D}^\sPc_m\left(k\right)
=
\frac{1}{\sqrt{m}} \max_{j \in \llb m,k \llb} 
&
 \left\Vert  \frac{k}{m} (\Sigma^\sPc)^{\frac12} \lbrace  \bm W_{2,m} (m) + \bm W_{1,m} (j-m)\rbrace \right. \nonumber \\
& \quad
\left.-\frac{j}{m} (\Sigma^\sPc)^{\frac12} \lbrace \bm W_{2,m} (m) + \bm W_{1,m} (k-m) \rbrace \right\Vert_{(\Sigma^\sPc)^{-1}},
\end{align}
where we recall that for each $m \in \N$, $\bm W_{1,m}$ and $\bm W_{2,m}$ are independent $p$-dimensional standard Brownian motions.

\begin{lem}
	\label{lem:bar:D}
Assume that Condition \ref{cond:H0} holds. Then, for any $\eta>0$, 
\begin{align*}
\sup_{k>m}  \Big(\frac mk\Big)^{\frac32+\eta} \left| \tilde D_m^\sPc (k) - \bar D_m^\sPc (k) \right| = o_{\P} (1).
\end{align*}
\end{lem}

\begin{proof}
Fix $\eta>0$. Applying the reverse triangle inequality for suprema 
\begin{equation}
|\sup_{x \in A}f(x) -\sup_{x \in A}g(x)|\le \sup_{x \in A}|f(x)-g(x)|
\label{eq:reverse:triangle}
\end{equation}
to the maximum over $j \in \llb m,k \llb$ below gives
\begin{align*}
\Big| & \tilde D_m^\sPc (k) - \bar{D}_m^\sPc (k) \Big| \\
&\leq
\frac{1}{\sqrt{m}} \max_{j \in \llb m, k \llb} \left| \frac{j (k-j)}{m}  \| \bm {\bar Y}_{1:j}^\sPc - \bm {\bar Y}_{j+1:k}^\sPc \|_{(\Sigma^\sPc)^{-1}}  
 -
\Big\Vert \frac{k}{m} (\Sigma^\sPc)^{\frac12} \lbrace \bm W_{2,m} (m) + \bm W_{1,m} (j-m) \rbrace \right. \\
&\qquad \qquad \qquad \qquad
-
\left.
  \frac{j}{m} (\Sigma^\sPc)^{\frac12} \lbrace \bm W_{2,m} (m) + \bm W_{1,m} (k-m) \rbrace \Big\Vert_{(\Sigma^\sPc)^{-1}}
  \right| \\
& \leq
 \frac{1}{\sqrt{m}}\max_{j \in \llb m, k \llb} 
  \left\Vert \frac{j (k-j)}{m} \{
   \bm {\bar Y}_{1:j}^\sPc - \bm {\bar Y}_{j+1:k}^\sPc \} - 
   \frac{k}{m} (\Sigma^\sPc)^{\frac12} \lbrace \bm W_{2,m} (m) + \bm W_{1,m} (j-m)\rbrace 
   \right.\\
&\qquad\qquad\qquad\qquad
+
\left.  
  \frac{j}{m} (\Sigma^\sPc)^{\frac12} \lbrace \bm W_{2,m} (m) + \bm W_{1,m} (k-m) \rbrace 
  \right\Vert_{(\Sigma^\sPc)^{-1}} =
U_m^\sPc(k).
\end{align*}
Next, we rewrite $j (k-j) \lbrace \bm {\bar Y}_{1:j}^\sPc - \bm {\bar Y}_{j+1:k}^\sPc \rbrace$ as 
\begin{align*}
j (k-j) \left\{ \frac{1}{j} \sum_{i=1}^{j} \bm Y_{i}^\sPc - \frac{1}{k-j} \sum_{i=j+1}^{k}\bm Y_{i}^\sPc \right\} 
&=
(k-j) \sum_{i=1}^{j} \left\{ \bm Y_{i}^\sPc - \Ex (\bm Y_{1}^\sPc) \right\} - j \sum_{i=j+1}^{k} \left\{ \bm Y_{i}^\sPc - \Ex (\bm Y_{1}^\sPc) \right\} \\
&=
k \sum_{i=1}^{j} \left\{ \bm Y_{i}^\sPc - \Ex (\bm Y_{1}^\sPc) \right\} - j \sum_{i=1}^{k}\left\{ \bm Y_{i}^\sPc - \Ex (\bm Y_{1}^\sPc) \right\} \\
&=
k \sum_{i=1}^{m} \left\{ \bm Y_{i}^\sPc - \Ex (\bm Y_{1}^\sPc) \right\} + k \sum_{i=m+1}^{j}\left\{ \bm Y_{i}^\sPc - \Ex (\bm Y_{1}^\sPc) \right\} \\
&\quad
 - 
j\sum_{i=1}^{m} \left\{ \bm Y_{i}^\sPc - \Ex (\bm Y_{1}^\sPc) \right\} - j\sum_{i=m+1}^{k}\left\{\bm Y_{i}^\sPc - \Ex (\bm Y_{1}^\sPc) \right\}.
\end{align*}
Hence, using the triangle inequality, we have that 
\begin{align*}
\sup_{k > m}  \Big(\frac mk\Big)^{\frac32+\eta} \left| \tilde D_m^\sPc (k) - \bar D_m^\sPc (k) \right|
&\leq
\sup_{k > m}  \Big(\frac mk\Big)^{\frac32+\eta} U_m^\sPc (k) \\
&\leq 
I_m^\sPc + (I_m^\sPc)' + J_m^\sPc + (J_m^\sPc)',
\end{align*}
where (using the convention that empty sums are equal to 0) 
\begin{align*}
I_m^\sPc
& =\frac{1}{\sqrt{m}}
\sup_{k>m} \left(\frac{m}{k}\right)^{\frac12+\eta} 
\max_{j \in \llb m,k \llb} \left\| \sum_{i=1}^{m} \left\{ \bm Y_{i}^\sPc - \Ex (\bm Y_{1}^\sPc)\right\} -(\Sigma^\sPc)^{\frac12} \bm W_{2,m} (m) \right\|_{(\Sigma^\sPc)^{-1}}, \\
(I_m^\sPc)'
& =
\frac{1}{\sqrt{m}}
\sup_{k>m} \left(\frac{m}{k}\right)^{\frac12+\eta} 
\max_{j \in \llb m,k \llb} \left\| \sum_{i=m+1}^{j} \left\{ \bm Y_{i}^\sPc - \Ex (\bm Y_{1}^\sPc) \right\} -(\Sigma^\sPc)^{\frac12} \bm W_{1,m} (j-m) \right\|_{(\Sigma^\sPc)^{-1}}, \\
J_m^\sPc
& =
\frac{1}{\sqrt{m}} \sup_{k>m} \left(\frac{m}{k}\right)^{\frac32+\eta} 
\max_{j \in \llb m,k \llb} \frac{j}{m} \left\| \sum_{i=1}^{m} \left\{ \bm Y_{i}^\sPc - \Ex (\bm Y_{1}^\sPc) \right\} - (\Sigma^\sPc)^{\frac12} \bm W_{2,m} (m) \right\|_{(\Sigma^\sPc)^{-1}}, \\
(J_m^\sPc)'
& =
\frac{1}{\sqrt{m}} \sup_{k>m} \left(\frac{m}{k}\right)^{\frac32+\eta} 
\max_{j \in \llb m,k \llb} \frac{j}{m} \left\| \sum_{i=m+1}^{k} \left\{ \bm Y_{i}^\sPc - \Ex (\bm Y_{1}^\sPc)\right\} - (\Sigma^\sPc)^{\frac12} \bm W_{1,m} (k-m) \right\|_{(\Sigma^\sPc)^{-1}}.
\end{align*}
It suffices to show that $I_m^\sPc$, $(I_m^\sPc)'$, $J_m^\sPc$, and $(J_m^\sPc)'$ are $o_\P(1)$.  Since $(m/k)\times (j/m)=j/k\le 1$ when $m\le j<k$, we have that $J_m^\sPc\le I_m^\sPc$ and $(J_m^\sPc)'\le (I_m^\sPc)'$.  So it suffices to consider $I_m^\sPc$ and $(I_m^\sPc)'$. Let $0<\xi<\frac12$ be as in Condition \ref{cond:H0}. Then,
\begin{align*}
I_m^\sPc
&=
\frac{1}{m^\xi} \left\| \sum_{i=1}^{m} \left\{ \bm Y_{i}^\sPc - \Ex (\bm Y_{1}^\sPc) \right\} -(\Sigma^\sPc)^{\frac12} \bm W_{2,m} (m) \right\|_{(\Sigma^\sPc)^{-1}} m^{\xi-\frac12} \sup_{k>m} \Big(\frac{m}{k}\Big)^{\frac12+\eta} \\
&\leq
\frac{1}{m^\xi} \left\| \sum_{i=1}^{m} \left\{ \bm Y_{i,m}^\sPc - \Ex (\bm Y_{1}^\sPc)\right\} -(\Sigma^\sPc)^{\frac12} \bm W_{2,m} (m) \right\|_{(\Sigma^\sPc)^{-1}} m^{\xi-\frac12}.
\end{align*}
Hence, by equivalence of norms on $\R^p$,  \eqref{eq:cond} and the fact that $\xi<\frac12$, $I_m^\sPc$ converges in probability to zero as $m \to \infty$. Next, since the norm in $(I_m^\sPc)'$ is zero when $j=m$ and using the fact that $j-m<k-m<k$, we see that $(I_m^\sPc)'$ is equal to
\begin{align*}
&\frac{1}{\sqrt{m}}
\sup_{k>m} \left(\frac{m}{k}\right)^{\frac12+\eta} 
\max_{j \in \rrb m, k \llb} \left\| \sum_{i=m+1}^{j} \left\{ \bm Y_{i}^\sPc - \Ex (\bm Y_{1}^\sPc) \right\} -(\Sigma^\sPc)^{\frac12} \bm W_{1,m} (j-m) \right\|_{(\Sigma^\sPc)^{-1}} \\
&\le \frac{1}{m^{\frac12-\xi}}
\sup_{k>m} \left(\frac{m}{k}\right)^{\frac12-\xi+\eta} 
\max_{j \in \rrb m, k \llb} \frac{1}{(j-m)^\xi}
\left\| \sum_{i=m+1}^{j} \left\{ \bm Y_{i}^\sPc - \Ex (\bm Y_{1}^\sPc) \right\} -(\Sigma^\sPc)^{\frac12} \bm W_{1,m} (j-m) \right\|_{(\Sigma^\sPc)^{-1}}\\
&\le \frac{1}{m^{\frac12-\xi}}
\sup_{j>m} \frac{1}{(j-m)^\xi}
\left\| \sum_{i=m+1}^{j} \left\{ \bm Y_{i}^\sPc - \Ex (\bm Y_{1}^\sPc) \right\} -(\Sigma^\sPc)^{\frac12} \bm W_{1,m} (j-m) \right\|_{(\Sigma^\sPc)^{-1}} \sup_{k>m} \left(\frac{m}{k}\right)^{\frac12-\xi+\eta}.
\end{align*}
The supremum over $k$ above is less than 1 since $\xi<\frac12$.  The supremum over $j$ is bounded in probability by \eqref{eq:sup:cond} and equivalence of norms.  Since $m^{-\frac12+\xi}\to 0$, this proves that $(I_m^\sPc)'$ converges to 0 in probability.
\end{proof}

Given $p$-dimensional independent standard Brownian motions $\bm W_1$ and $\bm W_2$ and $\eta>0$, define the random function $D_\eta$ by 
\begin{align}
  \label{eq:D:eta}
D_\eta \left(s,t\right)
&=
t^{-\frac32-\eta} \left\Vert (t-s) \bm W_2 (1) + t \bm W_1 (s-1) - s \bm W_1 (t-1) \right\Vert_{I_p},\quad 1\leq s\leq t < \infty.
\end{align}

\begin{lem}
For any $\eta>0$, the function $D_\eta$ is almost surely bounded and uniformly continuous.
\end{lem}

\begin{proof}
Fix $\eta>0$.  For $i \in \llb 1, p \rrb$ and $1\leq s\leq t < \infty$, define $A^{\sss[i]} \left(s,t\right)$ to be the $i$-th coordinate of 
\begin{align*}
\bm A \left(s,t\right)
&=
t^{-\frac32-\eta}  \left\{ (t-s) \bm W_2 (1) + t \bm W_1 (s-1) - s \bm W_1 (t-1) \right\},\quad 1\leq s\leq t < \infty,
\end{align*}
and note that $D_{\eta}(s,t)=\|\bm A\left(s,t\right)\|_{I_p}$. From Lemma A.2 in~\cite{HolKoj21}, we know that $|A^{\sss[i]}|$ is almost surely bounded and uniformly continuous for each $i$.  It follows that each $A^{\sss[i]}$ is also bounded and uniformly continuous, almost surely.  Now note that 
\begin{align*}
|D_{\eta}(s,t)-D_{\eta}(s',t')|=\Big|\|\bm A \left(s,t\right)\|_{I_p}-\|\bm A \left(s',t'\right)\|_{I_p}\Big|&\le \|\bm A\left(s,t\right)-\bm A \left(s',t'\right)\|_{I_p}\\
&=\sqrt{\frac{1}{p}\sum_{i=1}^p \{ A^{\sss[i]}(s,t)-A^{\sss[i]}(s',t') \}^2}.
\end{align*}  
So $D_\eta$ is almost surely bounded and uniformly continuous since each $A^{\sss[i]}$ is.
\end{proof}

\begin{lem}
\label{lem:bar:D-to-D:eta}
For any $\eta>0$,
\begin{align*}
\sup_{k>m}  \Big(\frac mk\Big)^{\frac32+\eta}\bar{D}_m^\sPc
(k) \leadsto \sup_{1\leq s \leq t < \infty} D_\eta (s,t),
\end{align*}
where $\bar{D}_m^\sPc$ is defined in~\eqref{eq:bar:Dm} and  $D_\eta$ is defined in~\eqref{eq:D:eta}.
\end{lem}

\begin{proof}
Fix $\eta>0$. The random variable $\sup_{k>m} (m/k)^{\frac32+\eta} \bar D_m^\sPc (k)$ is equal in distribution to
\begin{align*}
& \frac{1}{\sqrt{m}}\sup_{k\ge m}  \Big(\frac mk\Big)^{\frac32+\eta} \max_{j \in \llb m, k \rrb}
\left\Vert \frac{k}{m} \{ \bm W_{2} (m) + \bm W_{1} (j-m)\} -\frac{j}{m} \{  \bm W_{2} (m) + \bm W_{1} (k-m) \} \right\Vert_{I_p},
\end{align*}
where we note that the norm is equal to zero when  $j=k$, which has allowed us to include the cases $j=k$ and $k=m$ in the maximum and supremum, respectively.  Using Brownian scaling, this is equal in distribution to 
\begin{align*}
\sup_{k\ge m} \Big(\frac mk\Big)^{\frac32+\eta} \max_{j \in \llb m, k \rrb}
\left\Vert \frac{k}{m} \Big\{ \bm W_{2} (1) + \bm W_{1} \Big(\frac{j}{m}-1\Big)\Big\} -\frac{j}{m} \Big\{  \bm W_{2} (1) + \bm W_{1} \Big(\frac{k}{m}-1\Big) \Big\} \right\Vert_{I_p}.
\end{align*}
In the above, $j$ and $k$ are integers.  Letting $k = \floor{mt}$ and $j = \floor{ms}$ (where $s,t\in\R$ and $1\leq s \le t < \infty$), the above display becomes  
\begin{align*}
&
\sup_{t \ge 1} \Big(\frac{m}{\floor{mt}}\Big)^{\frac32+\eta}
\\
& \times \sup_{s \in [1,t]}
\left\| \frac{\floor{mt}}{m} \Big\{ \bm W_{2} (1) + \bm W_{1} \Big(\frac{\floor{ms}}{m}-1\Big)\Big\} -\frac{\floor{ms}}{m} \Big\{  \bm W_{2} (1) + \bm W_{1} \Big(\frac{\floor{mt}}{m}-1\Big) \Big\} \right\|_{I_p} \\
&=\sup_{1 \leq s \le  t < \infty}  D_\eta \left( \frac{\floor{ms}}{m}, \frac{\floor{mt}}{m} \right).
\end{align*}
Using \eqref{eq:reverse:triangle} with the supremum over $(s,t)$, we have
\begin{align*}
 \left| \sup_{1 \leq s \le  t < \infty} D_\eta \left( \frac{\floor{ms}}{m}, \frac{\floor{mt}}{m} \right)- \sup_{1 \leq s \leq t < \infty} D_\eta (s, t)  \right| 
& \leq 
\sup_{1 \leq s \leq t < \infty} \left| D_\eta \left( \frac{\floor{ms}}{m}, \frac{\floor{mt}}{m} \right) -  D_\eta \left(s, t \right)  \right|,
\end{align*}
which converges to zero almost surely as $m \to \infty$ since $D_\eta$ is almost surely uniformly continuous and $t-1/m< \frac{\floor{mt}}{m}\le t$.
\end{proof}

For a $p \times p$ matrix $A$, denote the operator norm of $A$ by 
\begin{equation}
  \label{eq:operator:norm}
  \|A\|_{\textrm{op}} = \inf \{ c \geq 0 : \| A \bm v \|_2 \leq c \|\bm v \|_2,\text{ for all } \bm v\in \R^p \}.
\end{equation}

\begin{lem}
\label{lem:tilde:D-to-D}
Assume that Condition~\ref{cond:H0} holds and that $\Sigma_m^\sPc \p \Sigma^\sPc$. Then, for any $\eta>0$,  
\begin{equation*} 
  \label{eq:tilde:D-to-D}
  \sup_{k > m} \Big( \frac m k \Big)^{\frac32+\eta} | \tilde D_m^\sPc (k) - D_m^\sPc (k) | = o_\P(1),
\end{equation*}
where $\tilde D_m^\sPc$ is defined in~\eqref{eq:tilde:Dm} and $D_m^\sPc$ is defined in~\eqref{eq:Dm}.
\end{lem}

\begin{proof}
Fix $\eta>0$. Applying~\eqref{eq:reverse:triangle} to the maximum over $j \in \llb m,k \llb$ and using the inequality $| a^{\frac12} - b^{\frac12}| \le |a - b|^{\frac12}$ for $a,b\geq0$, we obtain that
\begin{align} 
\sup_{k > m} &\Big(\frac m k\Big)^{\frac32+\eta} | \tilde D_m^\sPc (k) - D_m^\sPc (k) | \nonumber \\
&\leq
\sup_{k > m} \Big(\frac m k\Big)^{\frac32+\eta} \max_{j \in \llb m,k \llb} \frac{j (k-j)}{m^{\frac32}} \left|  \| \bm {\bar Y}_{1:j}^\sPc - \bm {\bar Y}_{j+1:k}^\sPc \|_{(\Sigma^\sPc)^{-1}} - \| \bm {\bar Y}_{1:j}^\sPc - \bm {\bar Y}_{j+1:k}^\sPc \|_{(\Sigma_m^\sPc)^{-1}}\right| \nonumber \\
&\leq
\sup_{k > m} \Big(\frac m k\Big)^{\frac32+\eta} \label{eq:tilde:D-to-D-bound-i} 
\max_{j \in \llb m,k \llb} \frac{j (k-j)}{p^{\frac12} m^{\frac32}} \left| \left( \bm {\bar Y}_{1:j}^\sPc - \bm {\bar Y}_{j+1:k}^\sPc \right)^{\top} \left((\Sigma^\sPc)^{-1} - (\Sigma_m^\sPc)^{-1} \right) \left( \bm {\bar Y}_{1:j}^\sPc - \bm {\bar Y}_{j+1:k}^\sPc \right) \right|^{\frac12}. 
\end{align}
For $A \in \R^{p\times p}$ and $\bm v\in \R^p$, we have that $| \bm v^\top A \bm v | \leq \| A \bm v \|_2 \| \bm v \|_2$ by the Cauchy-Schwarz inequality. By \eqref{eq:operator:norm}, $\|A\bm v \|_2 \leq \|A\|_{\textrm{op}} \|\bm v\|_2$,  so that $| {\bm v}^\top A \bm v |^{\frac12} \leq \|A\|_{\textrm{op}}^{\frac12} \|\bm v \|_2$. From~\eqref{eq:tilde:D-to-D-bound-i}, we therefore obtain
\begin{multline*}
  \sup_{k > m} \Big(\frac m k\Big)^{\frac32+\eta} | \tilde D_m^\sPc (k) - D_m^\sPc (k) |  \\
  \leq \frac{1}{p^{\frac12}} \left\| (\Sigma^\sPc)^{-1} - (\Sigma_m^\sPc)^{-1} \right\|_{\textrm{op}}^{\frac12} \sup_{k > m} \left(\frac m k \right)^{\frac32+\eta} \max_{j \in \llb m,k \llb} \frac{j (k-j)}{m^{\frac32}} \left\| \bm {\bar Y}_{1:j}^\sPc - \bm {\bar Y}_{j+1:k}^\sPc \right\|_2.
\end{multline*}
It is well-known that the mapping that maps an invertible square matrix to its inverse is continuous. Since $\Sigma_m^\sPc \p \Sigma^\sPc$ and Condition~\ref{cond:sigma} holds, the continuous mapping theorem immediately implies that $(\Sigma_m^\sPc)^{-1} \p (\Sigma^\sPc)^{-1}$, which in turn implies that $\left\| (\Sigma^\sPc)^{-1} - (\Sigma_m^\sPc)^{-1} \right\|_{\textrm{op}} = o_\P(1)$ by equivalence of norms. Furthermore, from Lemmas~\ref{lem:bar:D} and~\ref{lem:bar:D-to-D:eta}, we have that $\sup_{k > m} (m/k)^{\frac32+\eta} \tilde D_m^\sPc(k)$ converges weakly, which implies that it is bounded in probability. By equivalence of norms on $\R^p$, we immediately obtain that
$$
\sup_{k > m} \left(\frac m k \right)^{\frac32+\eta} \max_{j \in \llb m,k \llb} \frac{j (k-j)}{m^{\frac32}} \left\| \bm {\bar Y}_{1:j}^\sPc - \bm {\bar Y}_{j+1:k}^\sPc \right\|_2 = O_\P(1)
$$
and therefore the desired result.
\end{proof}

\begin{proof}[\bf Proof of Theorem \ref{thm:H0}]
From Lemmas~\ref{lem:bar:D}--\ref{lem:tilde:D-to-D}, we have that  
\begin{align*}
\sup_{k > m} \left(\frac m k \right)^{\frac32+\eta}  D_m^\sPc(k) \leadsto \sup_{1\leq s \leq t < \infty} D_\eta\left( s,t \right),
\end{align*}
where $D_\eta$ is defined in~\eqref{eq:D:eta}. It remains to be shown that 
\begin{equation}
\label{eq:final}
\sup_{1\leq s \leq t < \infty} D_\eta\left( s,t \right) 
\d
\sup_{1 \leq s \leq t < \infty} t^{-\frac32-\eta} \| t \bm W(s) - s \bm W(t) \|_{I_p}.
\end{equation}
Let $\bm U_p$ and $\bm V_p$ be two $p$-dimensional Gaussian processes defined, for any $1\leq s \leq t$, by
\begin{align*}
\bm U_p\left(s,t\right)
&= \left(t-s\right)\bm W_2\left(1\right) + t\bm W_1\left(s-1\right) -s \bm W_1\left(t-1\right), \\
\bm V_p\left(s,t\right)
&=
t \bm W(s) - s \bm W(t).
\end{align*}
Since $\bm W_1$, $\bm W_2$, and $\bm W$ are $p$-dimensional standard Brownian motions, it follows that the coordinates of the Gaussian processes $\bm U_p$ and $\bm V_p$ are centered. Thus, to show that the random functions $\bm U_p$ and $\bm V_p$ are equal in distribution (which will immediately imply~\eqref{eq:final}), it suffices to establish equality of their covariance functions at any $(s,t,s',t')$ with $1\leq s \leq t$ and $1\leq s' \leq t'$. On one hand, the covariance function of the Gaussian process $\bm U_p$ at $(s,t,s',t')$ is
\begin{align*}
\Ex \{ \bm U_{p} (s,t) \bm U_{p} (s',t')^{\top} \}
&=
\Ex \big[ \big( \left(t-s\right)\bm W_{2}\left(1\right)+t \bm W_{1} (s-1) - s \bm W_{1}(t-1) \big) \\
&\quad
\times \big( (t' - s') \bm W_{2} (1) + t'\bm W_{1} (s'-1) - s' \bm W_{1} (t'-1) \big)^{\top} \big] \\
&=
\big[ (t-s) (t'-s') + t' t \{\min (s,s')-1 \} - s' t \{\min (s,t') - 1 \} \\
&\quad
-t's \{\min (t,s')-1\}+s's \{\min (t,t')-1\} \big] I_p\\
&=
\left\{ t' t \min (s,s') - s't \min (s,t') - t' s \min (t,s') + s's \min (t,t') \right\} I_p,
\end{align*}
while, on the other hand, the covariance function of the Gaussian process $\bm V_p$  at $(s,t,s',t')$ is
\begin{align*}
\Ex \{ \bm V (s,t) \bm V (s',t')^{\top} \}
&=
\Ex \big[\big( t \bm W (s) - s \bm W (t) \big) \big( t' \bm W(s') - s' \bm W (t') \big)^{\top}\big]\\
&=
\{ tt' \min (s,s') - s' t \min (s,t') - st' \min (t,s') + ss' \min (t,t') \} I_p,
\end{align*}
which concludes the proof.
\end{proof}


\section{Proofs of Propositions~\ref{prop:cond:H0},~\ref{prop:eta:0} and~\ref{prop:L:continuous}}
\label{proof:prop:H0}

\begin{proof}[\bf Proof of Proposition~\ref{prop:cond:H0}] Let $\Pc \in (\R^d)^p$ be  such that Condition~\ref{cond:sigma} holds.  Since $\bm Y_i^\sPc = \big( \1(\bm X_i \leq \bm x_1), \dots, \1(\bm X_i \leq \bm x_p) \big)$, it is immediate that $\alpha_r^{\bm Y}\le \alpha_r^{\bm X}= O(r^{-a})$ as $r \to \infty$. Furthermore, as all the components of the $\bm Y_i^\sPc$ are bounded in absolute value by 1, from Theorem~4 of \cite{KuePhi80}, we can redefine the sequence $( \bm Y_i^\sPc )_{i \in \N}$ on a new probability space together with a $p$-dimensional standard Brownian motion $\bm W$ such that, almost surely,
$$
\left\| \sum_{i=1}^m \{\bm Y_i^\sPc - \Ex(\bm Y_1^\sPc)\} - (\Sigma^{\sPc})^{\frac12} \bm W(m) \right\|_2 = O(m^{\frac12 - \lambda}),
$$
for some $\lambda \in (0,\frac12)$ that depends on $a$, $p$ and $\Pc$. Let $\xi \in (\frac12 - \lambda, \frac12)$. Then,  almost surely,
\begin{equation}
  \label{eq:as}
  \lim_{m \to \infty} \frac{1}{m^\xi} \left\| \sum_{i=1}^m \{\bm Y_i^\sPc - \Ex(\bm Y_1^\sPc)\} - (\Sigma^{\sPc})^{\frac12} \bm W(m)  \right\|_2 = 0.
\end{equation}
Let $\bm W'$ be another $p$-dimensional standard Brownian motion independent of $\bm W$ and define $\bm W^{(m)}_2$ as
$$
\bm W^{(m)}_2(s) =
\begin{cases}
\bm W(s) \text{ if } s \in [0, m], \\
\bm W'(s-m) + \bm W(m) \text{ otherwise,}
\end{cases}
$$
and $\bm W^{(m)}_1$ as $\bm W^{(m)}_1(s) = \bm W(m+s) - \bm W(m)$, $s \geq 0$. Then, for each $m \in \N$, $\bm W^{(m)}_1$ and $\bm W^{(m)}_2$ are independent $p$-dimensional standard Brownian motions.

For each $m \geq 0$, let
$$
V_m = \sup_{k > m} \frac{1}{(k-m)^\xi} \left\| \sum_{i=m+1}^{k} \{\bm Y_i^\sPc - \Ex(\bm  Y_1^\sPc) \}  - (\Sigma^{\sPc})^{\frac12} \bm W^{(m)}_{1}(k-m) \right\|_2
$$
and note that the sequence $(V_m)_{m \geq 0}$ consists of identically distributed random variables. Therefore, to show that~\eqref{eq:sup:cond} holds, it is sufficient to show that $V_0 = O_\P(1)$. From the previous definition, $V_0$ can be rewritten as
$$
V_0 = \sup_{k > 0} \frac{1}{k^\xi} \left\| \sum_{i=1}^{k} \{\bm Y_i^\sPc - \Ex(\bm  Y_1^\sPc) \}  - (\Sigma^{\sPc})^{\frac12} \bm W(k) \right\|_2,
$$
and, from~\eqref{eq:as}, it is the supremum of an almost surely convergent sequence. It follows that $V_0$ is an almost surely finite  random variable, so $V_0 = O_\P(1)$ and therefore~\eqref{eq:sup:cond} holds. Finally, \eqref{eq:as} and the definition of $\bm W^{(m)}_2$ immediately imply~\eqref{eq:cond}.
\end{proof}

\begin{proof}[\bf Proof of Proposition~\ref{prop:eta:0}]
It is immediate that
\begin{align*}
  \sqrt{p} \sup_{1 \leq s \leq t < \infty} t^{-\frac32} \| t \bm W(s) - s \bm W(t) \|_{I_p}  \geq  \sup_{1 \leq s \leq t < \infty}  t^{-\frac32} |t W^{\sss [1]}(s) - s W^{\sss [1]}(t)|,
\end{align*}
where $W^{\sss [1]}$ is the first component of the $p$-dimensional Brownian motion $\bm W$. Hence, for any fixed $M > 0$,
\begin{multline*}
\P \Big( \sup_{1 \leq s \leq t < \infty} t^{-\frac32} \| t \bm W(s) - s \bm W(t) \|_{I_p} \geq M \Big) \\ \geq \P \Big( p^{-1/2} \sup_{1 \leq s \leq t < \infty}  t^{-\frac32} |t W^{\sss [1]}(s) - s W^{\sss [1]}(t)| \geq M \Big) = 1,
\end{multline*}
where the last equality is a consequence of Proposition~3.4 of \cite{HolKoj21}.
\end{proof}

\begin{proof}[\bf Proof of Proposition~\ref{prop:L:continuous}]
Fix $\eta>0$ and $p\in \N$.  Note first that $\bm W$ is a $p$-variate continuous Gaussian process, which implies that $\bm W \in C([0,\infty),\R^p)$ almost surely. For $v\ge 2$, let $f_{v}:C([0,v],\R^p)\to [0,\infty)$ be defined by 
\begin{align*}
f_v(\bm w)=\sup_{1\le s\le t\le v}t^{-\frac32-\eta}\|t\bm w(s)-s\bm w(t)\|_{I_p}, \qquad \bm w \in C([0,v],\R^p),
\end{align*}
and, similarly, let $f:C([0,\infty),\R^p) \to [0,\infty)$ be defined by
\begin{align*}
f(\bm w)=\sup_{1\le s\le t<\infty}t^{-\frac32-\eta}\|t\bm w(s)-s\bm w(t)\|_{I_p}, \qquad \bm w \in C([0,\infty),\R^p).
\end{align*}
Notice that $\Lc_{p,\eta}=f(\bm{W})$ and that $f_v(\bm w)\le f(\bm w)$ for every $v\ge 2$.

Next, for any $v \geq 2$, we equip $C([0,v],\R^p)$ with the uniform distance, which is then a separable (and locally convex) metric space. Furthermore, as we shall verify below, $f_v$ is continuous and convex for any $v \geq 2$. We can then apply Theorem~7.1 of \cite{DavLif84} to obtain that, for any $v \geq 2$, the distribution of $f_v(\bm W)$ is concentrated on $[0,\infty)$ and absolutely continuous on $(0,\infty)$. In addition, some thought reveals that, for any $v \geq 2$, 
$$
\P(f_v(\bm{W})=0) \leq \P \big( f_2(\bm{W})=0) \leq \P(\|2\bm{W}(1)-1\bm{W}(2)\|_{I_p}=0 \big) = 0
$$
since $2\bm{W}(1)-1\bm{W}(2)$ is a centered multivariate normal random vector with covariance matrix $2I_p$. Hence, for any $v \geq 2$, the distribution of $f_v(\bm W)$ has no atom at 0 and is therefore absolutely continuous.  

\emph{Proof of the continuity of $f_v$:} To show (uniform) continuity on $C([0,v],\R^p)$ (equipped with the uniform distance), let $\vep>0$ be given and let $\delta=\vep/3$. Let $\bm w,\bm w'\in C([0,v],\R^p)$ be such that $\sup_{0 \leq t\le v} \| \bm w(t)-\bm w'(t) \|_{I_p} <\delta$. Then,
\begin{align*}
f_v(\bm w')&=\sup_{1\le s\le t\le v} t^{-\frac32-\eta}\|t\bm w'(s)-s\bm w'(t)\|_{I_p}\\
&=\sup_{1\le s\le t\le v} t^{-\frac32-\eta}\| \{ t\bm w(s)-s\bm w(t) \} + t \{ \bm w'(s)-\bm w(s) \} - s \{ \bm w'(t)-\bm w(t) \} \|_{I_p}\\
&\le \sup_{1\le s\le t\le v} t^{-\frac32-\eta}\Big\{ \| t\bm w(s)-s\bm w(t) \|_{I_p} + t\|\bm w'(s)-\bm w(s)\|_{I_p}+s\|\bm w'(t)-\bm w(t)\|_{I_p}\Big\}\\
&\le \sup_{1\le s\le t\le v} t^{-\frac32-\eta}\| t\bm w(s)-s\bm w(t) \|_{I_p}+2 \delta.
\end{align*}
This shows that $f_v(\bm w')-f_v(\bm w)\le 2\delta <\vep$.  Similarly (or just by symmetry of the above argument), $f_v(\bm w)-f_v(\bm w')<\vep$.  Thus, $|f_v(\bm w)-f_v(\bm w')|<\vep$ as required.

\emph{Proof of the convexity of $f_v$:} To show convexity, let $\lambda \in (0,1)$ and $\bm w, \bm w'\in C([0,v],\R^p)$.  Then $\lambda \bm w + (1-\lambda) \bm w'\in C([0,v],\R^p)$ and 
\begin{align*}
f_v(\lambda\bm w +(1-\lambda)\bm w')&=\sup_{1\le s\le t\le v} t^{-\frac32-\eta} \Big\|t\{ \lambda\bm w(s) +(1-\lambda)\bm w'(s) \} -s \{ \lambda\bm w(t) +(1-\lambda)\bm w'(t) \} \Big\|_{I_p}\\
&=\sup_{1\le s\le t\le v} t^{-\frac32-\eta} \Big\| \lambda \{ t \bm w(s) -s\bm w(t)\} + (1-\lambda) \{ t\bm w'(s)-s \bm w'(t) \} \Big\|_{I_p}\\
&\le \lambda f_v(\bm w)+ (1-\lambda) f_v(\bm w'),
\end{align*}
as required.

To complete the proof of the proposition, it suffices to fix $r,\vep>0$ and show that $\P(f(\bm W)=r)<\vep$. Notice first that, for any $u \geq 1$, almost surely, 
$$
\sup_{(s,t):1 \leq s \leq t, u\le t<\infty} t^{-\frac32-\eta} \| t \bm W(s) - s \bm W(t) \|_{I_p} \leq u^{-\frac{\eta}{2}} \sup_{1 \leq s \leq t <\infty} t^{-\frac32-\frac{\eta}{2}} \| t \bm W(s) - s \bm W(t) \|_{I_p} .
$$
Since according to Theorem~\ref{thm:H0}, $\Lc_{p,\frac{\eta}{2}}$ is almost surely finite, we see that 
\[
  \sup_{(s,t):1 \leq s \leq t, u\le t<\infty} t^{-\frac32-\eta} \| t \bm W(s) - s \bm W(t) \|_{I_p} \as 0, \qquad \text{  as }u \to \infty.
\]
Hence, there exists $u_0<\infty$ (depending on $r,\vep$) such that 
\begin{align}
  \label{eq:bound}
  \P\Bigg( \sup_{(s,t):1 \leq s \leq t, u_0\le t<\infty} t^{-\frac32-\eta} \| t \bm W(s) - s \bm W(t) \|_{I_p}\ge r/2 \Bigg)<\vep.
\end{align}
Now, if $f(\bm W)=r$, then either $f_{u_0}(\bm W)=r$ or  the supremum ($r$) in the definition of $f$ is attained for some $s,t$ with $t>u_0$.  Absolute continuity of $f_{u_0}(\bm W)$ shows that the former has probability 0, while  \eqref{eq:bound} shows that the latter has probability at most $\vep$. 
\end{proof}


\section{Proof of Proposition~\ref{prop:H1}}
\label{proof:prop:H1}

\begin{proof}[\bf Proof of Proposition~\ref{prop:H1}]
Let $A$ be a $p \times p$ symmetric positive-definite matrix. Then, there exits a diagonal matrix $\Delta$ with $\lambda_1 > 0$,\dots,$\lambda_p > 0$ on its diagonal and a $p \times p$  orthogonal matrix $P$ such that the columns of $P$ are the eigenvectors $\bm e_1, \dots, \bm e_p \in \R^p$ of $A$ with corresponding eigenvalues $\lambda_1,\dots,\lambda_p$ and $A = P \Delta P^\top$. Starting from the definition of $\| \cdot \|_{A^{-1}}$ given below~\eqref{eq:sigma}, for any $\bm v \in \R^p$, 
\begin{align*}
 \sqrt{p} \| \bm v \|_{A^{-1}} &=  | \bm v^\top A^{-1} \bm v|^{\frac12} = | \bm v^\top ( P \Delta P^\top )^{-1} \bm v|^{\frac12} = | \bm v^\top P \Delta^{-1} P^\top \bm v |^{\frac12} \\
                       &\geq \min_{i \in \llb 1, p \rrb} \lambda_i^{-\frac12}   \times \| P^\top \bm v \|_2 = \min_{i \in \llb 1, p \rrb} \lambda_i^{-\frac12}   \times \sqrt{ (P^\top \bm v) \cdot (P^\top \bm v)} \\
                       &= \min_{i \in \llb 1, p \rrb} \lambda_i^{-\frac12} \times \sqrt{ \bm v \cdot \bm v} = \min_{i \in \llb 1, p \rrb} \lambda_i^{-\frac12} \times \|\bm v\|_2 \\
  &\geq \min_{i \in \llb 1, p \rrb} \lambda_i^{-\frac12}  \times \|\bm v\|_\infty,
\end{align*}
where we have used the fact that orthogonal matrices preserve the dot product. 

Let $\psi$ be the map from the set $\Sc_p$ of $p \times p$ symmetric positive-definite matrices to $(0,\infty)$ such that, for any $A \in \Sc_p$, $\psi(A) = \min_{i \in \llb 1, p \rrb} \lambda_i^{-\frac12}$. Since eigenvalue decomposition is a continuous operation and since the minimum is a continuous function, the map $\psi$ is continuous. Hence, $\Sigma_m^\sPc \p \Sigma^\sPc$ and the continuous mapping theorem imply that $\psi(\Sigma_m^\sPc) \p \psi(\Sigma^\sPc) > 0$.

From the previous derivations and the assumptions of the proposition, we thus have that, for all $m \in \N$, $\| \cdot \|_{(\Sigma_m^\sPc)^{-1}} \geq p^{-\frac12} \psi(\Sigma_m^\sPc) \| \cdot \|_\infty$ almost surely. Therefore, for any $m \in \N$ and $k \geq m+1$, almost surely, 
\begin{align*}
  D_m^\sPc(k) &= \max_{j \in \llb m,k \llb} \frac{j (k-j)}{m^{\frac32}}  \| \bm {\bar Y}^\sPc_{1:j} - \bm {\bar Y}^\sPc_{j+1:k} \|_{(\Sigma_m^\sPc)^{-1}}  \\
              &\geq p^{-\frac12} \psi(\Sigma_m^\sPc) \max_{j \in \llb m,k \llb} \frac{j (k-j)}{m^{\frac32}}  \| \bm {\bar Y}^\sPc_{1:j} - \bm {\bar Y}^\sPc_{j+1:k} \|_{\infty} \geq p^{-\frac12} \psi(\Sigma_m^\sPc) E_m^{\bm x_\ell}(k),
\end{align*}
which immediately implies that $\sup_{k>m} (m/k)^{\frac32+\eta} D_m^\sPc(k) \p \infty$ since $\psi(\Sigma_m^\sPc) \p \psi(\Sigma^\sPc) > 0$ and $\sup_{k>m} (m/k)^{\frac32+\eta}  E_m^{\bm x_\ell}(k) \p \infty$.
\end{proof}


\section{Proofs of Propositions~\ref{prop:mult:point}, \ref{prop:cond:points}, and~\ref{prop:equiv:points}}
\label{proofs:prop:points}

\begin{proof}[\bf Proof of Proposition~\ref{prop:mult:point}]
Let $\bm a, \bm b \in [0,1]^d$ such that $\bm a < \bm b$. Note that
\begin{equation*}
  \nu_C \big( (\bm{a},\bm{b}] \big)  = \sum_{\bm{i}\in\{0,1\}^d}(-1)^{\sum_{j=1}^d i_j}C \bigl( a_1^{i_1}b_1^{1-i_1},\dots,a_d^{i_d}b_d^{1-i_d} \bigr),
\end{equation*}
and
\begin{equation*}
  \nu_m \big( (\bm{a},\bm{b}] \big)  = \sum_{\bm{i}\in\{0,1\}^d}(-1)^{\sum_{j=1}^d i_j}C_m \bigl( a_1^{i_1}b_1^{1-i_1},\dots,a_d^{i_d}b_d^{1-i_d} \bigr).
\end{equation*}
By Condition~\ref{cond:mult:point}~(ii), we immediately see that, for each $\bm \pi \in \Pi$ (where $\Pi$ is defined in~\eqref{eq:unif:grid}),  $\nu_m \big( (\bm \pi - \bm s, \bm \pi ]\big) \as \nu_C \big( (\bm \pi - \bm s, \bm \pi ] \big)$ as $m \to \infty$. By Condition~\ref{cond:mult:point}~(i), for each $\bm \pi \in \Pi$, $\nu_C\big( (\bm \pi - \bm s, \bm \pi ] \big) \ne 1/(\kappa (r+1)^d)$. We conclude that, almost surely, for each $\bm \pi \in \Pi$,
$$
\nu_m \big( (\bm \pi - \bm s, \bm \pi ] \big) > 1/(\kappa (r+1)^d) \text{ for all $m$ sufficiently large} \iff \nu_C \big( (\bm \pi - \bm s, \bm \pi ] \big) > 1/(\kappa (r+1)^d).
$$
This completes the proof.
\end{proof}

\begin{proof}[\bf Proof of Proposition~\ref{prop:cond:points}]
Let
$$
R(n,\bm x) = \sum_{i=1}^n \{ \1(\bm X_i \leq \bm x) - F(\bm x) \}, \qquad n \in \N, \bm x \in \R^d.
$$
From Theorem~3.1 part~2(b) in \cite{DedMerRio14}, without changing its distribution, the empirical process $R$ can be redefined on a richer probability space on which there exists a \emph{Kiefer process}, that is, a two-parameter centered continuous Gaussian process $K$ with covariance function given by~\eqref{eq:Gamma}, and a random variable $C>0$ such that, almost surely (a.s.), 
\begin{equation}
  \label{eq:strong:approx}
  \sup_{t \in [0,1]} \sup_{\bm x \in \R^d} |R(\ip{nt}, \bm x) - K(\ip{nt}, \bm x)| \leq C n^{\frac12 - \lambda}, \qquad \text{ for all }n \in \N
\end{equation}
for some $\lambda \in (0,1/2)$ only depending on $d$ and $a$. Let $i \in \llb 1, p \rrb$. Then,
\begin{align*}
  \sup_{k>m} & \, k^{-\frac12} \max_{j \in \llb 1, k \rrb} j | F_{1:j}(\bm \Xc_{i,m}) - F(\bm \Xc_{i,m}) -  F_{1:j}(\bm x_i) +  F(\bm x_i) |  \\
                          =& \sup_{k>m} k^{-\frac12} \max_{j \in \llb 1, k \rrb} | R(j, \bm \Xc_{i,m}) - R(j, \bm x_i) | \\
  \leq& \sup_{k>m} k^{-\frac12} \max_{j \in \llb 1, k \rrb} | R(j, \bm \Xc_{i,m}) - K(j, \bm \Xc_{i,m}) |  + \sup_{k>m} k^{-\frac12} \max_{j \in \llb 1, k \rrb} |  R(j, \bm x_i) -  K(j, \bm x_i) | \\
             &+ \sup_{k>m} k^{-\frac12} \max_{j \in \llb 1, k \rrb} | K(j, \bm \Xc_{i,m}) - K(j, \bm x_i) |.
\end{align*}
From~\eqref{eq:strong:approx}, with probability one, the first two terms on the right-hand side of the last display are smaller than
$$
2 \sup_{k>m} k^{-\frac12} \sup_{t \in [0,1]} \sup_{\bm x \in \R^d} | R(\ip{kt}, \bm x) - K(\ip{kt}, \bm x) | \leq 2 \sup_{k>m} k^{-\frac12} C k^{\frac12 - \lambda} \leq 2 C m^{- \lambda} \as 0.
$$
The third term is smaller than
\begin{align}
  \label{eq:third:term}
  I_m &= \sup_{k>m} k^{-\frac12} \max_{j \in \llb 1, k \rrb} \sup_{\substack{\bm x, \bm y \in \R^d \\ \|\bm x - \bm y\|_\infty \leq \Delta_m}}  | K(j, \bm x) - K(j, \bm y) |,
\end{align}
where $\Delta_m = \| \bm \Xc_{i,m} - \bm x_i \|_\infty$. To complete the proof, it thus suffices to show that $I_m \p 0$.

From Theorem~3.1 part~2(a) in \cite{DedMerRio14}, we know that the sample paths of the Gaussian process $K$ are almost surely uniformly continuous with respect to the pseudo-metric $\rho$ on $[0,\infty) \times \R^d$ defined by
$$
\rho \big( (s,\bm x), (t, \bm y) \big) = |s - t| + \sum_{\ell=1}^d |F^{\sss[\ell]}(x^{\sss[\ell]}) - F^{\sss[\ell]}(y^{\sss[\ell]})|, \qquad s,t \in [0,\infty), \bm x, \bm y \in \R^d.
$$
We will use this fact to show that $I_m \p 0$.  To this end, define $\Delta^*_m = \psi(\Delta_m)$, where
\begin{equation*}
\psi(\delta) = \sup_{\substack{\bm x, \bm y \in \R^d \\ \|\bm x - \bm y\|_\infty \leq \delta}}\sum_{\ell=1}^d  |F^{\sss[\ell]}(x^{\sss[\ell]}) - F^{\sss[\ell]}(y^{\sss[\ell]})|, \qquad \delta \geq 0.
\end{equation*}
Since $ \|\bm x - \bm y\|_\infty \leq \Delta_m$ implies that $\rho\big((t,\bm x),(t,\bm y)\big)\le \Delta^*_m$, we have that
\begin{align*}
I_m&\le  \sup_{k>m} k^{-\frac12} \sup_{\substack{s,t \in [0, k]\\ s=t}} \sup_{\substack{\bm x, \bm y \in \R^d \\ \|\bm x - \bm y\|_\infty \leq \Delta_m}}  | K(s, \bm x) - K(t, \bm y) | \le \sup_{k>m} k^{-\frac12} J_m \le m^{-\frac12} J_m,
\end{align*}
where $J_m = \phi(\Delta^*_m)$ with 
$$
\phi(\delta) = \sup_{\substack{s,t \in [0,\infty), \bm x, \bm y \in \R^d \\ \rho ((s,\bm x), (t, \bm y) ) \leq \delta}}  | K(s, \bm x) - K(t, \bm y) |, \qquad \delta \geq 0.
$$
Let $\eps > 0$.  By almost sure uniform continuity of the sample paths of the process $K$, there exists $\delta_1 = \delta_1(\eps) > 0$ such that, for all $0 \leq \delta \le  \delta_1$, $\phi(\delta) < \eps$ almost surely.  Since the $d$ univariate margins $F^{\sss[1]},\dots,F^{\sss[d]}$ of $F$ are (uniformly) continuous, there exists $\delta_0 = \delta_0(\delta_1) > 0$ such that, for all $0 \leq \delta \le  \delta_0$, $\psi(\delta) < \delta_1$.  Therefore 
\begin{align*}
\P(J_m > \eps) = \P \big( \phi(\Delta_m^*) > \eps \big)  &\leq \P \big( \phi(\Delta_m^*) > \eps, \Delta_m^* \leq \delta_1 \big) + \P( \Delta_m^* > \delta_1 )\\
&= \P( \Delta_m^* > \delta_1 )\\
&=\P \big( \psi(\Delta_m) > \delta_1\big)\\
&  \leq \P \big( \psi(\Delta_m) > \delta_1, \Delta_m \leq \delta_0 \big) + \P( \Delta_m > \delta_0 )\\
& = \P( \Delta_m > \delta_0 ).
\end{align*}
Since $\Delta_m = \| \bm \Xc_{i,m} - \bm x_i \|_\infty \p 0$ by assumption, this shows that $J_m \p 0$ which completes the proof since $I_m \leq m^{-\frac12} J_m$.
\end{proof}

\begin{proof}[\bf Proof of Proposition~\ref{prop:equiv:points}]
The second claim is immediate from the first claim and Theorem \ref{thm:H0}, so we need only prove the first claim.

First, we assert that, for any $\ell \in \llb 1, p \rrb$, 
\begin{equation}
  \label{eq:tmp:ij}
\sup_{k > m} k^{-\frac12} \max_{1 \leq i < j \leq k}(j-i+1) | F_{i:j}(\bm \Xc_{\ell,m}) - F(\bm \Xc_{\ell,m}) -  F_{i:j}(\bm x_\ell) +  F(\bm x_\ell) |  = o_\P(1).
\end{equation}
Indeed, the left-hand side of \eqref{eq:tmp:ij} is equal to
\begin{align*}
  \sup_{k > m} & \, k^{-\frac12} \max_{1 \leq i < j \leq k} \Big| \sum_{r=i}^j \{ 1(\bm X_r \leq \bm \Xc_{\ell,m}) - F(\bm \Xc_{\ell,m}) -  1(\bm X_r \leq \bm x_\ell) +  F(\bm x_\ell) \} \Big|  \\
  =&  \sup_{k > m} k^{-\frac12} \max_{1 \leq i < j \leq k} \Big| \sum_{r=1}^j \{ 1(\bm X_r \leq \bm \Xc_{\ell,m}) - F(\bm \Xc_{\ell,m}) -  1(\bm X_r \leq \bm x_\ell) +  F(\bm x_\ell) \} \\
                             &- \sum_{r=1}^{i-1} \{ 1(\bm X_r \leq \bm \Xc_{\ell,m}) - F(\bm \Xc_{\ell,m}) -  1(\bm X_r \leq \bm x_\ell) +  F(\bm x_\ell) \} \Big| \\
  \leq& 2 \sup_{k > m} k^{-\frac12} \max_{j \in \llb 1, k \rrb} j | F_{1:j}(\bm \Xc_{\ell,m}) - F(\bm \Xc_{\ell,m}) -  F_{1:j}(\bm x_\ell) +  F(\bm x_\ell) |,
\end{align*}
so the assertion holds by \eqref{eq:cond:points}. Using the fact that all norms on $\R^p$ are equivalent, \eqref{eq:tmp:ij} implies that
\begin{equation}
  \label{eq:tmp:sigma:norm}
\sup_{k > m} k^{-\frac12} \max_{1 \leq i < j \leq k} (j - i + 1) \| \bm F_{i:j}^{\sPc_m} - \bm F^{\sPc_m} - \bm F_{i:j}^{\sPc} - \bm F^{\sPc} \|_{(\Sigma^\sPc)^{-1}} = o_\P(1),
\end{equation}
where $\bm F^{\sPc_m} = \big( F(\bm \Xc_{1,m}),\dots,F(\bm \Xc_{p,m}) \big)$ and $\bm F^{\sPc} = \big( F(\bm x_1),\dots,F(\bm x_p) \big)$ (similarly for $\bm F^{\cdot}_{i:j}$).

Recall the definition of $\tilde D_m^\sPc$ in~\eqref{eq:tilde:Dm} which can be rewritten as
\begin{align*}
  \tilde D_m^{\sPc}(k) &= \max_{j \in \llb m, k \llb} \frac{j (k-j)}{m^{\frac32}}  \| \bm F_{1:j}^{\sPc} - \bm F_{j+1:k}^{\sPc} \|_{(\Sigma^\sPc)^{-1}}, \qquad k \geq m+1,
\end{align*} 
and define the unobservable detector $\tilde D_m^{\sPc_m}$ by changing the norm in the definition of $D_m^{\sPc_m}$ as
\begin{align*}
  \tilde D_m^{\sPc_m}(k) &= \max_{j \in \llb m, k \llb} \frac{j (k-j)}{m^{\frac32}}  \| \bm F_{1:j}^{\sPc_m} - \bm F_{j+1:k}^{\sPc_m} \|_{(\Sigma^\sPc)^{-1}}, \qquad k \geq m+1.
\end{align*} 
Using the reverse triangle inequality for the maximum norm and the norm $\|\cdot\|_{(\Sigma^\sPc)^{-1}}$, we then obtain that, for any $k \geq m+1$, 
\begin{align*}
  &  | \tilde D_m^{\sPc_m}(k) - \tilde D_m^\sPc(k) | \\
                                 &\leq   
   \max_{j \in \llb m, k \llb}  \frac{j (k-j)}{m^{\frac32}} 
   \big| \| \bm F_{1:j}^{\sPc_m} - \bm F^{\sPc_m} - \bm F_{j+1:k}^{\sPc_m} + \bm F^{\sPc_m}\|_{(\Sigma^\sPc)^{-1}} -  \| \bm F_{1:j}^{\sPc} - \bm F^{\sPc} -  \bm F_{j+1:k}^{\sPc} + \bm F^{\sPc} \|_{(\Sigma^\sPc)^{-1}} \big| \\
                                 &\leq   
    \max_{j \in \llb m, k \llb}  \frac{j (k-j)}{m^{\frac32}} \| \bm F_{1:j}^{\sPc_m} - \bm F^{\sPc_m} - \bm F_{1:j}^{\sPc} + \bm F^{\sPc} - \bm F_{j+1:k}^{\sPc_m} + \bm F^{\sPc_m} + \bm F_{j+1:k}^{\sPc} - \bm F^{\sPc} \|_{(\Sigma^\sPc)^{-1}}\\
                                 &\le \max_{j \in \llb m, k \llb}  \frac{j (k-j)}{m^{\frac32}}\Big\{ \| \bm F_{1:j}^{\sPc_m} - \bm F^{\sPc_m} - \bm F_{1:j}^{\sPc} + \bm F^{\sPc} \|_{(\Sigma^\sPc)^{-1}} + \| \bm F_{j+1:k}^{\sPc_m} - \bm F^{\sPc_m} - \bm F_{j+1:k}^{\sPc} + \bm F^{\sPc} \|_{(\Sigma^\sPc)^{-1}} \Big\}.
\end{align*}
Therefore,
 \begin{align*}
&\sup_{k > m} \left(\frac{m}{k}\right)^{\frac32+\eta}  | \tilde D_m^{\sPc_m}(k) - \tilde D_m^\sPc(k) |   \\
&\le   \sup_{k > m} \left(\frac{m}{k}\right)^\eta k^{-\frac32} \max_{j \in \llb m, k \llb}  j (k-j) \| \bm F_{1:j}^{\sPc_m} - \bm F^{\sPc_m} - \bm F_{1:j}^{\sPc} + \bm F^{\sPc} \|_{(\Sigma^\sPc)^{-1}} \\ 
& \quad +  \sup_{k > m} \left(\frac{m}{k}\right)^\eta k^{-\frac32} \max_{j \in \llb m, k \llb}  j (k-j) \| \bm F_{j+1:k}^{\sPc_m} - \bm F^{\sPc_m} - \bm F_{j+1:k}^{\sPc} + \bm F^{\sPc} \|_{(\Sigma^\sPc)^{-1}}.
\end{align*}
We claim that both terms on the right hand side converge to 0 in probability.  For example, since $j\le k$  and $k>m$, the second term is at most
\begin{align*}
&\sup_{k > m} k^{-\frac12} \max_{j \in \llb m, k \llb}(k-j)\| \bm F_{j+1:k}^{\sPc_m} - \bm F^{\sPc_m} - \bm F_{j+1:k}^{\sPc} + \bm F^{\sPc} \|_{(\Sigma^\sPc)^{-1}}\\
         & \le \sup_{k > m} k^{-\frac12} \max_{1 \leq i < j \leq k} (j - i + 1) \| \bm F_{i:j}^{\sPc_m} - \bm F^{\sPc_m} - \bm F_{i:j}^{\sPc} + \bm F^{\sPc} \|_{(\Sigma^\sPc)^{-1}},
\end{align*}
which converges to 0 in probability by  \eqref{eq:tmp:sigma:norm}.  The first term is similar. Hence,
\begin{align}
  \label{eq:tmp:proven}
  \sup_{k > m} \left(\frac{m}{k}\right)^{\frac32+\eta}  | \tilde D_m^{\sPc_m}(k) - \tilde D_m^\sPc(k) | = o_\P(1).
\end{align}
From Lemma~\ref{lem:tilde:D-to-D}, we have that $\sup_{k > m} (m/k)^{\frac32+\eta} | \tilde D_m^\sPc(k) - D_m^\sPc(k) | = o_\P(1)$.  Therefore, it remains to prove that
\begin{align*}
  \sup_{k > m} \left(\frac{m}{k}\right)^{\frac32+\eta} |  D_m^{\sPc_m}(k) - \tilde D_m^{\sPc_m}(k) | = o_\P(1).
\end{align*}
Proceeding as in the proof of Lemma~\ref{lem:tilde:D-to-D}, we have that 
\begin{align*} 
\sup_{k > m} &\Big(\frac m k\Big)^{\frac32+\eta} | \tilde D_m^{\sPc_m} (k) - D_m^{\sPc_m}(k) | \\
&\leq
\sup_{k > m} \Big(\frac m k\Big)^{\frac32+\eta} \max_{j \in \llb m,k \llb} \frac{j (k-j)}{m^{\frac32}} \left|  \| \bm F_{1:j}^{\sPc_m} - \bm F_{j+1:k}^{\sPc_m} \|_{(\Sigma^\sPc)^{-1}} - \| \bm F_{1:j}^{\sPc_m} - \bm F_{j+1:k}^{\sPc_m} \|_{(\Sigma_m^{\sPc_m})^{-1}}\right| \\
&\leq
\sup_{k > m} \Big(\frac m k\Big)^{\frac32+\eta} \max_{j \in \llb m,k \llb} \frac{j (k-j)}{p^{\frac12} m^{\frac32}} \left| \left( \bm F_{1:j}^{\sPc_m} - \bm F_{j+1:k}^{\sPc_m} \right)^\top \left((\Sigma^{\sPc})^{-1} - (\Sigma_m^{\sPc_m})^{-1} \right) \left( \bm F_{1:j}^{\sPc_m} - \bm F_{j+1:k}^{\sPc_m} \right) \right|^{\frac12} \\
&\leq \frac{1}{p^{\frac12}} \left\| (\Sigma^\sPc)^{-1} - (\Sigma_m^{\sPc_m})^{-1} \right\|_{\textrm{op}}^{\frac12} \sup_{k > m} \left(\frac m k \right)^{\frac32+\eta} \max_{j \in \llb m,k \llb} \frac{j (k-j)}{m^{\frac32}} \left\| \bm F_{1:j}^{\sPc_m} - \bm F_{j+1:k}^{\sPc_m} \right\|_2.
\end{align*}
We claim that this converges to zero in probability (which completes the proof).  By Condition~\ref{cond:sigma} and the fact that $\Sigma_m^{\sPc_m} \p \Sigma^\sPc$, we have that $\left\| (\Sigma^\sPc)^{-1} - (\Sigma_m^\sPc)^{-1} \right\|_{\textrm{op}} = o_\P(1)$.  To prove this final claim, it therefore suffices to show that 
$$
 \sup_{k > m} \left(\frac m k \right)^{\frac32+\eta} \max_{j \in \llb m,k \llb} \frac{j (k-j)}{m^{\frac32}} \left\| \bm F_{1:j}^{\sPc_m} - \bm F_{j+1:k}^{\sPc_m} \right\|_2 = O_\P(1).
$$
Indeed, by equivalence of norms on $\R^p$, this follows from the fact that $\sup_{k > m} (m/k)^{\frac32+\eta} \tilde D_m^{\sPc_m}(k)  = O_\P(1)$, itself a consequence of~\eqref{eq:tmp:proven}, Lemma \ref{lem:tilde:D-to-D}, and Theorem \ref{thm:H0}.
\end{proof}


\section{Details of Monte Carlo experiments}
\label{sec:MC:details}

In this section, we provide the implementation details of the Monte Carlo experiments that we carried out in the case of low-dimensional ($d \in \{1,2,3\}$) continuous observations and whose main findings are summarized in Section~\ref{sec:MC}. In all experiments, the sequential tests were carried out at the $\alpha = 5\%$ nominal level.

\subsection{Univariate experiments for the procedure based on $D_m^{\sPc_m}$ under the null}

To investigate the empirical levels of the sequential test based on $D_m^{\sPc_m}$ when $d=1$, we considered 9 data generating models, denoted M1, \dots, M9. Models M1, \dots, M6 are AR(1) models with independent standard normal innovations whose autoregressive parameter is equal to 0, 0.1, 0.3, 0.5, 0.7 and $-0.7$, respectively. Model M7 is a GARCH(1,1) model with independent standard normal innovations and parameters $\omega = 0.012$, $\beta = 0.919$ and $\alpha = 0.072$ to mimic SP500 daily log-returns following \cite{JonPooRoc07}. Models M8 and M9 are the nonlinear autoregressive model used in \citet[Section 3.3]{PapPol01} and the exponential autoregressive model considered in \cite{AueTjo90} and \citet[Section 3.3]{PapPol01}, respectively. The underlying generating equations are
\begin{equation*}
X_i = 0.6 \sin( X_{i-1} ) + \epsilon_i
\end{equation*}
and
\begin{equation*}
X_i = \{ 0.8 - 1.1 \exp ( - 50 X_{i-1}^2 ) \} X_{i-1} + 0.1 \epsilon_i,
\end{equation*}
respectively, where the $\epsilon_i$ are independent standard normal innovations. Note that, for all time series models, a burn-out sample of 100 observations was used.

\input{H0.tex}

For each of the nine models, the probability of rejection of $H_0$ in~\eqref{eq:H0} was estimated from 1000 samples of size $n = m + 5000$ with $m \in \{200, 400, 800, 1600\}$ and for $p \in \{2, 5, 10, 20\}$. The empirical levels are reported in Table~\ref{tab:H0}. As one can see, for any fixed $p$, reassuringly, they decrease as $m$ increases. For $m \geq 800$, it is mostly for the models with strong serial dependence such as M5, M8 and M9 that the empirical levels are not below the 5\% nominal level. The latter is not so surprising and highlights the difficulty of the estimation of the long-run covariance matrix $\Sigma^\sPc$ using the estimator  $\Sigma^{\sPc_m}_m$ in the case of strong serial dependence. It may be slightly more surprising for the GARCH(1,1) model M7 for which $m=1600$ seems necessary to obtain a reasonably good estimate of $\Sigma^\sPc$. The fact that for most other models, the empirical levels are all below the 5\% nominal level when $m \geq 800$ is a consequence of the fact that they are underestimated in all settings. Indeed, the monitoring was stopped after 5000 steps whereas it would theoretically be necessary to monitor ``indefinitely'' to compute empirical levels accurately. For any fixed $m$, we see that increasing $p$ tends in general to increase the empirical level. This is again a consequence of the difficulty of the estimation of the long-run covariance matrix $\Sigma^\sPc$ which is a $p \times p$ matrix. Note that, since the monitoring procedure based on $D_m^{\sPc_m}$ is margin-free (as verified in Section~\ref{sec:margin:free}), it is not necessary to empirically study the influence of the contemporary distribution $F$ on the empirical levels.

\input{H0IidTrueLRV.tex}

In a second experiment, we briefly investigated the effect of the estimation of $\Sigma^\sPc$ on the empirical levels in the case of independent observations. Instead of estimating $\Sigma^\sPc$ from the learning sample, we used its true value whose elements, in the considered setting (see Remark~\ref{remark:points}), are given by
$$
\Cov\{ \1(X_1^{\sss[1]} \leq F^{-1}( i/(p+1) ) , \1(X_1^{\sss[1]} \leq F^{-1}( j/(p+1) ) \} = \min(i, j) / (p+1) -  i j / (p+1)^2, 
$$
for $i,j \in \llb 1, p \rrb$. The empirical levels were then estimated from 1000 random samples of size $n = m + 5000$ from the standard normal distribution. The results are reported in Table~\ref{tab:H0IidTrueLRV}. By comparing the results with the first horizontal block of Table~\ref{tab:H0}, we see, as could have been expected, that for the same value of $m$, the use of the true long-run covariance matrix leads to lower empirical levels than when it is estimated.

\input{H0IidLargeR.tex}

As a last experiment under the null, we investigated the quality of the model fitted at the end of Section~\ref{sec:quant} to extrapolate the values of the quantiles of the distribution of $\Lc_{p,\eta}$ for $p > 20$ and $\eta = 0.001$. Using $m=1600$ and 1000 random samples of size $n = m + 5000$ from the standard normal distribution, we estimated rejection percentages for $p \in \{30,40,50\}$. These are given in Table~\ref{tab:H0IidLargeR} and suggest that the quality of the model for extrapolating the values of the quantiles may be acceptable when $20 < p \leq 50$ (although it may lead to some slightly more conservative tests).

\subsection{Univariate experiments for the procedure based on $D_m^{\sPc_m}$ under alternatives}

In order to understand the behavior of the monitoring procedure based on the detector $D_m^{\sPc_m}$ considered in Section~\ref{sec:univ:points} under alternatives to $H_0$ in~\eqref{eq:H0}, we considered successively changes in the expectation of the $X_i^{\sss[1]}$'s, in their variance and in their d.f.\ (while keeping their expectation and variance constant).  As a first experiment, we studied the finite-sample behavior of the sequential test under a change in the expectation of an AR(1) model with autoregressive parameter equal to 0.3 (Model M3). Specifically, to estimate rejection percentages, we generated 1000 samples of size $n = m + 5000$ from Model M3 with $m=800$ and, for each sample, added a positive offset of $\delta$ to all observations after position $m+k$ with $k \in \{0, 500, 1000, 2000\}$. The results are reported in Figure~\ref{fig:mean}. Notice that only the exceedences (of the detectors with respect to their thresholds) after position $m+k$ are taken into account when calculating the rejection percentages.

\begin{figure}[t!]
\begin{center}
  \includegraphics*[width=1\linewidth]{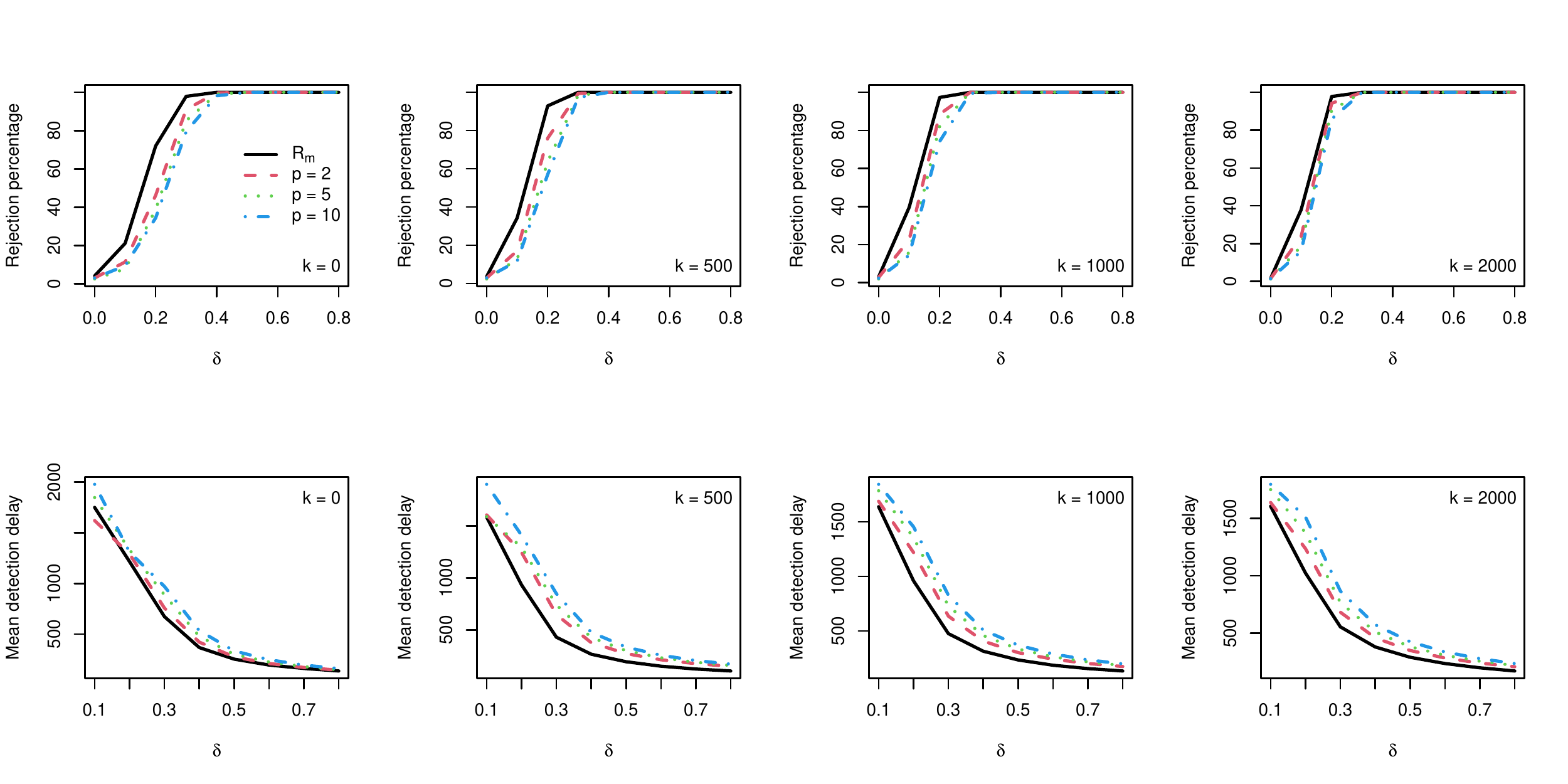}
  \caption{\label{fig:mean} Rejection percentages of $H_0$ in~\eqref{eq:H0} (first row) and corresponding mean detection delays (second row) for the procedure based on $R_m$ in~\eqref{eq:det:mean} (solid line) and for the procedure $D_m^{\sPc_m}$ considered in Section~\ref{sec:univ:points} with $p \in \{2, 5, 10\}$ (dash, dotted and dash-dotted lines) estimated from 1000 samples of size $n = m + 5000$ from Model M3 with $m=800$ such that, for each sample, a positive offset of $\delta$ was added to all observations after position $m+k$.}
\end{center}
\end{figure}

As one can see from the top row of graphs in Figure~\ref{fig:mean}, as expected, the power of all procedures increases as $\delta$ increases. Furthermore, for the procedure based on $D_m^{\sPc_m}$ and a fixed value of the offset $\delta$, increasing the number $p$ of evaluation points slightly lowers the power of the test. We also see that the procedure based on $R_m$ in~\eqref{eq:det:mean} is always the most powerful. This was to be expected as the latter was specifically designed to be sensitive to changes in the mean. Similarly, from the second row of plots in Figure~\ref{fig:mean}, we see that mean detection delays are smallest for the procedure based on $R_m$ and increase for the procedure based on $D_m^{\sPc_m}$ as $p$ increases. Note finally that the power of every procedure becomes larger as the time $k$ at which the offset $\delta$ is added becomes closer to half of $n = m + 5000$. This is a consequence of using of CUSUM statistics to define the detectors. 

As a second experiment, we considered a change in the variance of independent centered observations. To estimate the power of the sequential test, we generated 1000 samples of size $n = m + 5000$  with $m=800$ such that, observations up to position $m+k$  with $k \in \{0, 500, 1000, 2000\}$ are from the standard normal distribution while observations after position $m+k$ are from the $N(0,\sigma^2)$ distribution. The results are reported in Figure~\ref{fig:var}. As expected, the power of the procedure based on $D_m^{\sPc_m}$ increases as $\sigma$ deviates further away from one. In contrast to the first experiment however, the rejection percentages (resp.\ mean detection delays) increase (resp.\ decrease) as the number of evaluation points $p$ increases. Notice that the improvement as $p$ increases from 5 to 10 appears to be rather small.

\begin{figure}[t!]
\begin{center}
  \includegraphics*[width=1\linewidth]{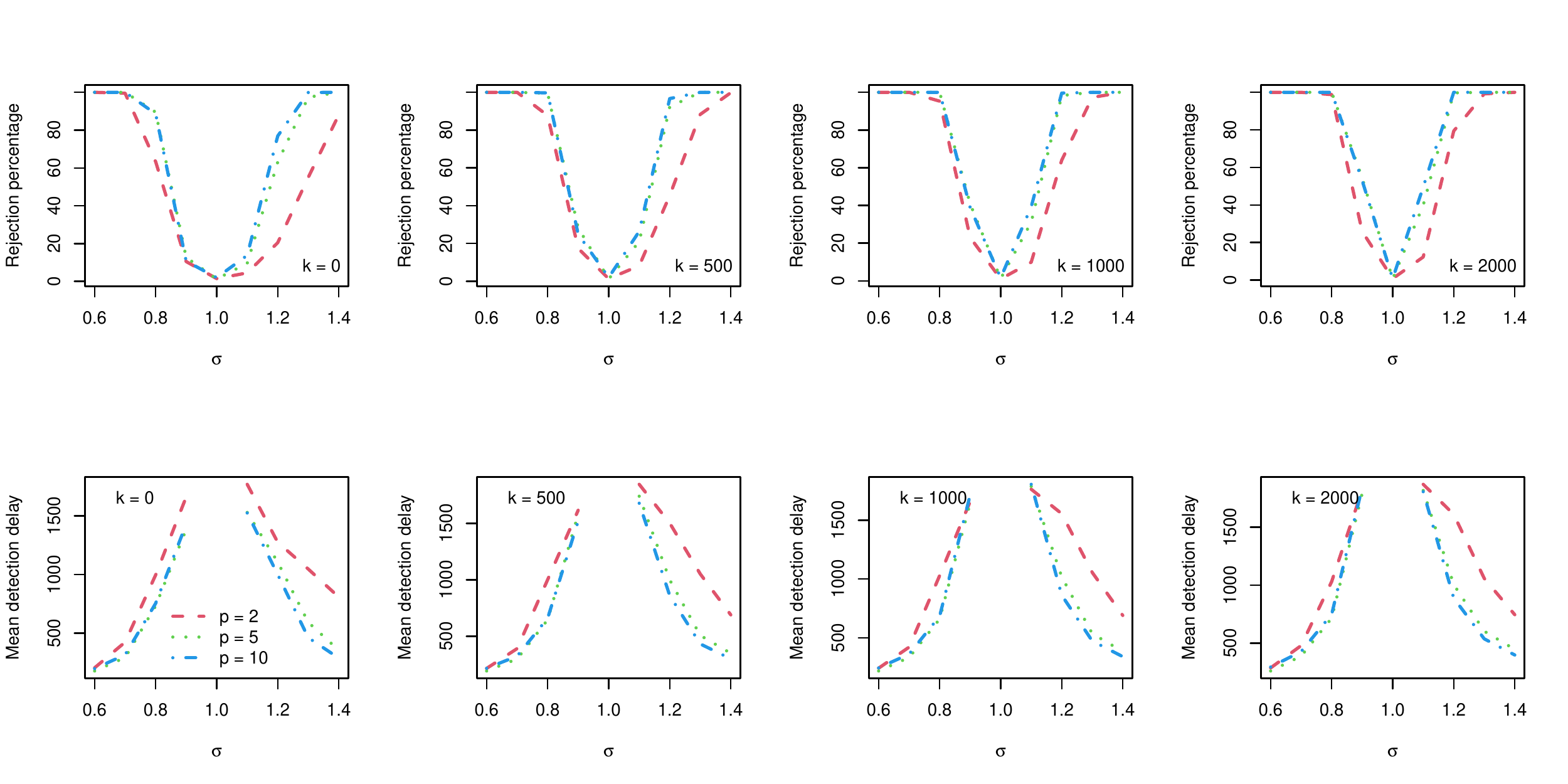}
  \caption{\label{fig:var} Rejection percentages of $H_0$ in~\eqref{eq:H0} and corresponding mean detection delays for the procedure based on $D_m^{\sPc_m}$ with $p \in \{2, 5, 10\}$ estimated from 1000 random samples of size $n = m + 5000$  with $m=800$ such that observations up to position $m+k$ with $k \in \{0, 500, 1000, 2000\}$ are from the standard normal distribution while observations after position $m+k$ are for the $N(0,\sigma^2)$ distribution.}
\end{center}
\end{figure}

As a final experiment, we considered a change in the contemporary distribution of independent observations that keeps the expectation and the variance constant. To estimate the rejection percentages, we generated 1000 samples of size $n = m + 5000$  with $m=800$ such that observations up to position $m+k$ with $k \in \{0, 500, 1000, 2000\}$ are from the scaled Student $t$ distribution with 3 degrees of freedom (where the scaling is performed so that the variance is equal to one) while observations after position $m+k$ are from the scaled Student $t$ distribution with $\nu \in \llb 3, 10 \rrb$ degrees of freedom. The results are reported in Figure~\ref{fig:dist}. As expected, the power of the procedure based on $D_m^{\sPc_m}$ increases as $\nu$ increases. Furthermore, as in the previous experiment, the rejection percentages (resp.\ mean detection delays) are larger (resp.\ smaller) when $p \in \{5,10\}$. Somehow surprisingly however, the results seem slightly better when $p=5$. 

\begin{figure}[t!]
\begin{center}
  \includegraphics*[width=1\linewidth]{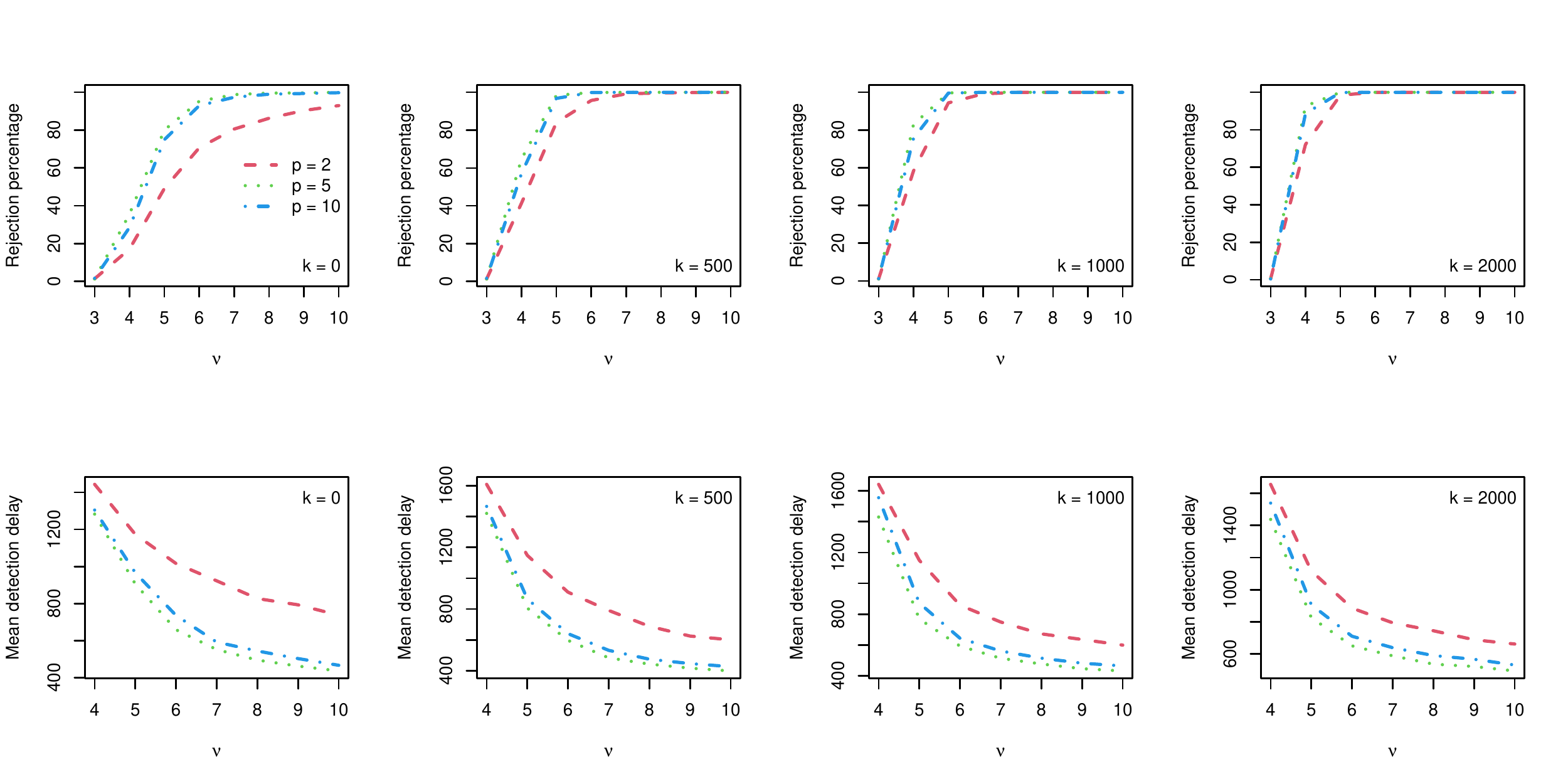}
  \caption{\label{fig:dist} Rejection percentages of $H_0$ in~\eqref{eq:H0} and corresponding mean detection delays for the procedure based on $D_m^{\sPc_m}$ with $p \in \{2, 5, 10\}$ estimated from 1000 random samples of size $n = m + 5000$  with $m=800$ such that observations up to position $m+k$  with $k \in \{0, 500, 1000, 2000\}$ are from the scaled Student $t$ distribution with $\nu=3$ degrees of freedom while observations after position $m+k$ are from the scaled Student $t$ distribution with $\nu \in \llb 3, 10 \rrb$ degrees of freedom.}
\end{center}
\end{figure}

\subsection{Multivariate experiments for the procedure based on $D_m^{\sPc_m}$ under the null}

Given $d \in \{2,3\}$ and a $d$-dimensional copula $C$, we used a multivariate AR(1) model to generate potentially serially dependent observations under $H_0$ in~\eqref{eq:H0}. Let $\bm U_i$, $i \in \llb -100, n \rrb$, be a $d$-dimensional i.i.d.\ sample from a copula~$C$. Then, set $\bm \epsilon_i = (\Phi^{-1}(U_i^{\sss [1]}),\dots,\Phi^{-1}(U_i^{\sss [d]}))$, where $\Phi$ is the d.f.\ of the standard normal distribution, and $\bm X_{-100} = \bm \epsilon_{-100}$. Finally, for any $j \in \llb 1, d \rrb$ and $i \in \llb -99, n \rrb$, compute recursively
\begin{equation}
\label{eq:ar1}
  X_i^{\sss [j]} = \beta X_{i-1}^{\sss [j]} + \epsilon_i^{\sss [j]}.
\end{equation}

\input{H0BivAR}

Recall that, when $d > 1$, the evaluation points of the monitoring procedure based on $D_m^{\sPc_m}$ are chosen from the learning sample using the point selection procedure described in Section~\ref{sec:mult:points}. To evaluate the behavior of the procedure when $d=2$ with $r \in \{3,4\}$ and $\kappa \in \{1.5,2,3\}$ under the null, in a first experiment, we computed its rejection percentages from 1000 bivariate samples of size $n = m + 5000$ generated from the time series model~\eqref{eq:ar1} with $\beta=0.3$ and $C$ the bivariate Gumbel--Hougaard copula with a Kendall's tau of $\tau \in \{0,0.33,0.66\}$. The empirical levels are reported in the columns $D_m^{\sPc_m}$ of Table~\ref{tab:H0BivAR}. The columns $\bar p$ report the average number of grid points retained by the point selection procedure of Section~\ref{sec:mult:points}. As one can see, reassuringly, the empirical levels improve in all settings as $m$ increases. Unsurprisingly, they are higher for $r=4$ than for $r=3$ since a larger value of $r$ tends to result in a larger number of selected points $p$ and thus in a more difficult estimation of the underlying long run covariance matrix. Also unsurprisingly, the number of selected points $p$ tends to increase as $\kappa$ increases and to decrease as $\tau$ increases, that is, as the cross-sectional dependence in the underlying time series changes from independence to stronger positive association.

\input{H0TriAR}

We additionally considered a trivariate version of the previous experiment under the null based on the Clayton copula. The results, reported in Table~\ref{tab:H0TriAR}, are qualitatively the same.

\subsection{Multivariate experiments for the procedure based on $D_m^{\sPc_m}$ under alternatives}

\begin{figure}[t!]
\begin{center}
\includegraphics*[width=1\linewidth]{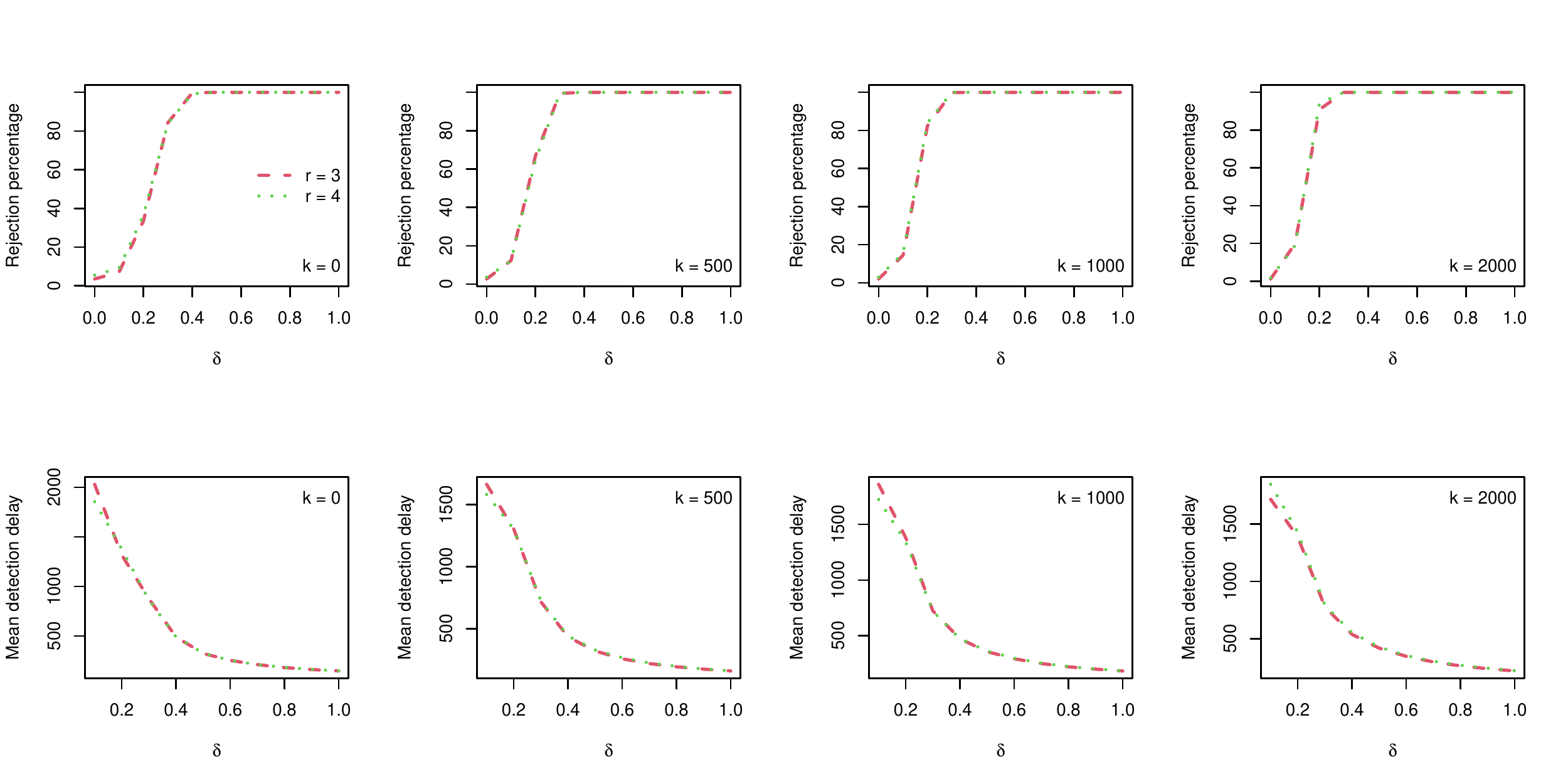}
\caption{\label{fig:H1BivTs} Rejection percentages of $H_0$ in~\eqref{eq:H0} and corresponding mean detection delays for the procedure based on $D_m^{\sPc_m}$ with $r \in \{3, 4\}$ and $\kappa=1.5$ estimated from 1000 bivariate samples of size $n = m + 5000$  with $m=800$ generated from the time series model~\eqref{eq:ar1} with $\beta=0.3$ and $C$ the bivariate Frank copula with a Kendall's tau of $0.5$ such that, for each sample, a positive offset of $\delta$ was added to the first component of all bivariate observations after position $m+k$. The average number of selected points is approximately 7 for $r=3$ and 10.9 for $r=4$.}
\end{center}
\end{figure}

In a last series of bivariate and trivariate experiments, we investigated the power of the procedure based on $D_m^{\sPc_m}$.

We first estimated its rejection percentages and corresponding mean detection delays for $r \in \{3, 4\}$ and $\kappa=1.5$ from 1000 samples of size $n = m + 5000$  with $m=800$ generated from the time series model~\eqref{eq:ar1} with $\beta=0.3$ and $C$ the bivariate Frank copula with a Kendall's tau of $\tau=0.5$ such that, for each sample, a positive offset of $\delta$ was added to the first component of all bivariate observations after position $m+k$. The results are represented in Figure~\ref{fig:H1BivTs}. As one can see, the value of $r \in \{3,4\}$ has hardly any influence on the power or on the mean detection delay.

We next considered a similar experiment where the change affects only the first margin which changes from the scaled Student $t$ with 3 degrees of freedom to the scaled Student $t$ distribution with $\nu \in \llb 3, 10 \rrb$ degrees of freedom. The copula (the bivariate Frank with a Kendall's tau of $0.5$) and the second margin (the Student $t$ with $\nu=3$ degrees of freedom) remain constant. The results are displayed in Figure~\ref{fig:H1BivIid}. As one can see, using $r=4$ rather than $r=3$ leads to a slightly more powerful procedure which detects the change faster on average.

In a third experiment, we focused on the effect of a change of the dependence parameter of the copula in the case of serially independent data. Before the change, observations are generated from the bivariate Normal copula with a Kendall's tau of $0.5$, while after the change they come from the bivariate Normal copula with a Kendall's tau of $\tau \in \{0.1,\dots,0.9\}$. The rejection percentages and corresponding mean detection delays are represented in Figure~\ref{fig:H1BivTau}. As in the previous experiment, the results for $r=4$ are slightly better than for $r=3$.

In a fourth bivariate experiment, we considered the situation where the copula changes while the strength of association measured in terms of Kendall's tau remains constant. Specifically, before the change, observations are generated from the bivariate Clayton copula with a Kendall's tau of $0.5$ (which is lower tail dependent), while after the change they arise from the bivariate Gumbel--Hougaard copula with a Kendall's tau of $0.5$ (which is upper tail dependent). The results are reported Table~\ref{tab:H1BivCop}. The procedure with $r=4$ is again slightly  more powerful and detects the change faster than the procedure with $r=3$.

\begin{figure}[t!]
\begin{center}
\includegraphics*[width=1\linewidth]{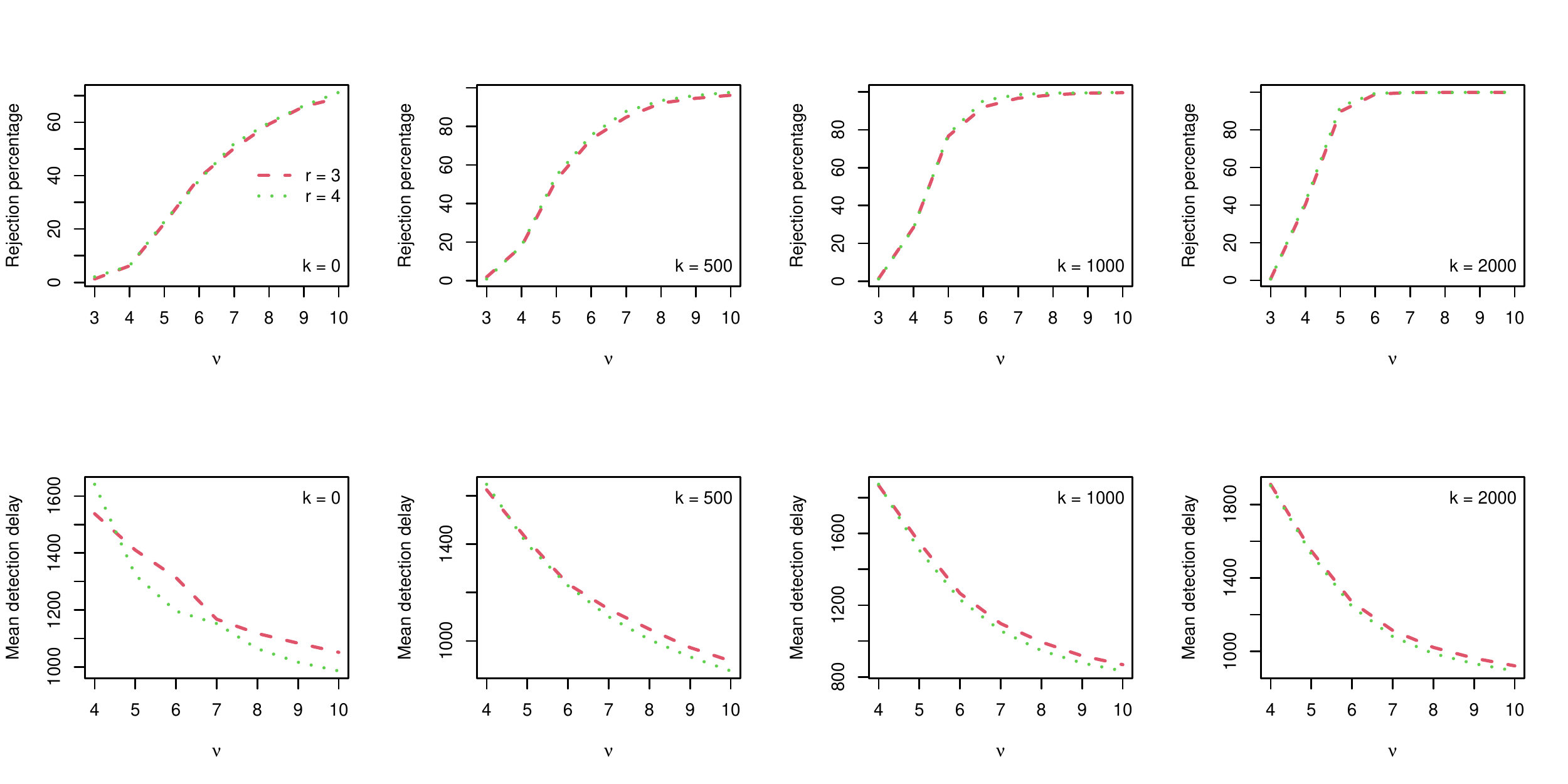}
\caption{\label{fig:H1BivIid} Rejection percentages of $H_0$ in~\eqref{eq:H0} and corresponding mean detection delays for the procedure based on $D_m^{\sPc_m}$ with $r \in \{3, 4\}$ and $\kappa=1.5$ estimated from 1000 bivariate random samples of size $n = m + 5000$  with $m=800$ such that observations up to position $m+k$ are from a d.f.\ whose copula is the bivariate Frank with a Kendall's tau of $0.5$ and whose margins are scaled Student $t$ with $\nu=3$ degrees of freedom, while observations after position $m+k$ are still from a d.f. with the same copula and same second margin but with first margin the scaled Student $t$ with $\nu \in \llb 3, 10 \rrb$ degrees of freedom. The average number of selected points is approximately 7 for $r=3$ and 10.6 for $r=4$.}
\end{center}
\end{figure}

\begin{figure}[t!]
\begin{center}
\includegraphics*[width=1\linewidth]{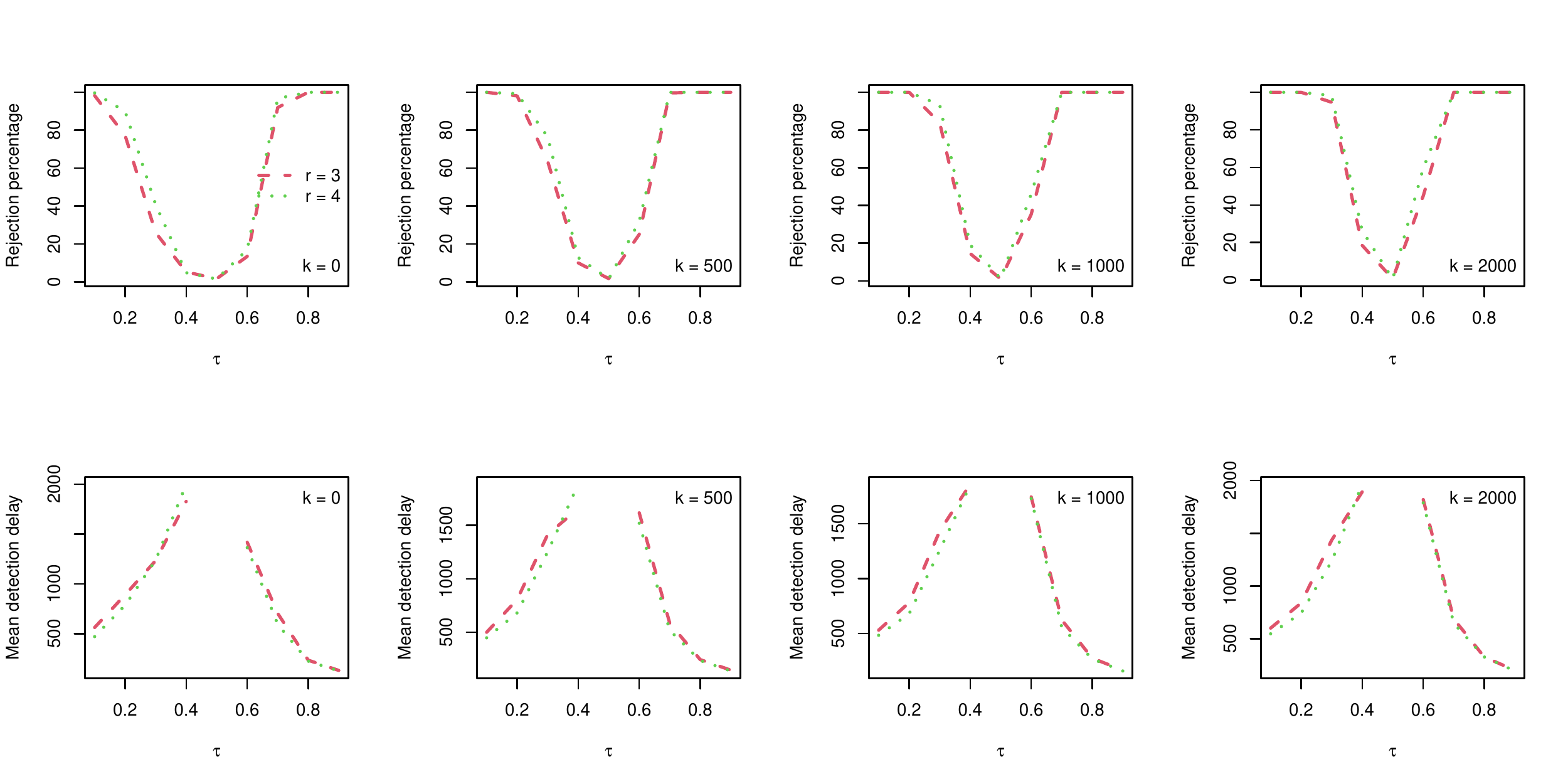}
\caption{\label{fig:H1BivTau} Rejection percentages of $H_0$ in~\eqref{eq:H0} and corresponding mean detection delays for the procedure based on $D_m^{\sPc_m}$ with $r \in \{3, 4\}$ and $\kappa=1.5$ estimated from 1000 bivariate random samples of size $n = m + 5000$  with $m=800$ such that observations up to position $m+k$  are from the bivariate normal copula with a Kendall's tau of $0.5$ while observations after position $m+k$ are from the bivariate normal copula with a Kendall's tau of $\tau \in \{0.1,\dots,0.9\}$. The average number of selected points is approximately 7 for $r=3$ and 11.8 for $r=4$.}
\end{center}
\end{figure}

\input{H1BivCop.tex}


We concluded our multivariate simulations under alternatives by considering trivariate versions of the previous bivariate experiments. We used $r = 3$ and $\kappa = 1.5$. The results are reported in Figure~\ref{fig:H1Tri} and Table~\ref{tab:H1TriCop} and are qualitatively the same as in the bivariate case. Notice however the rather low estimated rejection percentages for the trivariate version of the second bivariate experiment reported in the second row of graphs of Figure~\ref{fig:H1Tri}.

\begin{figure}[t!]
\begin{center}
\includegraphics*[width=1\linewidth]{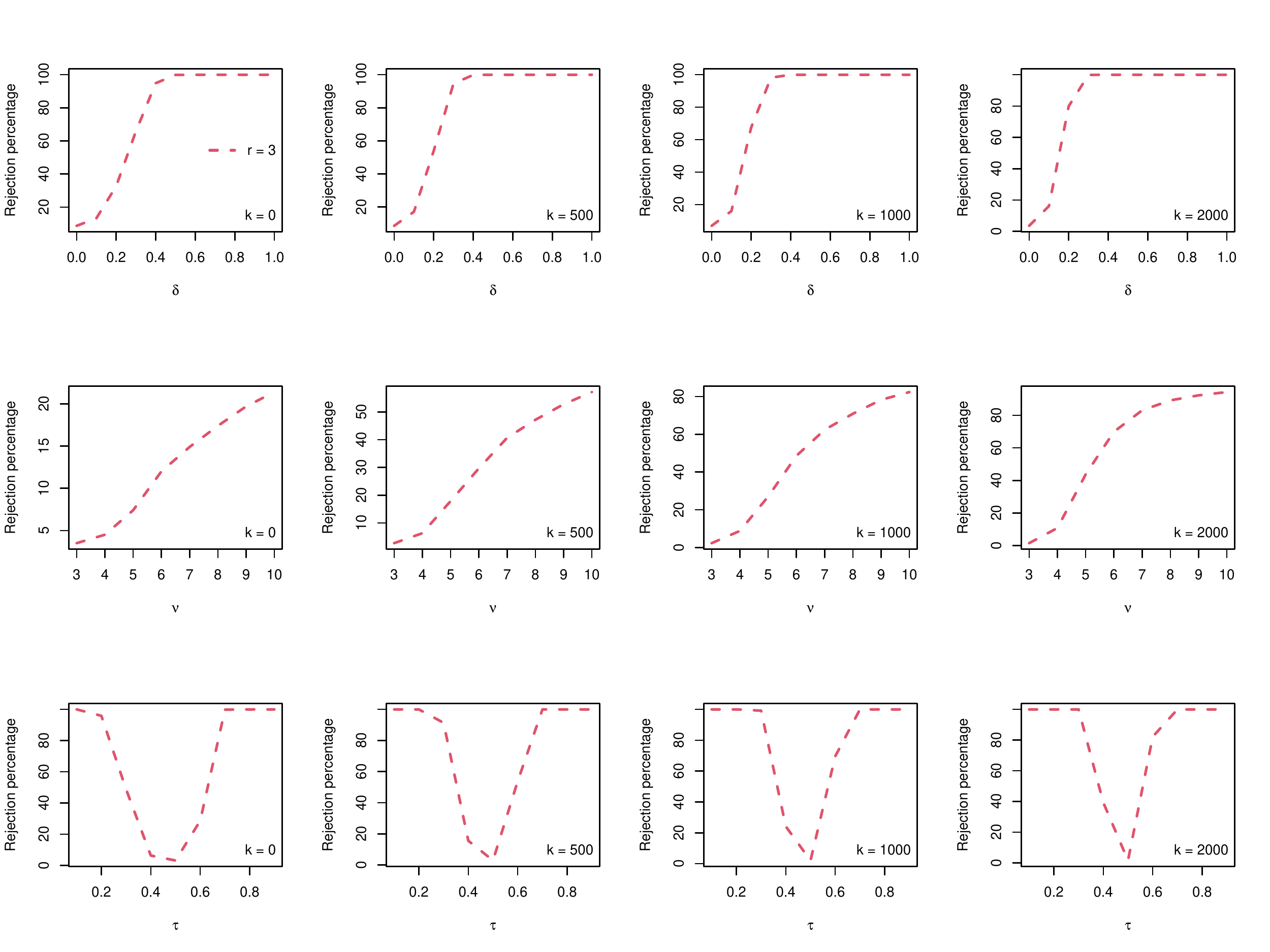}
\caption{\label{fig:H1Tri} Rejection percentages of $H_0$ in~\eqref{eq:H0} for the procedure based on $D_m^{\sPc_m}$ in the trivariate case with $r=3$ and $\kappa = 1.5$. The graphs in the first, second and third row correspond to experiments which are the trivariate analogs of those reported in Figures~\ref{fig:H1BivTs},~\ref{fig:H1BivIid} and~\ref{fig:H1BivTau}, respectively. The corresponding average numbers of selected points are 17.7, 17.3, and 18.3, respectively.}
\end{center}
\end{figure}

\input{H1TriCop.tex}


\end{appendix}

\bibliographystyle{imsart-nameyear}
\bibliography{biblio}

\end{document}

%% file: quantiles.tex
\begin{table}[t!]
\centering
\caption{Columns 2 to 5: for $\eta = 0.001$, $p \in \{2,5,10,20\}$ and $\alpha \in \{0.01, 0.05, 0.1 \}$, estimates $\hat q_{p,\eta}^{\sss{(1 - \alpha)}}$ of the $(1-\alpha)$-quantiles $q_{p,\eta}^{\sss{(1 - \alpha)}}$ of the distribution of $\Lc_{p,\eta}$. Columns 6 to 8: corresponding estimates of the parameters of the transformed asymptotic regression models that can be used to interpolate (resp.\ extrapolate) the value of $\hat q_{p,\eta}^{\sss{(1 - \alpha)}}$ for $p \in [2,20]$ (resp.\ for $p$ slightly larger than 20).} 
\label{tab:quantiles}
\begingroup\normalsize
\begin{tabular}{llllllll}
  \hline
   & \multicolumn{4}{c}{$p$} & \multicolumn{3}{c}{$p \not \in \{2,5,10,20\}$}\\\cline{2-8} $1-\alpha$ & 2 & 5 & 10 & 20 & $\hat \beta_1$ & $\hat \beta_2$ & $\hat \beta_3$ \\ \hline
0.99 & 1.654 & 1.234 & 1.010 & 0.860 & -0.126 & 1.535 & 2.080 \\ 
  0.95 & 1.511 & 1.141 & 0.946 & 0.825 & 0.060 & 1.475 & 1.921 \\ 
  0.90 & 1.450 & 1.099 & 0.921 & 0.806 & 0.140 & 1.462 & 1.870 \\ 
   \hline
\end{tabular}
\endgroup
\end{table}

%% file: H0.tex
\begin{table}[t!]
\centering
\caption{Percentages of rejection of $H_0$ in~\eqref{eq:H0} for the procedure based on $D_m^{\sPc_m}$ considered in Section~\ref{sec:univ:points}. The rejection percentages are computed from 1000 samples of size $n = m + 5000$ generated from the time series models M1, \dots, M9.} 
\label{tab:H0}
\begingroup\small
\begin{tabular}{llllll}
  \hline
  Model & $m$ & $p=2$ & $p=5$ & $p=10$ & $p=20$ \\ \hline
M1 & 200 & 3.4 & 3 & 3.3 & 5.4 \\ 
   & 400 & 1.9 & 2.1 & 3.2 & 2.8 \\ 
   & 800 & 1.5 & 1.3 & 2 & 1.3 \\ 
   & 1600 & 1 & 0.9 & 0.9 & 1 \\ 
  M2 & 200 & 4.4 & 4.3 & 8.1 & 13.8 \\ 
   & 400 & 3.2 & 2.7 & 3.6 & 5.6 \\ 
   & 800 & 1.8 & 2.8 & 2.1 & 3.2 \\ 
   & 1600 & 1.3 & 1 & 1.7 & 1 \\ 
  M3 & 200 & 6.9 & 10.6 & 21.5 & 53.9 \\ 
   & 400 & 3.4 & 4.1 & 7.2 & 16.7 \\ 
   & 800 & 3 & 2.5 & 3 & 5.6 \\ 
   & 1600 & 1.5 & 1.3 & 1.5 & 1.4 \\ 
  M4 & 200 & 9.5 & 18.1 & 44.6 & 93.2 \\ 
   & 400 & 5.3 & 7.1 & 16.5 & 42.9 \\ 
   & 800 & 3 & 3.6 & 5.3 & 12.3 \\ 
   & 1600 & 2.2 & 2.3 & 2.8 & 2.2 \\ 
  M5 & 200 & 18.9 & 39.4 & 82.1 & 100 \\ 
   & 400 & 9.6 & 17.5 & 40.3 & 87 \\ 
   & 800 & 5.7 & 6.7 & 14.4 & 34.5 \\ 
   & 1600 & 2.4 & 2.7 & 4.4 & 7.5 \\ 
  M6 & 200 & 7.3 & 8.6 & 14.1 & 26.7 \\ 
   & 400 & 3.7 & 4.6 & 6.4 & 11 \\ 
   & 800 & 2.6 & 2.8 & 3 & 3.6 \\ 
   & 1600 & 1.2 & 1.5 & 1.3 & 1.2 \\ 
  M7 & 200 & 9.1 & 14.8 & 17.2 & 16.8 \\ 
   & 400 & 8 & 12 & 12.6 & 9.9 \\ 
   & 800 & 7.5 & 9.8 & 9.2 & 5.2 \\ 
   & 1600 & 4.6 & 6.1 & 4.8 & 2.9 \\ 
  M8 & 200 & 7.8 & 13 & 30.6 & 74.3 \\ 
   & 400 & 4.7 & 5.5 & 10.8 & 30.8 \\ 
   & 800 & 2.8 & 3.4 & 4.7 & 7.9 \\ 
   & 1600 & 1.8 & 1.8 & 1.9 & 2.1 \\ 
  M9 & 200 & 15 & 29 & 61.1 & 95.6 \\ 
   & 400 & 8.6 & 12.9 & 25.5 & 57.9 \\ 
   & 800 & 5.5 & 6.7 & 10.1 & 17.3 \\ 
   & 1600 & 2.1 & 2.4 & 3.4 & 5.3 \\ 
   \hline
\end{tabular}
\endgroup
\end{table}

%% file: H0IidTrueLRV.tex
\begin{table}[t!]
\centering
\caption{Percentages of rejection of $H_0$ in~\eqref{eq:H0} for the procedure based on $D_m^{\sPc_m}$ considered in Section~\ref{sec:univ:points} when the true underlying long-run covariance matrix is used instead of its estimate. The rejection percentages are computed from 1000 random samples of size $n = m + 5000$ generated from the normal distribution.} 
\label{tab:H0IidTrueLRV}
\begingroup\small
\begin{tabular}{rrrrr}
  \hline
  $m$ & $p=2$ & $p=5$ & $p=10$ & $p=20$ \\ \hline
200 & 2.4 & 2.4 & 2.3 & 3.1 \\ 
  400 & 1.7 & 1.9 & 3.1 & 2.3 \\ 
  800 & 1.3 & 1.3 & 1.6 & 0.7 \\ 
  1600 & 0.8 & 0.9 & 0.9 & 0.6 \\ 
   \hline
\end{tabular}
\endgroup
\end{table}

%% file: H0IidLargeR.tex
\begin{table}[t!]
\centering
\caption{Percentages of rejection of $H_0$ in~\eqref{eq:H0} for the procedure based on $D_m^{\sPc_m}$ considered in Section~\ref{sec:univ:points} when, for $p > 20$ and $\eta=0.001$, estimates of the 0.95-quantiles of the distribution $\Lc_{p,\eta}$ are extrapolated using the model fitted at the end of Section~\ref{sec:quant}. The rejection percentages are computed from 1000 samples of size $n = m + 5000$ generated from the standard normal distribution with $m=1600$.} 
\label{tab:H0IidLargeR}
\begingroup\small
\begin{tabular}{rrrrrrr}
  \hline
  $p=2$ & $p=5$ & $p=10$ & $p=20$ & $p=30$ & $p=40$ & $p=50$ \\ \hline
1.0 & 0.9 & 0.9 & 1.0 & 0.4 & 0.3 & 0.5 \\ 
   \hline
\end{tabular}
\endgroup
\end{table}

%% file: H0BivAR.tex
\begin{table}[t!]
\centering
\caption{Percentages of rejection of $H_0$ in~\eqref{eq:H0} for the procedure based on $D_m^{\sPc_m}$ considered in Section~\ref{sec:mult:points} with $r \in \{3,4\}$ and $\kappa \in \{1.5,2,3\}$. The rejection percentages are computed from 1000 bivariate samples of size $n = m + 5000$ generated from the time series model~\eqref{eq:ar1} with $\beta=0.3$ and $C$ the bivariate Gumbel--Hougaard copula with a Kendall's tau of $\tau \in \{0,0.3,0.6,0.9\}$. The column $\bar p$ reports the average number of grid points retained by the point selection procedure.} 
\label{tab:H0BivAR}
\begingroup\small
\begin{tabular}{rrrrrrrrrrrrrr}
  \hline
  \multicolumn{2}{c}{} & \multicolumn{4}{c}{$\kappa=1.5$} & \multicolumn{4}{c}{$\kappa=2$} & \multicolumn{4}{c}{$\kappa=3$}  \\ \cmidrule(lr){3-6} \cmidrule(lr){7-10} \cmidrule(lr){11-14} \multicolumn{2}{c}{} & \multicolumn{2}{c}{$r=3$} & \multicolumn{2}{c}{$r=4$} & \multicolumn{2}{c}{$r=3$} & \multicolumn{2}{c}{$r=4$} & \multicolumn{2}{c}{$r=3$} & \multicolumn{2}{c}{$r=4$}  \\ \cmidrule(lr){3-4} \cmidrule(lr){5-6} \cmidrule(lr){7-8} \cmidrule(lr){9-10} \cmidrule(lr){11-12} \cmidrule(lr){13-14} $m$ & $\tau$ & $\bar p$ & $D_m^{\sPc_m}$ & $\bar p$ & $D_m^{\sPc_m}$ & $\bar p$ & $D_m^{\sPc_m}$ & $\bar p$ & $D_m^{\sPc_m}$ & $\bar p$ & $D_m^{\sPc_m}$ & $\bar p$ & $D_m^{\sPc_m}$ \\ \hline
400 & 0.00 & 8.9 & 9.1 & 15.4 & 14.8 & 9.0 & 10.4 & 15.9 & 18.4 & 9.0 & 10.4 & 16.0 & 21.6 \\ 
   & 0.30 & 8.6 & 7.4 & 14.3 & 13.2 & 8.9 & 8.4 & 15.3 & 16.6 & 9.0 & 9.4 & 15.9 & 20.5 \\ 
   & 0.60 & 7.0 & 7.7 & 10.7 & 10.0 & 7.1 & 8.0 & 11.7 & 11.3 & 7.7 & 8.1 & 13.4 & 13.6 \\ 
   & 0.90 & 3.0 & 4.0 & 5.8 & 3.7 & 3.4 & 3.6 & 7.8 & 4.6 & 5.9 & 4.5 & 9.8 & 12.9 \\ 
  800 & 0.00 & 9.0 & 4.5 & 15.9 & 6.6 & 9.0 & 4.7 & 16.0 & 7.6 & 9.0 & 4.7 & 16.0 & 7.6 \\ 
   & 0.30 & 8.7 & 4.6 & 14.5 & 3.7 & 9.0 & 4.9 & 15.6 & 4.9 & 9.0 & 5.0 & 16.0 & 5.6 \\ 
   & 0.60 & 7.0 & 3.6 & 10.3 & 4.4 & 7.0 & 3.7 & 11.9 & 3.5 & 7.8 & 3.5 & 13.8 & 5.9 \\ 
   & 0.90 & 3.0 & 2.1 & 5.2 & 1.8 & 3.2 & 2.2 & 8.4 & 3.1 & 6.3 & 3.0 & 10.0 & 7.5 \\ 
  1600 & 0.00 & 9.0 & 2.2 & 16.0 & 2.7 & 9.0 & 2.2 & 16.0 & 2.7 & 9.0 & 2.2 & 16.0 & 2.7 \\ 
   & 0.30 & 8.9 & 1.8 & 14.7 & 2.8 & 9.0 & 2.2 & 15.8 & 2.6 & 9.0 & 2.2 & 16.0 & 2.7 \\ 
   & 0.60 & 7.0 & 1.2 & 10.1 & 1.1 & 7.0 & 1.2 & 12.0 & 1.1 & 7.6 & 1.3 & 14.0 & 1.8 \\ 
   & 0.90 & 3.0 & 1.2 & 4.9 & 1.3 & 3.0 & 1.3 & 8.8 & 2.0 & 6.5 & 0.9 & 10.0 & 3.3 \\ 
   \hline
\end{tabular}
\endgroup
\end{table}

%% file: H0TriAR.tex
\begin{table}[t!]
\centering
\caption{Percentages of rejection of $H_0$ in~\eqref{eq:H0} for the procedure based on $D_m^{\sPc_m}$ considered in Section~\ref{sec:mult:points} with $r \in \{2,3\}$ and $\kappa \in \{1.5,2,3\}$. The rejection percentages are computed from 1000 trivariate samples of size $n = m + 5000$ generated from the time series model~\eqref{eq:ar1} with $\beta=0.3$ and $C$ the trivariate Clayton copula whose bivariate margins have a Kendall's tau of $\tau \in \{0,0.3,0.6,0.9\}$. The column $\bar p$ reports the average number of grid points retained by the point selection procedure.} 
\label{tab:H0TriAR}
\begingroup\small
\begin{tabular}{rrrrrrrrrrrrrr}
  \hline
  \multicolumn{2}{c}{} & \multicolumn{4}{c}{$\kappa=1.5$} & \multicolumn{4}{c}{$\kappa=2$} & \multicolumn{4}{c}{$\kappa=3$}  \\ \cmidrule(lr){3-6} \cmidrule(lr){7-10} \cmidrule(lr){11-14} \multicolumn{2}{c}{} & \multicolumn{2}{c}{$r=2$} & \multicolumn{2}{c}{$r=3$} & \multicolumn{2}{c}{$r=2$} & \multicolumn{2}{c}{$r=3$} & \multicolumn{2}{c}{$r=2$} & \multicolumn{2}{c}{$r=3$}  \\ \cmidrule(lr){3-4} \cmidrule(lr){5-6} \cmidrule(lr){7-8} \cmidrule(lr){9-10} \cmidrule(lr){11-12} \cmidrule(lr){13-14} $m$ & $\tau$ & $\bar p$ & $D_m^{\sPc_m}$ & $\bar p$ & $D_m^{\sPc_m}$ & $\bar p$ & $D_m^{\sPc_m}$ & $\bar p$ & $D_m^{\sPc_m}$ & $\bar p$ & $D_m^{\sPc_m}$ & $\bar p$ & $D_m^{\sPc_m}$ \\ \hline
400 & 0.00 & 7.6 & 12.5 & 20.8 & 32.0 & 7.9 & 14.6 & 24.1 & 43.6 & 8.0 & 16.0 & 26.0 & 55.6 \\ 
   & 0.30 & 7.4 & 10.4 & 18.3 & 22.4 & 7.8 & 12.2 & 21.4 & 33.7 & 8.0 & 15.4 & 24.0 & 49.2 \\ 
   & 0.60 & 6.5 & 6.5 & 14.6 & 18.4 & 7.5 & 10.9 & 15.4 & 22.3 & 8.0 & 16.5 & 16.5 & 25.4 \\ 
   & 0.90 & 2.0 & 3.9 & 6.4 & 4.4 & 2.0 & 3.8 & 8.5 & 6.7 & 3.2 & 3.4 & 11.2 & 14.9 \\ 
  800 & 0.00 & 7.9 & 5.6 & 24.1 & 11.7 & 8.0 & 6.6 & 26.3 & 18.1 & 8.0 & 6.8 & 26.9 & 22.6 \\ 
   & 0.30 & 7.6 & 4.1 & 19.9 & 9.8 & 8.0 & 5.7 & 23.1 & 14.3 & 8.0 & 6.1 & 25.6 & 19.8 \\ 
   & 0.60 & 6.7 & 2.8 & 15.0 & 7.8 & 7.9 & 5.2 & 15.3 & 9.1 & 8.0 & 6.7 & 16.5 & 8.8 \\ 
   & 0.90 & 2.0 & 2.7 & 6.6 & 3.0 & 2.0 & 2.7 & 9.2 & 4.7 & 2.5 & 1.9 & 12.2 & 6.7 \\ 
  1600 & 0.00 & 8.0 & 2.4 & 26.2 & 3.8 & 8.0 & 2.4 & 27.0 & 5.7 & 8.0 & 2.4 & 27.0 & 6.0 \\ 
   & 0.30 & 7.8 & 1.9 & 20.7 & 2.5 & 8.0 & 2.6 & 24.3 & 4.2 & 8.0 & 2.6 & 26.4 & 6.3 \\ 
   & 0.60 & 6.9 & 1.0 & 15.0 & 2.2 & 8.0 & 1.8 & 15.1 & 2.3 & 8.0 & 2.2 & 16.2 & 2.4 \\ 
   & 0.90 & 2.0 & 1.6 & 6.8 & 1.4 & 2.0 & 1.6 & 9.5 & 2.3 & 2.1 & 1.4 & 13.0 & 4.3 \\ 
   \hline
\end{tabular}
\endgroup
\end{table}

%% file: H1BivCop.tex
\begin{table}[t!]
\centering
\caption{Percentages of rejection of $H_0$ in~\eqref{eq:H0} for the procedure based on $D_m^{\sPc_m}$ considered in Section~\ref{sec:mult:points} with $r \in \{3,4\}$ and $\kappa=1.5$. The rejection percentages are computed from 1000 bivariate samples of size $n = m + 5000$ with $m=800$ such that, up to time $m+k$, observations come from a bivariate Clayton copula with a Kendall's tau of 0.5, while observations after time $m+k$ are generated from a Gumbel--Hougaard with a Kendall's tau of 0.5. The abreviation ``m.d.d.'' stands for ``mean detection delay''. } 
\label{tab:H1BivCop}
\begingroup\small
\begin{tabular}{rrrrrrr}
  \hline
   & \multicolumn{3}{c}{$r=3$} & \multicolumn{3}{c}{$r=4$}  \\ \cmidrule(lr){2-4} \cmidrule(lr){5-7} $k$ & $\bar p$ & $D_m^{\sPc_m}$ & m.d.d. & $\bar p$ & $D_m^{\sPc_m}$ & m.d.d. \\ \hline
0 & 6.8 & 72.3 & 1063.9 & 10.7 & 74.4 & 1007.7 \\ 
  500 & 6.8 & 92.9 & 967.9 & 10.8 & 94.8 & 936.4 \\ 
  1000 & 6.8 & 99.0 & 990.1 & 10.7 & 98.9 & 905.6 \\ 
  2000 & 6.9 & 99.8 & 1002.3 & 10.7 & 99.8 & 965.1 \\ 
   \hline
\end{tabular}
\endgroup
\end{table}

%% file: H1TriCop.tex
\begin{table}[t!]
\centering
\caption{Percentages of rejection of $H_0$ in~\eqref{eq:H0} for the procedure based on $D_m^{\sPc_m}$ considered in Section~\ref{sec:mult:points} with $r=3$ and $\kappa=1.5$. The rejection percentages are computed from 1000 trivariate samples of size $n = m + 5000$ with $m=800$ such that, up to time $m+k$, observations come from a trivariate Clayton copula whose bivariate margins have a Kendall's tau of 0.5, while observations after time $m+k$ are generated from a Gumbel--Hougaard copula whose bivariate margins have a Kendall's tau of 0.5.} 
\label{tab:H1TriCop}
\begingroup\small
\begin{tabular}{rrr}
  \hline
  $k$ & $\bar p$ & $D_m^{\sPc_m}$\\ \hline
0 & 15.3 & 95.2 \\ 
  500 & 15.2 & 99.8 \\ 
  1000 & 15.2 & 100.0 \\ 
  2000 & 15.2 & 100.0 \\ 
   \hline
\end{tabular}
\endgroup
\end{table}